%% file: main.tex
\documentclass[runningheads]{llncs}
\usepackage[T1]{fontenc}
\usepackage{graphicx}
\usepackage{hyperref}
\usepackage{color}

\definecolor{fmlinks}{RGB}{48,92,214}
\hypersetup{
	colorlinks,
	linkcolor=fmlinks,
	citecolor=black,
	urlcolor=fmlinks
}

\usepackage{amsfonts}
\usepackage{amsmath}

\usepackage{cleveref}
\usepackage{thmtools,thm-restate}
\usepackage{mathtools}
\usepackage{cmll}
\usepackage{virginialake}
\vlnostructuressyntax
\usepackage{proofzilla}
\usepackage{xspace}
\usepackage{adjustbox}

\usepackage{comment}
\usepackage{amsfonts}

\usepackage[X2,T1]{fontenc}

\input{macros.tex}

\spnewtheorem{nota}{Notation}{\bfseries}{\itshape}

\begin{document}
%
\title{Formulas as Processes, \\Deadlock-Freedom as Choreographies \\(Extended Version)}
\titlerunning{Formulas as Processes, Deadlock-Freedom as Choreographies}
%
\author{
	Matteo Acclavio\inst{1}\orcidID{0000-0002-0425-2825} 	\and\\
	Giulia Manara\inst{2,3}\orcidID{0009-0003-9583-1017} 	\and\\
	Fabrizio Montesi\inst{3}\orcidID{0000-0003-4666-901X}
}

\authorrunning{
	M. Acclavio et al.
}
%
\institute{
	University of Sussex, Brighton, UK 					\and
	Université Paris Cité, Paris, FR
	and
	Universitá Roma Tre, Roma, IT
	\and
	University of Southern Denmark, Odense, DK
}
\maketitle              

\newcommand{\mcor}[2]{{\color{gray}#1}{\color{pzbrickred}#2}}
\newcommand{\gcor}[2]{{\color{gray}#1}{\color{magenta}#2}}

\newcommand{\mrev}[2]{{\color{gray}#1}{\color{blue}#2}}

\newcommand{\maltIntro}[2]{{\color{gray}#1}{\color{pzgreen}#2}}

\begin{abstract}
	We introduce a novel approach to studying properties of \procs in the \picalc based on a processes-as-formulas interpretation, by establishing a correspondence between specific sequent calculus derivations and computation trees in the reduction semantics of the recursion-free \picalc.
	Our method provides a simple logical characterisation of \lfreedom for the recursion- and race-free fragment of the \picalc, supporting key features such as cyclic dependencies and an independence of the name restriction and parallel operators.
	Based on this technique, we establish a strong completeness result for a nontrivial choreographic language: all deadlock-free and \rfree finite \picalc processes composed in parallel at the top level can be faithfully represented by a choreography.

	With these results, we show how the paradigm of computation-as-derivation extends the reach of logical methods for the study of concurrency, by bridging important gaps between logic, the expressiveness of the \picalc, and the expressiveness of choreographic languages.
\end{abstract}

\section{Introduction}

The Curry-Howard isomorphism is a remarkable example of the synergy between logic and programming languages, which establishes a \emph{formulas-as-types} and \emph{proofs-as-programs} {(and \emph{computation-as-reduction})} correspondences for functional programs~\cite{lof:constructive,notes:CH}.
In view of this success, an analogous \emph{proofs-as-processes} correspondence has been included in the agenda of the study of concurrent programming languages \cite{abramsky1994proofs,bellin:scot:pi}.
The main idea in this research line is that,
as types provide a high-level specification of the input/output data types for a function, propositions in linear logic correspond to \emph{session types}~\cite{hon:vas:kub} that specify the communication actions performed by processes (as in process calculi)~\cite{cai:pfe:ST}.
However, while the Curry-Howard correspondence fits functional programming naturally, it does not come without issues when applied to concurrency (we discuss the details in related work, \cref{sec:rel}).
This is because functional programming deals with the `sequential' aspect of computation (given an input, return a specific output), while most of the interesting aspects of concurrent computation are about the communication patterns used during the computation itself.
Therefore, as suggested in numerous works (e.g., \cite{miller:hereditary,miller:uniform,andreoli:focSeq,andreoli:focSeq,and:maz:PN}), the Curry-Howard correspondence may not be the right lens for the study of concurrent programs.
	
In this paper we investigate an alternative proof-theoretical approach to the study of processes, which is based on logical operators that faithfully model the fundamental operators of process calculi (like prefixing, parallel, and restriction).
Our approach is close to the \emph{computation-as-deduction} paradigm, where program executions are modelled as proof searches in a given sequent calculus; executions are therefore grounded in a logic by construction, as originally proposed by Miller in \cite{miller:hereditary}.
Specifically, we interpret formulas as processes, inference rules as rules of an operational semantics, and (possibly partial) derivations as \comptrees (snapshots of computations up to a certain point).

The approach we follow is not to be confused with the
intent of using logic as an auxiliary language to enunciate statements about computations, that is, viewing \emph{computation-as-model} as done in Hennessy-Milner logic \cite{hen:mil:80,hen:mil:90}, modal $\mu$-calculus \cite{kozen:modalMu}, Hoare logic \cite{hoare:logic}, or dynamic logics \cite{DLbook,acc:mon:per:OPDL}.
We are instead interested in directly reasoning on programs and their execution using the language of the programs itself.
This allows for an immediate transfer of properties of proofs to properties of programs, without needing intermediate structures (e.g., models) or languages (e.g., types).

\subsection{Contributions of the paper}\label{subsec:contributions}

We consider the recursion-free fragment of the \picalc{} -- as presented in \cite{CDM14,gay:hole} -- and embed it in the language of the system $\PIL$ from \cite{acc:man:NL}. $\PIL$ extends Girard's first order multiplicative and additive linear logic \cite{gir:ll} with a non-commutative and non-associative sequentiality connective ($\lprec$), and nominal quantifiers ($\lnewsymb$ and its dual $\lyasymb$) for variable scoping.

Using this embedding, we prove the following main results.
\begin{enumerate}
	\item\label{contribution:1}
	We show that the operational semantics of the \picalc is captured by the linear implication ($\limp$) in $\PIL$: if we denote by $\fof P$ the formula encoding the process $P$, then
	\begin{itemize}
		\item if $P$ is a process reducing to $P'$ by performing a communication or an external choice, then $\proves[\PIL] \fof {P'} \limp \fof P$; while
		\item if $P$ may reduce to $P_{\lab_1},\ldots, P_{\lab_n}$ by performing an internal choice, then $\proves[\PIL] \bigwith_{i=1}^n\left(\fof{ P_{\lab_i}}\right) \limp \fof P$.
	\end{itemize}
	Crucial to prove this result is our proof that the system $\PIL$ supports a substitution principle, which allows us to simulate reductions within a context.

	\item\label{contribution:2}
	We establish a computation-as-deduction correspondence, which we use to characterise two key safety properties studied for \rfree processes in terms of derivability in $\NL$:
	\emph{\lfreedom}, i.e., the property that a process can always keep executing until it eventually terminates~\cite{wadler:PaS}; and
	\emph{\progress}, i.e., the property that if a process gets stuck, it is always because of a missing interaction with an action that can be provided by the environment~\cite{DBLP:journals/mscs/CoppoDYP16}.%
	\footnote{For \progress, {in this paper} we restrict our attention to processes that do not send restricted names.}

	In particular, thanks to the structure of $\PIL$ and its operators, we can successfully detect safe processes that were previously problematic in the logical setting due to cyclic dependencies, like the \proc in \Cref{eq:intro} below.
	\begin{equation}\label{eq:intro}
		\hskip-2em
		\pnu x\pnu y
		\left(
			\psend xa
			\plsend y{\lab: \psend yb \pnil}{}
		\ppar
			\precv xa
			\plrecv y {\lab: \precv yb \pnil , \lab': \psend zc \pnil}{}
		\right)
	\end{equation}

	\item\label{contribution:3}
	We show that our approach provides an adequate logical foundation for choreographic programming, a paradigm where programs are \emph{choreographies} (coordination plans) that express the communications that a network of processes should enact~\cite{M13:phd}.%
	\footnote{
		Networks are parallel compositions of sequential processes assigned to distinct locations.
	}
	Specifically, we establish a \emph{choreographies-as-proofs} correspondence by using a sequent system, called $\ChorL$, that consists of rules derivable in $\PIL$.

	Our choreographies-as-proofs correspondence has an important consequence. An open question in theory of choreographic programming~\cite{montesi:book} is about how expressive choreographies can be:
	\begin{equation}
		\mbox{
		\emph{What are the processes that can be captured by choreographies?}
		}
	\end{equation}
	To date, there are no answers to this question for the setting of processes with unrestricted name mobility and cyclic dependencies.
	Our correspondence implies a strong completeness result: in our setting, all and only race- and deadlock-free networks can be expressed as choreographies.
	This is the first such completeness result in the case of recursion-free networks.
\end{enumerate}

\begin{figure}[t]
	\centering
	\adjustbox{max width=\textwidth}{$
		\hskip-1em
	\begin{array}{ccccc}
		P \mbox{ \rfree \lfree} & \xRightarrow{\qquad}&P \mbox{ {\rfree} {\lfree} {\flat} }
		\\
		\qquad\qquad\displaystyle\left\Updownarrow\vphantom{\int}\right.{\mbox{\Cref{thm:deadlock}}}
		&&
		\qquad\qquad\displaystyle\left\Updownarrow\vphantom{\int}\right.{\mbox{\begin{tabular}{c}\Cref{cor:flatChor} \\of \Cref{thm:deadlockfreeChor}\end{tabular}}}
		\\[-1.2cm]
		\proves[\PIL] \fof P
		&\xRightarrow{\Cref{lem:ChorLderiv}}&
		\proves[\ChorL] \fof P
		&
		\xLeftrightarrow{\mbox{\Cref{thm:proofsASchoreo}}}
		&
		\boxed{\mbox{\raisebox{1cm}{\begin{tabular}{c}Exists a\\\chor $\chC$ \\ such that \\$\EPP \chC \steq P$\\[-20pt]\end{tabular}}}}
	\end{array}$}
	\caption{Road map of the main technical results in this work.}
	\label{fig:thm_dependencies}
\end{figure}

\subsection{Structure of the paper}
In \Cref{sec:non} we report $\PIL$ and prove the additional technical results required for our development.
In \Cref{sec:FaP} we recall the syntax and semantics of the \picalc, introduce an alternative \opsem with the same expressiveness with respect to the property of \lfreedom, show that the \opsem of the \picalc is captured by linear implication in $\PIL$, and that each step of the \opsem of the \picalc can be seen as blocks of rules in this system.
In \Cref{sec:choreo} we define our choreographic language and provide our completeness result for \chors.
We discuss related work in \Cref{sec:rel}.
We conclude in \Cref{sec:conc}, where we discuss research directions opened by this work.
Due to space constraints, details of certain proofs are provided in the Appendix.

\section{Non-commutative logic}\label{sec:non}

In this section we recall $\PIL$ and some of its established properties~\cite{acc:man:NL}. We then prove additional results required for our development.

\subsection{The system $\PIL$}\label{sec:PiL-language}

\begin{figure}[t]
	\centering
	\adjustbox{max width=.9\textwidth}{$
		\begin{array}{c|c|c}
			\mbox{Formulas}
			&
			\mbox{De Morgan Laws}
			&
			\mbox{$\alpha$-equivalence}
		\\
			\begin{array}{l@{\;}c@{\;}llll}
				A,B &\coloneqq	&
				\lunit
				&\mbox{unit (atom)}			\\&\mid&
				\lsend xy
				&\mbox{atom}				\\&\mid&
				\lrecv xy
				&\mbox{atom}				\\&\mid&
				A\lpar B
				&\mbox{par}					\\&\mid&
				A\ltens B
				&\mbox{tensor}				\\&\mid&
				A\lprec B
				&\mbox{prec}				\\&\mid&
				A\lplus B
				&\mbox{oplus}				\\&\mid&
				A\lwith B
				&\mbox{with}				\\&\mid&
				\lFa xA
				&\mbox{for all}				\\&\mid&
				\lEx xA
				&\mbox{exists}				\\&\mid&
				\lNu xA
				&\mbox{new}					\\&\mid&
				\lYa xA
				&\mbox{ya}
			\end{array}
		&
			\begin{array}{ccc}
				\cneg \lunit		&=&	\lunit
				\\
				\cneg{(\cneg A)}	&=&	A
				\\
				\cneg{\lsend xy}	&=& \lrecv xy
				\\
				\cneg{(A \lpar B)}	&=& \cneg A \ltens \cneg B
				\\
				\cneg{(A \lprec B)}	&=& \cneg A \lprec \cneg B
				\\
				\cneg{(A \lplus B)}	&=& \cneg A \lwith \cneg B
				\\
				\cneg{(\forall x.A)}&=& \lEx x{\cneg A}
				\\
				\cneg{(\lNu \xX A)}	&=& \lYa \xX{\cneg A}
			\end{array}
		&
			\begin{array}{c}
				a 	 = 	a
				\\
				\mbox{if $a \in \set{\lunit, \lsend xy , \lrecv xy}$}
			\\\\
				A_1 \circleddot A_2 	= B_1 \circleddot B_2
				\\
				\mbox{if $A_i = B_i$}
				\\
				\mbox{and $\circleddot \in \set{\lpar,\lprec,\ltens,\lplus,\lwith}$}
			\\\\
				\lQu x A 		= 	\lQu y A \fsubst yx
				\\
				\mbox{$y$ fresh for $A$}
				\mbox{and $\lqusymb \in \set{\lnewsymb{}, \lyasymb{}, \forall, \exists}$}
			\end{array}
		\end{array}
	$}
	\caption{Formulas (with $x,y\in \varset$), and their syntactic equivalences.}
	\label{fig:Seq}
\end{figure}


The language of $\PIL$ is a first-order language containing the following.
\begin{itemize}
	\item
	Atoms generated by (i) a countable set of variables $\varset$, (ii) a binary predicate on symbols $\lsend{-}{-}$, and (iii) its dual binary predicate $\lrecv{-}{-}$;

	\item
	The multiplicative connectives for disjunction ($\lpar$) and conjunction ($\ltens$)
	and
	the additive connectives for disjunction ($\lplus$) and conjunction ($\lwith$) form \emph{multiplicative additive linear logic} \cite{gir:ll}.
	We generalize the standard binary $\lplus$ and $\lwith$, allowing them to have any positive arity (including $1$) in order to avoid dealing with associativity and commutativity when modelling choices;

	\item
	A binary non-commutative and non-associative self-dual multiplicative connective \defn{precede} ($\lprec$), whose properties reflect the ones of the prefix operator used in standard process calculi (see \CCS \cite{M80} and \picalc \cite{mil:par:wal:pi,gay:hole});

	\item
	A unit ($\unit$). We observe that its properties reflect the ones of the terminated process $\pnil$.
	Notably, it is the neutral element of the connectives that we will use to represent parallelism ($\lpar$) and sequentiality ($\lprec$), and it is derivable from no assumption (thus also neutral element of $\ltens$);

	\item
	The standard first order existential ($\exists$) and universal ($\forall$) quantifiers.

	\item
	A \emph{nominal quantifier} \defn{new} ($\lnewsymb$), and its dual \defn{ya} ($\lyasymb$). We observe that these restrict variable scope in formulas in the same way that the $\nu$ constructor in process calculi restricts the scope of names.
\end{itemize}

More precisely, we consider \defn{formulas} generated by grammar in \Cref{fig:Seq}
modulo the standard \defn{$\alpha$-equivalence} from the same figure.
{
	From now on, we assume formulas and sequents to be \defn{clean}, that is, such that each variable $x\in\varset$ occurring in them can be bound by at most a unique universal quantifier or at most a pair of dual nominal quantifiers, and, if bound, it cannot occur free .%
	\footnote{
		This can be considered as a variation of \emph{Barendregt’s convention}.
		It allows us to avoid variable renaming for universal and nominal quantifier rules in derivations, by assuming the bound variable to be the eigenvariable of the quantifier or the shared fresh name in the case of a pair of dual nominal quantifiers.
	}
}
The \defn{(linear) implication} $A\limp B$ \resp{the \defn{logical equivalence} $A\feq B$} is defined as $\cneg A \lpar B$ \resp{as $(A \limp B)\ltens (B\limp A)$}, where the \defn{negation} ($\cneg\cdot$) is defined by the \defn{de Morgan duality} in \Cref{fig:Seq}.

The set $\freeof A$ \resp{$\freeof\Gamma$} of \defn{free variables} of a formula $A$ \resp{of a sequent $\Gamma=A_1,\ldots, A_n$} is the set of atoms occurring in $A$  which are not bound by any quantifier \resp{the set $\bigcap_{i=1}^{n}\freeof{A_i}$}.
A \defn{context} is a formula containing a single occurrence of a propositional variable $\chole$ (called \defn{hole}).
We denote by $\cC\ctx[A]\coloneqq\cC\fsubst A\chole$.
A \defn{\nuparcont} is a context $\cK \ctx$ of the form $\cK \ctx = \lNu{x_1}{\dots\lNu{x_n}{(\ctx \lpar A)}}$ for a $n\in\N$.


In this work we assume the reader to be familiar with the syntax of sequent calculus (see, e.g., \cite{troelstra_schwichtenberg_2000}), but we recall here the main definitions.

\begin{nota}
	A \defn{sequent} is a set of occurrences of formulas.
	\footnote{
		In a set of occurrences of formulas, it is assumed that each formula has a unique identifier, differently from a multiset of formulas where each formula has a multiplicity.
		The former definition simplifies the process of tracing occurrences of formulas in a derivation, as we need in \Cref{sec:choreo}.
	}%
	A \defn{sequent rule} $\rrule$ with \defn{premise} sequents $\Gamma_1,\ldots,\Gamma_n$ and \defn{conclusion} $\Gamma$ is an expression of the form $\vliiinf{\rrule}{}{\Gamma}{\Gamma_1}{\cdots}{\Gamma_n}$~.
	A formula occurring in the conclusion \resp{in a premise} of a rule but in none of its premises \resp{not in its conclusion} is said \defn{principal} \resp{\defn{active}}.
	Given a set of rules $\XS$, a \defn{derivation} in $\XS$ is a non-empty tree $\dD$ of sequents, whose root is called \defn{conclusion},
	such that every sequent occurring in $\dD$ is the conclusion of a rule in $\XS$, whose children are (all and only) the premises of the rule.
	An \defn{open derivation} is a derivation whose leaves may be conclusions of no rules, in which case are called \defn{open premises}.
	We may denote a derivation \resp{an open derivation with a single open premise $\Delta$}
	with conclusion $\Gamma$ by $\vldownsmash{\vlderivation{\vlpr{\dD}{}{\vdash\Gamma}}}$
	$\left(\mbox{resp. }\vldownsmash{\vlderivation{\vlde{\dD}{}{\vdash\Gamma}{\vlhy{\vdash\Delta}}}}\right)$.
	Finally we may write $\vliiiqf{\rrule}{}{\vdash \Gamma}{\vdash \Gamma_1}{\cdots}{\vdash \Gamma_n}$ if there is an open derivation with premises $\Gamma_1,\ldots,\Gamma_n$ and conclusion $\Gamma$ made only of rules $\rrule$.
\end{nota}

A \defn{nominal variable} is an element of the form $\isna x$ with $x\in\varSet$ and $\nabla\in\set{\isnusymb,\isyasymb}$.
If $\sS$ is a set of nominal variables, we say that $x$ \defn{occurs} in $\sS$ if $\isnu x$ or $\isya x$ is an element of $\sS$.
A  \defn{(nominal) store} $\sS$ is a set of nominal variables such that each variable occurs at most once in $\sS$.
A \defn{judgement} $\sdash\Gamma$ is a pair consisting of a clean sequent $\Gamma$
and a store $\sS$.
We write judgements $\sdash\Gamma$ with
$\sS=\emptyset$
\resp{$\sS=\set{ \isna[1]{x_1},\ldots,\isna[n]{x_n}}$}
simply as
$\vdash \Gamma$ \resp{$\sdash[ {\isna[1]{x_1},\ldots,\isna[n]{x_n} }]\Gamma$, i.e. omitting parenthesis}.
We write $\sS_1,\sS_2$ to denote the union of two stores such that a same variable does not occur in both $\sS_1$ and $\sS_2$ {-- i.e., a disjoint union}.

The system $\PIsys$ is defined by the all rules in \Cref{fig:rules} except the rule $\cutr$.
We write $\proves[\PIsys]\Gamma$ to denote that the judgement $\sdash[\emptyset]\Gamma$ is derivable in $\PIsys$.

\begin{remark}
	The system $\MLL\wF=\set{\axrule,\lpar,\ltens,\exists,\forall}$ is the standard one for first order multiplicative linear logic \cite{gir:ll}. The rules $\lplus$ and $\lwith$ are generalisations of the standard ones in additive linear logic for the $n$-ary generalized connectives we consider here; thus, in proof search, the rule $\lplus$ keeps only one $A_k$ among all $A_i$ occurring in $\multioplus{i=1}{n} A_i$, and the rule $\lwith$ branches the proof search in $n$ premises. %
{%
	The rule $\mixr$ and $\unit$ are standard for multiplicative linear logic with mix in presence of units\footnote{In presence of mix the two multiplicative units collapse.}
	\cite{cricri:MIX,COCKETT1997133}, and the rule $\precur$ ensures that the unit $\lunit$ is not only the unit for the connectives $\lpar$ and $\ltens$, but also for $\lprec$.
	The rule $\lprec$ is required to capture the self-duality of the connective $\lprec$; it should be read as introducing at the same time the connective $\lprec$ and its dual (which in this case is $\lprec$ itself) -- as a general underlying pattern for multiplicative connectives, see \cite[Remark 5]{acc:IJCAR24}.

	The store is used to guarantee that each $\lnewsymb$ is linked to at most a unique $\lyasymb$ (or vice versa) in any branch of a derivation.
	If a rule $\nuur$ \resp{$\yaur$} is applied, then the nominal quantifier is not linked, reason why the rule reminds the standard universal quantifier rule.
	Otherwise, either the rule $\nuloadr$ \resp{$\yaloadr$} loads a nominal variable in the store, or a rule $\nupopr$ \resp{$\yapopr$} uses a nominal variable (of dual type) occurring in the store as a witness variable.
	Note that in a derivation with the conclusion a judgement with empty store any
	$\nupopr$ \resp{$\yapopr$} is uniquely linked to a $\nuloadr$ \resp{$\yaloadr$} below it.
}
\end{remark}

\begin{figure}[t]
	\centering
	\adjustbox{max width=\textwidth}{$\begin{array}{c}
		\vlinf{\axrule}{}{\sdash \lsend xy, \lrecv xy}{}
		\qquad
		\vlinf{\lpar}{}{\sdash \Gamma, A\lpar B}{\sdash \Gamma, A, B}
		\qquad
		\vliinf{\ltens}{}{\sdash[\sS_1 , \sS_2] \Gamma, A\ltens B,\Delta}{\sdash[\sS_1] \Gamma, A}{\sdash[\sS_2] B, \Delta}
		\qquad
		\vlinf{\lunit}{}{\sdash \lunit}{}
		\qquad
		\vliinf{\mixr}{}{\sdash[\sS_1 , \sS_2] \Gamma, \Delta}{\sdash[\sS_1] \Gamma}{\sdash[\sS_2] \Delta}
	\\[15pt]
		\begin{array}{c}
			\vlinf{\lplus}{
			}{\sdash\Gamma, \bigoplus_{i=1}^n A_i}{ \sdash \Gamma, A_k}
			\\
			\text{for a $k\in\intset1n$}
		\end{array}
		\qquad
		\vliiinf{\lwith}{}{\sdash \Gamma,\bigwith_{i=1}^n A_i }{\sdash \Gamma,A_1}{\cdots}{\sdash \Gamma, A_n}
		\qquad
		\vlinf{\forall}{\dagger}{\sdash \Gamma, \lFa x{A}}{\sdash \Gamma, A}
		\qquad
		\vlinf{\exists}{}{\sdash \Gamma, \lEx x{A}}{\sdash \Gamma, A\fsubst yx}
	\\[15pt]
		\vliinf{\lprec}{}{
			\sdash[\sS_1 , \sS_2] \Gamma,\Delta, A\lprec B ,C \lprec D
		}{ \sdash[\sS_1] \Gamma, A,C}{\sdash[\sS_2] \Delta, B,D}
		\qquad
		\vliinf{\precur}{}{
			\sdash[\sS_1 , \sS_2] \Gamma,\Delta, A\lprec B
		}{ \sdash[\sS_1] \Gamma, A}{\sdash[\sS_2] \Delta, B}
		\qquad
		\boxed{\vliinf{\cutr}{}{\sdash[\sS_1,\sS_2] \Gamma, \Delta}{\sdash[\sS_1] \Gamma, A}{\sdash[\sS_2] \cneg A, \Delta}}
	\\[15pt]
		\vlinf{\nuur}{\dagger}{\sdash \Gamma, \lNu x A}{\sdash[\sS] \Gamma, A}
		\qquad
		\vlinf{\nuloadr}{\dagger}{\sdash \Gamma, \lNu xA}{\sdash[\sS , \isnu x] \Gamma , A}
		\qquad
		\vlinf{\nupopr}{}{\sdash[\sS , \isnu y] \Gamma, \lYa xA}{\sdash \Gamma,A\fsubst yx}
	\\[15pt]
		\vlinf{\yaur}{\dagger}{\sdash \Gamma, \lYa x A}{\sdash[\sS] \Gamma, A}
		\qquad
		\vlinf{\yaloadr}{\dagger}{\sdash \Gamma,\lYa xA}{\sdash[\sS , \isya x] \Gamma , A}
		\qquad
		\vlinf{\yapopr}{}{\sdash[\sS , \isya y] \Gamma, \lNu xA}{\sdash \Gamma, A\fsubst yx}
	\end{array}$}
	\caption{
		Sequent calculus rules, with $\dagger\coloneqq x\notin \freeof{\Gamma}$.
	}
	\label{fig:rules}
\end{figure}

\subsection{Proof Theoretical Properties of $\PIL$}\label{sec:PT}

We now recall some basic proof-theoretical properties of the system $\PIL$ and then prove additional results (\Cref{thm:der}) that will be important for the main technical results in this paper.

In $\PIL$ we can prove that atomic axioms are sufficient to prove that the implication $A\limp A$ holds for any formula $A$.
Moreover, the $\cutr$-rule is admissible in these systems, allowing us to conclude the transitivity of the linear implication, as well as the sub-formula property for all the rules of the systems.
\begin{theorem}[\cite{acc:man:NL}]\label{thm:cutelim}
	Let $\Gamma$ be a non-empty sequent in $\PIsys$.
	Then
	\begin{enumerate}
		\item\label{cutelim:1}
		$\proves[\PIL] \cneg A, A$ for any formula $A$;

		\item\label{cutelim:2}
		if $\proves[\PIL\wcut]\Gamma$, then $\proves[\PIL]\Gamma$;

		\item\label{cutelim:3}
		if $\proves[\PIL]A\limp B$ and $\proves[\PIL]B\limp C$, then $\proves[\PIL]A\limp C$.
	\end{enumerate}
\end{theorem}

\begin{proposition}[\cite{acc:man:NL}]\label{prop:feq}
	The following logical equivalences are derivable in $\PIL$ (for any $\sigma$ permutation over $\intset1n$).
	\begin{equation}\label{eq:logEqs}
		\adjustbox{max width=.8\textwidth}{$\begin{array}{c|c}
				\begin{array}{rcl}
					\left(A\lpar \unit\right) 			&\feq& A
					\\
					\left(A\lprec\lunit \right) 		&\feq& A
					\\
					(A\lpar B) \lpar C		&\feq& A\lpar (B \lpar C	)
					\\
					A\lpar B				&\feq& B\lpar A
					\\
					\left(\bigoplus_{i=1}^n A_i\right) &\feq& \left(\bigoplus_{i=1}^n A_{\sigma(i)}\right)
					\\
					\left(\bigwith_{i=1}^n A_i\right) &\feq& \left(\bigwith_{i=1}^n A_{\sigma(i)}\right)
				\end{array}
				\quad&\quad
				\begin{array}{rcl}
					(\lNu x \lNu yA )	& \feq& (\lNu y \lNu xA)
					\\
					(\bigwith_{i=1}^n \lNu xA_i )	& \feq& (\lNu x \bigwith_{i=1}^n A_i)
					\\
					(\bigoplus_{i=1}^n \lNu A_i ) 	&\feq&  (\lNu x \bigoplus_{i=1}^n A_i)
					\\[5pt]\hline\\[-10pt]
					\lNu x{(A\lpar D)}		&\feq&  (\lNu x A) \lpar D
					\\
					\lNu x D &\feq& D
					\\
					\multicolumn{3}{c}{\mbox{if $x \notin \freeof D$}}
				\end{array}
			\end{array}$}
	\end{equation}
	Moreover, $\proves[\NL]\left(\bigwith_{i=1}^n(A_i \lpar B)\right) \limp \left((\bigwith_{i=1}^n A_i) \lpar B\right)$.
\end{proposition}


In addition to these properties, our development requires some new ones showing that implication is preserved in different contexts.
This is necessary because some rules in the operational semantics of the \picalc enables rewriting of deeply nested subterms.
For this reason, we are required to establish properties similar to those for proving \emph{subject reduction} in $\lambda$-calculus.
That is, we prove that in $\PIL$ we can still reproduce the application of inference rules inside contexts preserving soundness and completeness.
These necessary properties are collected in the next theorem.
%

\begin{restatable}{theorem}{contextDers}\label{thm:der}
	For any context $\cC\ctx$ and \nuparcont $\cK\ctx$ we have:
	\begin{enumerate}
		\item\label{der:2} if $\proves[\PIL] A\limp B$, then $\proves[\PIL] \cC\ctx[A]\limp \cC\ctx[B]$;

		\item\label{der:3} if
		$ \proves[\PIL]A_i \limp B$ for $i\in\intset1n$,
		then
		$\proves[\PIL] \bigwith_{i=1}^n \cK\Ctx[A_i]\limp \cK\ctx[B] $;

		\item\label{cor:deepImp} if
		$ \proves[\PIL] A \limp A'$ and
		$ \proves[\PIL] \cC\ctx[A']\limp B$,
		then
		$\proves[\PIL] \cC\ctx[A]\limp B$.
	\end{enumerate}
\end{restatable}
\begin{proof}
	\Cref{der:2} and \Cref{der:3} are proven by induction on the structure of the contexts.
%
	To prove \Cref{cor:deepImp} we use \Cref{thm:cutelim}.\ref{cutelim:3} since if $ \proves[\PIL] A \limp A'$, then by \ref{der:2} also $ \proves[\PIL] \cC\ctx[A] \limp \cC\ctx[A']$.
	See the full proof in \Cref{app:proofs}.
\end{proof}

\begin{remark}
	With the same motivation, one might argue that a more natural way to establish a correspondence between proof search and execution of processes could be \emph{deep inference}~\cite{gug:SIS,gug:gun:par:10}, since that would allow for applying rules to subformulas.
	Here we chose to follow the more standard approach of sequent calculi, where rules can manipulate only top-level term-constructors (connectives).
\end{remark}

\begin{remark}\label{rem:cutelim}
	In the proofs-as-processes interpretation, $\cutr$ is the linchpin that triggers the rewriting simulating the \opsem (cut elimination).
	In our work, as in general in the study of processes-as-formulas, the $\cutr$-rule is freed from being the keystone of the system.
	Instead, the admissibility of the $\cutr$-rule in the computation-as-deduction approach guarantees the existence of canonical models~\cite{miller:pi,miller:92,hod:mil:94:LP}.
	In particular, we use this property to tame the syntactic bureaucracy of the \opsem due to rules $\rsparr$, $\rsresr$, and $\rsstreqr$, as well as to ensure the transitivity of logical implication (required to compose reduction steps).
\end{remark}

\section{Embedding the \picalc in $\PIL$}\label{sec:FaP}

In this section we provide an interpretation of \picalc processes as formulas in $\PIL$, showing also that each successful execution of a process corresponds to a branch in a correct derivation in $\PIL$.

We start by recalling the definition of the \picalc and its operational semantics. Our presentation has explicit primitives for communicating choices, as usual in the literature of session types~\cite{vasco:pi,wadler:PaS,kok:mon:per:better}.
We then present an alternative semantics in which we use structural precongruence instead of structural equivalence (this is a standard simplification~\cite{CM20,koba:lock}, which does not affect reasoning about deadlock-freedom, progress, or races).
We then provide a translation of \procs $P$ into formula $\fof P$ in $\PIL$ and characterise \dfreedom for $P$ in terms of provability of $\fof P$ in $\PIL$.

\subsection{The \picalc and its reduction semantics}

\begin{figure}[t]
	\centering
	\adjustbox{max width=\textwidth}{$
		\begin{array}{l@{\;}c@{\;}l@{\;}|@{\;}l@{\;}|@{\;}l@{\;}|@{\;}l}
			\multicolumn{4}{c|@{\;}}{\mbox{Processes}}&\mbox{Free names}&\mbox{bound names}\\
			P,Q,R
			&\coloneqq	&
			\pnil  			&\mbox{nil}						&\emptyset	&\emptyset
			\\&|&
			\psend xy P 	&\mbox{send ($y$ on $x$)}		&\freenamesof P\setminus \set{x,y}&\boundnamesof P
			\\&|&
			\precv xy P	&\mbox{receive ($y$ on $x$)}	&\freenamesof P\setminus\set x & \boundnamesof P\cup \set y
			\\&|&
			P \ppar Q		&\mbox{parallel}				&\freenamesof P\cup \freenamesof Q&\boundnamesof P\cup \boundnamesof Q
			\\&|&
			\pnu x P		&\mbox{nu}						&\freenamesof P\setminus\set{x}&\boundnamesof P\cup\set{x}
			\\&|&
			\plsend x{\lab : P_\lab}{\lab\in L}
			&\mbox{label send (on $x$)}				&\bigcup_{\lab \in L}\freenamesof {P_\lab}&\bigcup_{\lab \in L}\boundnamesof {P_\lab}
			\\&|&
			\plrecv x{\lab : P_\lab}{\lab\in L}
			&
			\mbox{label receive (on $x$)}
			&
			\bigcup_{\lab \in L}\freenamesof {P_\lab}&\bigcup_{\lab \in L}\boundnamesof {P_\lab}
		\end{array}
		$}
	\caption{Syntax for \picalc \procs with $x,y\in\namesset$ and $L \subset\labelsset$, and their sets of free and bound names.}
	\label{fig:terms}
\end{figure}

\begin{figure}[t]
	\adjustbox{max width=\textwidth}{$\begin{array}{c|c}
			\mbox{$\alpha$-equivalence}
			&
			\mbox{Structural equivalence generators}
			\\
			\begin{array}{r@{\;}c@{\;}ll}
				\pnil & \alphaeq & \pnil & \\
				\precv xy P &			\alphaeq 		& \precv xz P \fsubst zy
				&
				\mbox{$z$ fresh for $P$}
				\\
				\psend xy P &			\alphaeq		& \psend xy Q
				& \mbox{if $P \alphaeq Q$}
				\\
				P \ppar Q 		&		\alphaeq									&
				R \ppar S &
				\mbox{if $P \alphaeq R$ and $Q \alphaeq S$}
				\\
				\pnu x P	 &			\alphaeq 		& \pnu u P \fsubst ux
				&
				\mbox{$u$ fresh for $P$}
				\\
				\plsend x{\lab:P_\lab}{\lab\in L} 		&	\alphaeq		&
				\plsend x{\lab:Q_\lab}{\lab\in L}						&
				\mbox{if $P_{\lab} \alphaeq Q_{\lab}$ for all $\lab \in L$}
				\\
				\plrecv x{\lab:P_\lab}{\lab\in L}			& 	\alphaeq		&
				\plrecv x{\lab:Q_\lab}{\lab\in L} &
				\mbox{if $P_{\lab} \alphaeq Q_{\lab}$ for all $\lab \in L$}
				\\
			\end{array}
			&
			\begin{array}{rcl}
				P \ppar Q 			&\steqr& Q \ppar P
				\\
				(P \ppar Q) \ppar R &\steqr& P \ppar (Q \ppar R)
				\\
				\pnu x \pnu y P 	&\steqr&  \pnu y \pnu x P
				\\\hline
				P \ppar \pnil 		&\strew& P
				\\\hline\hline
				\pnu x S	 		&\strew& S
				\\
				\pnu x P \ppar S	&\strew& \pnu x (P \ppar S)
				\\
				\multicolumn{3}{c}{\mbox{ with $x\notin\freeof S$}}
			\end{array}
		\end{array}$}
	\caption{
		The standard $\alpha$-equivalence,
		and relations generating of the structural equivalence ($\steq$) \picalc \procs, where $A\steqr B$ stands for $A\strew B$ and $B\strew A$.
		}
	\label{fig:equivs}
\end{figure}

The set of \picalc{} \defn{\procs} is generated by the grammar in \Cref{fig:terms}, which uses a fixed countable set of \defn{(channel) names} $\namesset =\set{x,y,\ldots}$ and
a finite set of \defn{labels} $\labelsset$.
We may denote by $\pnus x$ a generic sequence $\pnu{x_1}\cdots\pnu{x_n}$ of $\nu$-constructors of length $n>0$, and
we may simply write $\plsend x{\lab:P_\lab}{}$ \resp{$\plrecv x{\lab:P_\lab}{}$} as a shortcut for
$\plsend x{\lab: P_\lab}{\lab\in L}$ \resp{$\plrecv x{\lab: P_\lab}{\lab\in L}$} whenever $L=\set\lab$.
A \proc is \defn{\sequential} if it contains no parallel ($\ppar$) or restrictions ($\nu$), it is \defn{\flat}\footnote{Sometimes referred to as  \emph{non hierarchical} in the literature.}  if of the form $P=\pnus x \left(P_1\ppar\cdots\ppar P_n\right)$ for some \sprocs $P_1,\ldots, P_n$ (also called \seqcomps of $P$).
We use the common notation $P\fsubst xy$ for substitution (recalled in \Cref{fig:subst} of \Cref{app:proofs}).

The set $\freenamesof P$ of \defn{free names} and the set $\boundnamesof P$ of \defn{bound names} in a process $P$ are defined in \Cref{fig:terms}.
The set of \defn{names} in $P$ is denoted $\namesof P$ and a name $x$ is \defn{fresh} in $P$ if $x\notin \namesof P$.
A \defn{context} is a \proc $\cP\ctx$ containing a single occurrence of a special free name $\chole$ called \defn{hole} such that $\cP\ctx[P]\coloneqq\cP\fsubst P\chole$ is a \proc.
A \defn{\netcont} is a context of the form $\cN\ctx=\pnus x\left(\ctx\ppar P_1 \ppar \cdots \ppar P_n\right)$.

{The} \defn{$\alpha$-equivalence} ($\alphaeq$) is recalled in \Cref{fig:equivs}.
To improve the presentation of the technical results, we assume \procs written in an \defn{unambiguous} form, that is, in such a way each bound variable $x\in \boundnamesof{P}$ is bound by a unique $\nu$-constructor and do not occur free in $P$.
In the same figure we provide the relation $\strew$, whose reflexive and transitive closure is denoted $\strews$, and we define the standard \defn{(structural) equivalence} ($\steq$) as the equivalence relation generated by the union of $\strew$ and $\alphaeq$.

\begin{figure}[t]
	\adjustbox{max width=\textwidth}{$\begin{array}{c}
			\begin{array}{c|c}
				\begin{array}{r@{:\;}r@{\;\redsem\;}ll}
					\rscomr														&
					\psend xa P \ppar \precv xb Q
					&
					P \ppar Q \fsubst ab
					&
					\\
					\rschoir														&
					\plsend x{\lab: P_\lab}{\lab\in L}							&
					\plsend x{\lab_k:P_{\lab_k}}{}									&
					\mbox{if }\lab_k\in L
					\\
					\rslabr														&
					\plsend x{\lab_k:P_{\lab_k}}{}
					\ppar
					\plrecv x{\lab:Q_{\lab}}{\lab\in L}
					&
					P_{\lab_k} \ppar Q_{\lab_k}
					&
					\mbox{if } \lab_k\in L
				\end{array}
				&
				\begin{array}{r@{\;:\;}r@{\;\redsem\;}l@{\quad\mbox{if}\quad}l}
					\rsresr	& \pnu x P	&\pnu x P'	& P\redsem P'
					\\
					\rsparr	& P\ppar Q	& P'\ppar Q	& P\redsem P'
				\end{array}
			\end{array}
			\\[5pt]\hline\\[-5pt]
			\rsstrr	:\quad P \redsem Q 	\quad\mbox{if}\quad P\steq P' \redsem Q'\steq Q
		\end{array}$}
	\caption{\Opsem for the \picalc.}
	\label{fig:redsem}
\end{figure}

The \defn{\opsem} for \procs is defined by the relation $\redsem$ over processes induced by the rules in \Cref{fig:redsem}.
As standard, we denote by $\redsems$ the reflexive and transitive closure of $\redsem$.
As in~\cite{CDM14}, to allow for nondeterminism, the syntax of processes contains a construct $\plsend{x}{\lab:P_\lab}{\lab\in L}$ allowing for different options rather than the typical $\plsend{x}{P:\lab}{}$.
Thus, the corresponding rule $\rschoir$ for choosing among the available options induces a branching in the computation tree of the process.
We say that a \proc $P$ is \defn{\stuck} if $P \not \steq \pnil$ and there is no $P'$ such that
$P \redsem P'$.
A \proc P is called \defn{\lfree} if there is no \stuck \proc $P'$ such that $P \redsems P'$.
Also, a \proc $P$ has \defn{\progress}%
\footnote{
	See \Cref{subsec:contributions} for the precise intended meaning of the term \emph{progress} in this paper.
}
if it is \lfree or $P\ppar Q$ is \lfree for a \stuck \proc $Q$.

\begin{remark}
	Intuitively, \lfreedom means that there is always a part of a process that can reduce~\cite{koba:lock,montesi:book}.
	Progress for processes, instead, was introduced in~\cite{DBLP:journals/mscs/CoppoDYP16} to characterise processes that get stuck merely because they lack a communicating partner that could be provided by the environment.

	For example, the process $P=\pnu x (\psend xa \pnil\ppar \precv xb\pnil \ppar\psend yc \pnil)$ is not \dfree because it reduces (via $\comr$) to the stuck \proc $\pnu x (\pnil \ppar \psend yc\pnil )\steq \precv yd\pnil$, but this later has progress since $\psend yd\pnil\ppar \precv yd\pnil$ is \dfree.
\end{remark}

A \proc $P$ has a \defn{race condition} if there is a \netcont $\cN\ctx$ such that $P$ is structurally equivalent to a term of the following shape.
$$\begin{array}{c@{\qquad}c}
	\cN\ctx[\psend xy R \ppar \psend xz Q]
	&
	\cN\ctx[\precv xy R \ppar \precv xz Q]
	\\
	\cN\ctx[\Plsend x\lab P L \ppar \Plsend x\lab P {L'} ]
	&
	\cN\ctx[\Plrecv x\lab P L \ppar \Plrecv x\lab P {L'} ]
\end{array}$$
A \proc $P$ is \defn{\rfree} if there is no $P'$ with a race condition such that $P\redsems P'$.

\begin{remark}\label{rem:determinismRfree}
	Race conditions identify in a syntactic way the semantic property of a \proc \emph{potentially} having nondeterministic executions because of concurrent actions on a same channel.
	For example, $P= \psend xa\pnil \ppar \precv xb \pnil \ppar \precv xc\pnil$ has a race condition, and it can reduce either to
	$P_b=\precv xb \pnil$ or to $P_c=\precv xc \pnil$ according to the way the reduction rule $\comr$ is applied.
	We specify `potentially' because, for example, the process $Q= \pnu x\left( \precv xb \pnil \ppar \precv xc\pnil\right)$ has a race but cannot reduce.
	In fact, in the execution of a \rfree \proc, rules $\rscomr$ and $\rslabr$ are applied deterministically.
	That is, the same send \resp{selection} is synchronised via a $\rscomr$ \resp{a $\rslabr$} with the same receive \resp{branching}, and vice versa, in any possible (branch of an) execution.
\end{remark}

\subsection{A Simpler Equivalent Presentation of the Reduction Semantics}
To simplify the presentation of the new methodologies we use in our new framework,
we replace the
structural equivalence $\steq$ with the \emph{precongruence} $\strews$ (as in~\cite{koba:lock,CM20}).
In particular, such a precongruence orients the direction of scope extrusion (by extending the scope of the binder as much as possible), but also rules out those rewritings that may add superfluous information such as $P\strew (P \ppar \pnil)$ or $P\strew (\pnu x P)$ for a $x\notin\namesof P$.
Thus in the \opsem we consider in this paper we employ the following rule instead of the standard $\rsstrr$ (see \Cref{fig:redsem}):
\begin{equation}\label{eq:prestreq}
	\begin{array}{cccccl}
		\rsstreqr	: P \redsem Q 	 &
		\mbox{if } 	&
		P\strews P'	&
		=			&
		\cP\ctx[S]\redsem \cP\ctx[S']&=Q
		\\&
		\mbox{with}	&
		P\neq P' 	&
		\mbox{and}	&
		S \redsem S' &\mbox{ not via }\rsstreqr
	\end{array}
\end{equation}

\begin{remark}\label{rem:semantics}
	The \opsem using the rule $\rsstreqr$ instead of $\rsstrr$ is weaker because the set of processes reachable via a step of $\rsstreqr$ is strictly contained in the set of processes reachable via $\rsstrr$.
	By means of example, consider the process $\psend xy \pnil \ppar \precv xz \pnil$ which reduces to
		both $\pnil\ppar\pnil$ and $\pnil$ using $\rsstrr$, but can only reduce to $\pnil \ppar \pnil$ using $\rsstreqr$.

	However, it is immediate to show that if $P\redsem P'$ via $\rsstrr$, then there is a $Q\steq P'$ such that $P\redsem Q$ via $\rsstreqr$.
	Therefore, the standard \opsem (containing the rule $\rsstrr$) is as informative as the one we consider here (where we use the rule $\rsstreqr$ instead) for the study of \lfreedom and for the definition of the race condition.
\end{remark}

In the definition of the rules of the \opsem,
the rules $\rscomr$, $\rschoir$ and $\rslabr$ are, in some sense, performing `meaningful' transformation on processes,
while rules $\rsresr$, $\rsparr$ and $\rsstreqr$ deal with the syntactic bureaucracy of rewriting modulo the structural equivalence.
In the proofs in the next sections we need to be able to identify in each reduction step $P\redsem P'$ the \subproc $\crxof{P}{P'}$ of $P$ (called \emph{core-redex}) which is irreversibly transformed to the \proc $\crtof{P}{P'}$ (called \emph{core-reductum}), as well as to measure the amount of syntactical manipulations we need to `reach' such a \subproc to apply a reduction step (which we call \emph{entropy}). We make these concepts precise in the next definitions and exemplify them in \Cref{fig:exEntropy}.
\begin{definition}
	Let $P$ and $P'$ \procs such that $P\redsem P'$.
	The \defn{core} $\crof P{P'}=(\crxof{P}{P'},\crtof{P}{P'})$ and
	the \defn{entropy} $\entof{P}{P'}\in\N$
	of $P\redsem P'$ are defined as:
	\begin{itemize}
		\item
		if $P\redsem P'$ via $\rscomr$, $\rslabr$ or $\rschoir$, then
		$\entof{P}{P'}=1$ and
		$\crof{P}{P'}=(P,P')$;

		\item
		if $P\redsem P'$ via $\rsparr$ \resp{$\rsresr$},
		then there are \procs $Q$ and $Q'$ such that $Q\redsem Q'$ and a context $\cP\ctx$ of the form $\chole\ppar R$ \resp{of the form $\pnu x (\chole)$} such that $P=\cP\ctx[Q]$ and $P'=\ctx[Q']$ by definition of the reduction step.
		Then
		$\entof{P}{P'}=2\entof{Q}{Q'}$
		and
		$\crof{P}{P'}=\crof{Q}{Q'}$;

		\item
		if $P\redsem P'$ via $\rsstreqr$, then there are \procs $Q$ and $Q'$ such that $P \strews Q \redsem Q'\strews P'$
		with
		$P \strews Q$ and
		$Q' \strews P'$.
		Then
		$\entof{P}{P'}=3\entof{Q}{Q'}$
		and
		$\crof{P}{P'}=\crof{Q}{Q'}$.
	\end{itemize}
	The \defn{core-reduction} of $P\redsem P'$ is the rule used to reduce $\crxof{P}{P'}$ to $\crtof{P}{P'}$.
\end{definition}

\begin{figure}[t]
	\adjustbox{max width=\columnwidth}{$
		\begin{array}{c@{\;}|@{\;}c@{\;}|@{\;}c@{\;}|@{\;}c@{\;}|@{\;}c}
			S & S' & \crxof{S}{S'} &\crtof{S}{S'}&\entof{S}{S'}
			\\\hline
			(\psend xa P \ppar \precv xy Q) \ppar R
			&
			(P \ppar Q\fsubst ay) \ppar R
			&
			\psend xa P \ppar \precv xy Q
			&
			P \ppar Q\fsubst ay
			&
			2
			\\\hline
			\pnu a \left(\psend ba P \right) \ppar \precv bc R
			&
			\pnu a \left(  P \ppar R \fsubst ac \right)
			&
			\psend ba P  \ppar \precv bc R
			&
			\psend ba P  \ppar \precv bc R
			&
			6
		\end{array}
		$}
	\caption{Examples of \procs $S$ and $S'$ such that $S\redsem S'$, and the core-redex, core-reductum, and entropy of the rewriting step.}
	\label{fig:exEntropy}
\end{figure}

\begin{definition}
	A \defn{\comptree} of a \proc $P$ is a tree of processes $\comptreeof{P}$ with root $P$, such that a \proc $Q'$ is a child of $Q$ if $Q\redsem Q'$, and such that:
	\begin{itemize}
		\item
		if the core-reduction of $Q\redsem Q'$ is a $\rscomr$ or a $\rslabr$, then $Q'$ is the unique child of $Q$;
		\item
		if the core-reduction of $Q\redsem Q'$ is a $\rschoir$, then the set $\set{Q_1,\ldots, Q_n}\ni Q'$ of children of $Q$ is such that the core-reduction of $Q\redsem Q_i$ is a $\rschoir$
		and
		$\crxof{Q}{Q_i}=\crxof{Q}{Q'}$ for all $i\in\intset1n$.
	\end{itemize}
	It is \defn{maximal} if each leaf of the tree is a process $R\steq \pnil$ or is \stuck.
\end{definition}

We conclude this subsection with this result, which, together with \Cref{rem:semantics}, allows us to consider each maximal \comptree as a witness of \lfreedom for \rfree processes.
\begin{restatable}{lemma}{maxcomptree}\label{lem:completeTree}
	Let $P$ be a \proc. If $P$ is \lfree, then each \comptree with root $P$ can be extended to a maximal \comptree whose leaves are \procs structurally equivalent to $\pnil$.
\end{restatable}

\subsection{Translating Processes into Formulas}

We define a translation of \picalc processes into $\PIL$ formulas.
\begin{definition}[Processes as Formulas]
	We associate to each \picalc \proc $P$ a formula $\fof P$ inductively defined as follows.
	\begin{equation}\label{eq:translation}
		\adjustbox{max width=.9\textwidth}{$\begin{array}{c}
			\begin{array}{r@{\;=\;}l@{\qquad}r@{\;=\;}l@{\qquad}r@{\;=\;}l}
				\fof{\pnil}		& \unit
				&
				\fof{P\ppar Q}	& \fof P \lpar \fof Q
				&
				\fof{\pnu x(P)}	& \lNu x {\fof P}
				\\[5pt]
				\fof\chole 		& \chole
				&
				\fof{\psend xyP}& \lsend xy\lprec \fof P
				&
				\fof{\precv xyP}& \lExp{y}{\lrecv xy\lprec \fof P}
			\end{array}
			\\[15pt]
			\fof{\plsend x {\lab:P_\lab}{\lab\in L}}
			=
			\biglbra[\ell\in L]{\left(\lsend{x}\ell\lprec\fof{P_{\ell}} \right)}
			\qquad
			\fof{\plrecv x {\lab:P_\lab}{\lab\in L}}
			=
			\biglsel[\ell\in L]{\left(\lrecv{x}{\ell} \lprec \fof{P_{\ell}}\right)}
		\end{array}$}
	\end{equation}
\end{definition}

{Note that assuming $P$ unambiguous, the translation is a clean formula.}

\begin{remark}
	The reader familiar with session types could be curious about the choice of representing by a $\lwith$-formula a \proc of the form $\plsend x {\lab:P_\lab}{\lab\in L}$ (whose session type is a $\lplus$-type)
	and, dually, by a $\lplus$-formula a process $\plrecv x {\lab:P_\lab}{\lab\in L}$ (whose session type is a $\lwith$-type).
	This is only an apparent contradiction because \emph{our formulas are not types}. Rather, they encode processes whose executions are then derivations in the $\PIL$ system.
	Under this new interpretation, {during proof search} the rule for $\lwith$ gives exactly the expected branching of possible executions of terms like $\plsend x {\lab:P_\lab}{\lab\in L}$, corresponding to rule $\rschoir$ in the \opsem.
	Rule $\rslabr$ can then be applied `afterwards' (above in the derivation) to select the appropriate branch at the receiver, discarding all the others. Thus, in the formulas-as-processes, receiving a label corresponds to $\oplus$.

	For the same reason, parallel composition is represented by $\lpar$ (as in \cite{miller:pi,asc:gen:par}), while in most works using propositions as session types it is represented by $\cutr$ and $\ltens$.
	We will come back to this aspect in \Cref{sec:rel}.
\end{remark}

\begin{proposition}\label{prop:steqISleqiv}
	Let $P_1$ and $P_2$ \procs.
	If $P_1\strew P_2$ then $\fof{P_2}\limp \fof{P_1}$.
\end{proposition}
\begin{proof}
	$\fof{P \ppar Q} \feq \fof{Q \ppar P}$ and $\fof{(P \ppar Q) \ppar R }\feq \fof{P \ppar (Q \ppar R)}$ derive from commutativity and associativity of $\lpar$ (see \Cref{prop:feq}).
	The logical equivalences $\fof{P \ppar \pnil} \feq \fof P$ and $\fof {\pnu x \pnil} \feq \fof \pnil$
	are direct consequence of the ones in \Cref{fig:Seq}.
	The implication $\fof{\pnu x (P \ppar Q)} \limp \fof{\pnu x P \ppar Q} $ for $x\notin\freeof Q$ is shown in \Cref{prop:feq}.
	Finally, $\fof{\pnu x \pnu y P} \feq \fof{\pnu y \pnu x P}$ derives from the quantifier shifts
	$\lNu x {\lNu y P} 	= \lNu y {\lNu x P}$ (\Cref{fig:Seq}).
\end{proof}

\subsection{\LFreedom as Provability in $\PIL$}

\begin{figure}[t]
	\centering
	\adjustbox{max width=\textwidth}{$\begin{array}{c}
		\vlderivation{
			\vliq{\lpar}{}{
				\vdash \left(\cneg{\fof P} \ltens \cneg{\fof Q}\fsubst{y}{z}\right) \lpar
				\left(\lsend xy\lprec \fof P \lpar \lExp{z}{\lrecv xz\lprec \fof Q}\right)
			}{
				\vlin{\exists}{}{
					\vdash (\cneg{\fof P} \ltens \cneg{\fof Q}\fsubst{y}{z}) ,
					\lsend xy\lprec \fof P , \lExp{z}{\lrecv xz\lprec \fof Q}
				}{
					\vliin{\lprec}{}{
						\vdash (\cneg{\fof P} \ltens \cneg{\fof Q}\fsubst{y}{z}),
						\lsend xy\lprec \fof P , \lrecv xy\lprec \fof Q\fsubst{y}{z}
					}{
						\vlin{\axrule}{}{\vdash \lsend xy ,\lrecv xy}{\vlhy{}}
					}{
						\vliin{\ltens}{}{
							\vdash (\cneg{\fof P} \ltens \cneg{\fof Q}\fsubst{y}{z}),
							\fof P ,  \fof Q\fsubst{y}{z}
						}{
							\vlpr{}{\text{\Cref{thm:cutelim}.\ref{cutelim:1}}}{
								\vdash \cneg{\fof P}, \fof P
							}
						}{
							\vlpr{}{\text{\Cref{thm:cutelim}.\ref{cutelim:1}}}{
								\vdash \cneg{\fof Q}\fsubst{y}{z},\fof Q\fsubst{y}{z}
							}
						}
					}
				}

			}
		}
		\qquad
		\vlderivation{
			\vlin{\lpar}{
			}{
				\vdash \biglsel[\ell\in L]{\left(\lrecv{x}{\ell} \lprec \cneg{\fof{P_\lab}}\right)} \lpar
				\biglbra[\ell\in L]{\left(\lsend{x}\ell\lprec\fof{P_{\ell}} \right)}
			}{
				\vlin{\lwith}{}{
					\vdash \biglsel[\ell\in L]{\left(\lrecv{x}{\ell} \lprec \cneg{\fof{P_\lab}}\right)},
					\biglbra[\ell\in L]{\left(\lsend{x}\ell\lprec\fof{P_{\ell}} \right)}
				}{
					\multiprem{
						\vlin{\oplus }{}{
							\vdash  \biglsel[\ell\in L]{\left(\lrecv{x}{\ell} \lprec \cneg{\fof{P_\lab}}\right)},
							\lsend{x}\ell\lprec\fof{P_{\ell}}
						}{
							\vliin{\lprec}{}{
								\vdash \lrecv{x}{\ell} \lprec \cneg{\fof{P_\lab}},
								\lsend{x}\ell\lprec\fof{P_{\ell}}
							}{
								\vlin{\axrule}{}{
									\vdash  \lsend{x}\lab , \lrecv{x}{\lab}
								}{\vlhy{}}
							}{
								\vlpr{}{\text{\Cref{thm:cutelim}.\ref{cutelim:1}}}{
									\vdash \cneg{\fof{P_{\lab}}}, \fof{P_{\lab}}
								}
							}
						}
					}{\lab\in L}
				}
			}
		}
		\\\\
		\vlderivation{
			\vliq{\lpar}{}{
				\vdash \left(\cneg{\fof{Q_{\ell_k}}} \ltens \cneg{\fof{R_{\ell_k}}}\right)
				\lpar
				\left(\lsend{x}{\ell_k}\lprec\fof{Q_{\ell_k}} \lpar \biglsel[\ell \in L]{\left(\lrecv{x}{\ell} \lprec {\fof{R_\lab}}\right)}\right)
			}{
				\vlin{\lplus}{}{
					\vdash \left(\cneg{\fof{Q_{\ell_k}}} \ltens \cneg{\fof{R_{\ell_k}}}\right)
					,
					\lsend{x}{\ell_k}\lprec\fof{Q_{\ell_k}} , \biglsel[\ell \in L]{\left(\lrecv{x}{\ell} \lprec {\fof{R_\lab}}\right)}
				}{
					\vliin{\lprec}{}{
						\vdash \left(\cneg{\fof{Q_{\ell_k}}} \ltens \cneg{\fof{R_{\ell_k}}}\right)
						,
						\lsend{x}{\ell_k}\lprec\fof{Q_{\ell_k}} ,\lrecv{x}{\ell_k} \lprec {\fof{R_{\lab_k}}}
					}{
						\vlin{\axrule}{}{\vdash \lsend{x}{\ell_k},\lrecv{x}{\ell_k}  }{\vlhy{}}
					}{
						\vliin{\ltens}{}{
							\vdash \left(\cneg{\fof{Q_{\ell_k}}} \ltens {\fof{R_{\ell_k}}}\right),
							\fof{Q_{\ell_k}} , \cneg{\fof{R_{\lab_k}}}
						}{
							\vlpr{}{\text{\Cref{thm:cutelim}.\ref{cutelim:1}}}{\vdash \cneg{\fof{Q_{\ell_k}}}, \fof{Q_{\ell_k}}}
						}{
							\vlpr{}{\text{\Cref{thm:cutelim}.\ref{cutelim:1}}}{\vdash {\fof{R_{\ell_k}}}, \cneg{\fof{R_{\ell_k}}}}
						}
					}
				}
			}
		}
	\end{array}$}
	\caption{
		Derivations in $\PIL$ corresponding to the rules $\rscomr$, $\rschoir$ and $\rslabr$ of the \opsem of the \picalc.
	}
	\label{fig:derivationspisys}
\end{figure}

We can now establish a correspondence between process reductions and linear implication in $\PIL$,
as well as a correspondence between each computation tree with root a process $P$ and a proof search strategy in $\PIL$ for the formula $\fof P$.
Combining these two results, we obtain a purely logical characterisation of \lfree \procs as pre-images via $\fof\cdot$ of formulas derivable in $\PIL$.

\begin{restatable}{lemma}{lemSimulation}\label{lem:simulation}
	Let $P$ and $P'$ \procs.
	\begin{enumerate}
		\item\label{simulation:1}
		If $P\strews P'$, then $\fof{P'}\limp \fof{P}$.

		\item\label{simulation:4}
		If $P\redsem P'$, then either
		\begin{enumerate}
			\item the core-reduction of $P \redsem P'$ is a $\rscomr$ or a $\rslabr$, and
			$\proves[\PIL]\fof{P'} \limp \fof P$;

			\item or the core-reduction of $P \redsem P'$ is a $\rschoir$ then
			there is a set $\set{P_\lab \mid \lab\in L}\ni P'$
			such that
			$P \redsem P_\lab$
			for all $\lab \in L$ and
			$\proves[\PIL] \left(  \bigwith_{\lab\in L}\fof{P_\lab}\right)\limp \fof{ P}$.
		\end{enumerate}
	\end{enumerate}
\end{restatable}
\begin{proof}
	\Cref{simulation:1} is proven using \Cref{prop:feq}
	and transitivity of $\limp$ (see \Cref{thm:cutelim}.\ref{cutelim:3}).
	To prove \Cref{simulation:4} we reason by induction on \entropy:
	\begin{itemize}
		\item if $\entof{P}{P'}=1$ then $P\redsem P'$ via $\rscomr$, $\rslabr$ or $\rschoir$ and we conclude using the derivations in \Cref{fig:derivationspisys};

		\item if $P\redsem P'$ via $\rsparr$ \resp{$\rsresr$},
		then there is a context $\cP\ctx=\left(\chole \ppar R\right)$ \resp{$\cP\ctx=\pnu x\chole$} such that
		$P=\cP\ctx[S]$  and  $P'=\cP\ctx[S']$.
		We conclude using \Cref{thm:der}.\ref{der:2} and \Cref{thm:cutelim}.\ref{cutelim:3};

		\item if $P\redsem P'$ via $\rsstreqr$, then there is $S$ such that $P\strews S$ and $S\redsem P'$ (via a rule different from $\rsstreqr$).
		We conclude by \Cref{thm:cutelim}.\ref{cutelim:3} using \Cref{simulation:1} and \Cref{thm:der}\ref{der:3}.
	\end{itemize}
\end{proof}

Using this lemma we can prove the correspondence between \lfreedom and derivability in $\PIL$.

\begin{restatable}{theorem}{thmTrees}\label{thm:deadlock}
	Let $P$ be a \rfree \proc.
	Then $P$ is \lfree iff $\proves[\PIL]\fof P$.
	{More precisely, $\proves[\set{\lunit,\axrule,\lpar,\mixr,\lprec,\lwith,\lplus,\exists,\nuurule}]\fof P$}.
\end{restatable}
\begin{proof}
	It suffices to establish a correspondence between maximal \comptrees and derivations in $\PIL$.
	Details are provided in the appendix.

	\textbf{($\Rightarrow$)}
	If $P$ is \lfree, then, any maximal \comptree $\comptreeof{P}$ with root $P$ has leaves which are processes structurally equivalent to $\pnil$ by \Cref{lem:completeTree}.
	By induction on the structure of $\comptreeof{P}$, we can define a derivation in $\NL\wcut$ composing (using $\cutr$) the derivations allowing simulating the transitions of the reduction semantics (\Cref{lem:simulation});
	thus we obtain a derivation in $\NL$ by applying cut-elimination (\Cref{thm:cutelim}.\ref{cutelim:2}).
	We conclude observing that the subformula property ensures that only rules in such cut-free derivation are in $\set{\lunit,\axrule,\lpar,\mixr,\lprec,\lwith,\lplus,\exists,\nuurule}$.
	Note that all judgements in such a derivation are empty.

	\textbf{($\Leftarrow$)}
	To prove the converse, we show that each derivation $\dD_P$ of $\fof P$ can be transformed using the \emph{rule permutations} in \Cref{fig:permutations} into a derivation $\widetilde{\dD_P}$ made of blocks of rules consisting of sequences of $\nurule$- and $\lpar$-rules only, or blocks as the ones shown in \Cref{eq:blocks} below.%
	\begin{equation}\label{eq:blocks}\adjustbox{max width=.9\textwidth}{$
		\vlderivation{
			\vlin{\exists}{}{
				\vpz1{}\vdash \lsend xy\lprec A,\lEx z\left(\lrecv xz \lprec B \right), \Gamma
			}{
				\vliin{\lprec}{}{
					\vdash \lsend xy\lprec A ,\lrecv xy \lprec B\fsubst yz, \Gamma
				}{
					\vlin{\axrule}{}{\vdash \lsend xy,\lrecv xy}{\vlhy{}}
				}{
					\vlhy{
						\vdash A , B\fsubst yz, \Gamma
					}
				}
			}
		}
	\qquad\qquad
		\vlderivation{
			\vlin{\lwith}{}{
				\vpz1{}\vdash
				\biglsel[\ell\in L_1]{\left(\lrecv{x}{\ell} \lprec \fof{Q_{\ell}}\right)},
				\biglbra[\ell\in L_2]{\left(\lsend{x}\ell\lprec\fof{R_{\ell}} \right)},
				\Gamma
			}{\multiprem{
					\vlin{\oplus}{}{
						\vdash
						\biglsel[\ell\in L_1]{\left(\lrecv{x}{\ell} \lprec \fof{Q_{\ell}}\right)},
						\lsend{x}\ell\lprec\fof{R_{\ell}},
						\Gamma\;
					}{
						\vliin{\lprec}{}{
							\vdash
							\lrecv{x}{\ell} \lprec \fof{Q_{\ell}},
							\lsend{x}\ell\lprec\fof{R_{\ell}},
							\Gamma
						}{
							\vlin{\axrule}{}{\vdash \lsend x\lab,\lrecv x\lab}{\vlhy{}}
						}{
							\vlhy{
								\vdash
								\fof{Q_{\ell}},
								\fof{R_{\ell}},
								\Gamma
							}
						}
					}
				}{\lab\in L_1}
			}
		}
	$}\end{equation}
	We conclude by induction on the number of such blocks, since each block in the left \resp{right} of \Cref{eq:blocks} identifies an application of a $\comr$ \resp{a $\brar$ followed by a $\selr$}.

	Note that since $P$ is \rfree, then it suffices to reason on a single \comptree and not to take into account all possible \comptrees of $P$.
\end{proof}

\begin{corollary}\label{cor:progress}
	Let $P$ be a \rfree \proc.
	Then
	$P$ has \progress
	iff
	there is a \nuparcont $\cC\ctx$ such that
	$\proves[\PIL] \cC\ctx[P]$.
\end{corollary}

\begin{figure}[t]
	\centering
	\adjustbox{max width=\textwidth}{$\begin{array}{c}
		\vlderivation{
			\vlin{\rrule[1]_1}{}{\vdash \Gamma,\Delta}{
				\vlin{\rrule[1]_2}{}{\vdash\Gamma', \Delta}{
					\vlhy{\vdash  \Gamma', \Delta'}
				}
			}
		}
	\;\sim\;
		\vlderivation{
			\vlin{\rrule[1]_2}{}{\vdash \Gamma,\Delta}{
				\vlin{\rrule[1]_1}{}{\vdash\Gamma, \Delta'}{
					\vlhy{\vdash  \Gamma', \Delta'}
				}
			}
		}
	\qquad
		\vlderivation{
			\vliin{\rrule[2]}{}{\vdash \Gamma, \Sigma, \Delta}{
				\vlin{\rrule[1]}{}{\vdash\Gamma, \Sigma'}{
					\vlhy{\vdash  \Gamma', \Sigma'}
				}
			}{
				\vlhy{\vdash \Sigma', \Delta}
			}
		}
	\;\sim\;
		\vlderivation{
			\vlin{\rrule[1]}{}{\vdash \Gamma, \Sigma, \Delta}{
				\vliin{\rrule[2]}{}{\vdash\Gamma', \Sigma, \Delta
				}{
					\vlhy{\vdash\Gamma', \Sigma' }
				}{
					\vlhy{\vdash \Sigma', \Delta}
				}
			}
		}
	\qquad
		\vlderivation{
			\vlin{\lwith}{}{\vdash \Gamma, \bigwith_{i\in I} A_i}{
				\multiprem{\vlin{\rrule[1]}{}{\vdash\Gamma, A}{
					\vlhy{\vdash  \Gamma', A}
				}}{i\in I}
			}
		}
	\;\sim\;
		\vlderivation{
			\vlin{\rrule[1]}{}{\vdash \Gamma, \bigwith_{i\in I} A_i}{
				\vlin{\lwith}{}{\vdash\Gamma', \bigwith_{i\in I} A_i
				}{
					\multiprem{\vlhy{\vdash\Gamma', A_i }}{i\in I}
				}
			}
		}
	\\[20pt]
		\vlderivation{
			\vliin{\rrule[2]_1}{}{\vdash\Gamma,\Delta,\Sigma}{
				\vliin{\rrule[2]_2}{}{\vdash\Gamma, \Delta,\Sigma'}{
					\vlhy{\vdash\Gamma',\Sigma'}
				}{
					\vlhy{\vdash\Gamma',\Delta}
				}
			}{
				\vlhy{\vdash \Sigma'}
			}
		}
	\;\sim\;
		\vlderivation{
			\vliin{\rrule[2]_2}{}{\vdash\Gamma,\Delta,\Sigma}{
				\vliin{\rrule[2]_1}{}{\vdash\Gamma',\Sigma}{
					\vlhy{\vdash\Gamma',\Sigma'}
				}{
					\vlhy{\vdash\Sigma'}
				}
			}{
				\vlhy{\vdash \Gamma',\Delta}
			}
		}
	\qquad
		\vlderivation{
			\vliin{\rrule[2]}{}{
				\vdash\Gamma,\Delta,\bigwith_{i\in I} A_i
			}{
				\vlin{\lwith}{}{
					\vdash\Gamma,\Delta_1,\bigwith_{i\in I} A_i
				}{
					\multiprem{\vlhy{\vdash\Gamma,\Delta_1,A_i}}{i\in I}
				}
			}{
				\vlhy{\vdash \Delta_2}
			}
		}
	\;\sim\;
		\vlderivation{
			\vlin{\with}{}{\vdash\Gamma,\Delta,\bigwith_{i\in I} A_i}{
				\multiprem{\vliin{\rrule[2]}{}{\vdash\Gamma,\Delta,\bigwith_{i\in I} A_i}{
					\vlhy{\vdash\Gamma,\Delta_1,A_i}
				}{
					\vlhy{\vdash \Delta_2}
				}}{i\in I}
			}
		}
	\end{array}$}
	\caption{
		Rule permutations
		with
		$\rrule[1],\rrule[1]_1,\rrule[1]_2\in\set{\lpar,\exists,\lplus,\nuur}$
		and
		$\rrule[2],\rrule[2]_1,\rrule[2]_2\in\set{\lprec,\ltens,\mixr}$.
	}
	\label{fig:permutations}
\end{figure}

We conclude this section by showing that \progress for \procs which never send `private' channels can be easily captured in this new setting.
Specifically, we say that a \proc $P$ has \defn{private mobility} if it is of the form $P=\cP\ctx[\Apsend ax]$ for an $a$ bound by a $\nu$ in $\cP$.
We also denote by $\fofi{P}$ the formula obtained by replacing with a unit ($\lunit$) any atom in $\fof P$ of the form $\lsend xy$ or $\lrecv xy$ for any $x\in\set{\xs}$.
\begin{theorem}\label{prop:progress}
	Let $P$ be a \rfree \proc without private mobility.
	Then
	$P$ has \progress
	iff
	$\proves[\PIL]\fofi[\freenamesof{P}]{P}$.
\end{theorem}
\begin{proof}
	We prove a simulation result (as \Cref{lem:simulation}) for $\fofi[\freenamesof{P}]{P}$, and we conclude with the same argument used in the proof of \Cref{thm:deadlock}.
	It $P$ is \lfree, then we conclude as in \Cref{thm:deadlock}.
	Otherwise, since $P$ has \progress there is $Q$ such that $P \ppar Q$ is \lfree.
	By definition, we must have that
	$\namesof{Q}=\freenamesof{Q}$ (otherwise either $Q$ is not \stuck or $P \ppar Q$ is not \lfree)
	and that
	$\freenamesof{Q}=\freenamesof{P}=\xs$.
	Thus $\fofi{P\ppar Q}\feq \fofi{P}$ by the fact that $\fofi{Q}$ contains no atoms (i.e., only units) and $\lunit \lprec A \feq A \feq \lunit \lpar A$ (see \Cref{prop:feq}).
	\begin{itemize}
		\item If $P \ppar Q $ is \lfree for $Q$ \stuck and $P \ppar Q \redsem R$ via $\rsparr$
		then, $R = P' \ppar Q$ and $P \redsem P'$.
		If the core-reduction of $P\redsems P'$ is a $\rscomr$ or a $\rslabr$,
		then $\proves[\PIL]\fofi{P'}\limp \fofi{P}$;
		if the core-reduction is a $\rschoir$,
		there is a set of \procs $\set{P_\lab}_{\lab\in L}\ni P'$ such that
		$\proves[\PIL] \left(  \bigwith_{\lab\in L}\fofi{P_\lab}\right)\limp \fofi{ P}$.
		This is proven by induction on the entropy as in \Cref{lem:simulation}.
		\item
		If $P'\ppar Q\redsem P'$ not via $\rsparr$ and then the core-reduction is either a $\comr$ or a $\selr$.
		In this case can prove as that
		$\proves[\PIL]\fofi{P'}\limp \fofi{P'}$ because $\fofi{P'}\feq\fofi{P'\ppar Q}\feq \fofi{P'}$.
	\end{itemize}
\end{proof}

\begin{remark}
	To understand the requirement on private mobility in~\Cref{prop:progress}, consider the process $P=\pnu a (\psend ba  \psend ac \pnil)$. This process has progress, because $$(P\ppar \precv bx \precv xc \pnil) \steq  \pnu a (\psend ba  \psend ac \pnil \ppar \precv bx \precv xc \pnil ) \redsems \pnil \mydot$$
	However,
	$\fofi[\set{b}]{P} = \lNup a{\lunit \lprec \lsend ac \lprec \lunit}$
	is not derivable in $\PIL$.
	This makes our characterisation of progress as powerful as in previous work~\cite{CDM14}
	(where the condition is not made explicit but clearly necessary, see the definition of `co-process' therein).
\end{remark}

\section{Completeness of \Chors}\label{sec:choreo}

In this section we prove that any \lfree \fproc can be expressed as a choreography, as intended in the paradigm of choreographic programming~\cite{M13:phd}.
Key to this result is establishing a \emph{proofs-as-\chors} correspondence, whereby \chors can be seen as derivations in the $\PIL$ system.

To this end, we first introduce the syntax and semantics of \emph{\chors}, the typical accompanying language for describing their implementations in terms of located processes (the \emph{endpoint calculus}), and a notion of endpoint projection (EPP) from choreographies to processes.
We then define the sequent calculus $\ChorL$ operating on sequents in which (occurrences of) formulas are labelled by process names, and we conclude by establishing the proofs-as-choreographies correspondence.

\subsection{\Chors}\label{back:chore}

\begin{figure}[t]
	\adjustbox{max width=\textwidth}{$\begin{array}{c}
			\mbox{\Chors}
			\\
				\begin{array}{rccccccl}
					\chC,\chC_\lab  \coloneqq&
					\chnil		&\mid&	\gencom  \chC	&\mid&\genchois&\mid&\cnu x \chC^x \mbox{ (with $\chC^x$ containing no $(\nu x)$)}
					\\&
					\mbox{end}	&&	\mbox{communication}&&\mbox{choice}&&\mbox{restriction}
				\end{array}
			\\\hline\hline
			\mbox{\Opsem for Choreographies}
			\\
			\begin{array}{lcr@{\;}c@{\;}ll}
				\ccomr	&:&
				\gencom \chC
				&\osto{\tcom \pp\pq k}&
				\chC\fsubst xy
				\\
				\cchoir &:&
				\genchois
				&\osto{\;\;\tbra\pp k\;\;}&
				\genchoi
				&\mbox{for a }\lab_i\in L
				\\
				\clabr &:&
				\genchoi
				&\osto{\tcom \pp\pq k}&
				\chC_{\lab_i}
				&\mbox{if }
				\lab_i \in L'
				\\
				\crestr &:&
				\cnu x \chC
				&\osto{\;\;\tm\;\;}&
				\cnu x \chC'
				&\mbox{if }
				\chC\osto{\tm}\chC'
				\\\hline
				\cdelir &:&
				\gencom \chC
				&\osto{\;\;\tm'\;\;}&
				\gencom\chC'
				&
				\mbox{if }
				\chC\osto{\tm'}\chC'
				\\
				\cdelcr &:&
				\genchois
				&\osto{\;\;\tm'\;\;}&
				\Choics[\chC']{p}{q}{k}{\lab}{L}
				&
				\begin{array}{l}
					\mbox{if } \chC_\lab\osto{\tm'}\chC_\lab'
					\\
					\mbox{for all $\lab\in L$}
				\end{array}
				\\
				\multicolumn{6}{c}{\mbox{with } \pnamesin{\tm'}\cap\set{\pp,\pq,k}=\emptyset}
			\end{array}
		\end{array}$}
	\caption{Syntax and semantics for choreographies, where
		$\pp$ and $\pq$ are distinct process names in $\procset$,
		$x,y\in\varset$,
		$L\subseteq L' \subseteq \labelsset$,
		and
		$S_\lab$ are \sprocs
		(see \Cref{rem:garbage}).
	}
	\label{fig:chor}
\end{figure}

In a choreographic language, terms (called \emph{choreographies}) are coordination plans that express the overall behaviour of a network of processes~\cite{montesi:book}.
The \defn{\chors} that we consider in this paper are generated by a set of \defn{process names} $\procset$, a set of variables $\varset$, and a set of \defn{selection labels} $\labelsset$ as shown in \Cref{fig:chor}.
A choreography can be either:
\begin{itemize}
	\item $\chnil$, the terminated choreography;
	\item $\gencom C$, a communication from a process $\pp$ to another $\pq$ with a continuation $C$ ($y$ is bound in $C$ and can appear only under $\pq$);
	\item $\genchois$, a choice by a process $\pp$ of a particular branch $L$ offered by another \proc $\pq$\footnote{
		The set $L'$ of labels the \proc $\pq$ can accept contains the set $L$ of labels $\pp$ can send.
		For the continuation of labels in $L'\setminus L$ we only allow \sprocs because we do not allow nested parallel in the target language of the projection (see next subsection).
	}; or
	\item $\cnu x \chC^x$, which restricts $x$ in a \chor $\chC^x$ in which the variable $x$ always occur free (i.e., no $\cnu x{}$ occurs in $\chC^x$).%
	\footnote{
		By allowing the construct $\cnu x $ only in the case in which $x$ is not bound in $\chC^x$, we ensure that \chors are always written using \emph{Barendregt's convention}.
		This means that each variable $x$ can be bound by at most one restriction $\nu$, and $x$ cannot appear both free and bound in a \chor $\chC$.
		As a result, we can adopt a lighter labelling discipline for the reduction semantics compared to the one used in \cite{CM13} -- see the rule $\crestr$ in \Cref{fig:chor}.
	}
\end{itemize}
Note that we consider communication of \proc names or variables only (that is, $k\in\procset\cup\varset$).
We say that a \chor is \defn{\cflat} if it is of the form $\cnus x \chCrf$ for a \defn{restriction-free} (i.e., containing no occurrences of $\nu$) \chor $\chCrf$.
In the same figure we also provide the \defn{\opsem} of our choreographic language, where each reduction step is labelled by a \defn{reduction label} $\tm$ from the following set.
\begin{equation}\label{eq:labels}
	\Set{\tcom \pp\pq k \qomma \tbra \pp k   \mid \pp,\pq\in\procset, \; k\in\varset}
\end{equation}
To each {reduction label} $\tm$ we associate the set $\pnamesin{\tm}$ of \procs names and variables occurring in it -- i.e. $\pnamesin{\tcom \pp\pq k}=\set{\pp,\pq,k}$ and $\pnamesin{ \tbra \pp k }=\set{\pp,k}$.

In the semantics, $\ccomr$ executes a communication while $\cchoir$ allows a process $\pp$ to make an internal choice.
Rule $\clabr$ then communicates a label from $\pp$ to $\pq$, which then continue with the \chor $\chC_\lab$ (but never with a \sproc $S_\lab$, see \Cref{rem:garbage}).
Rule $\crestr$ lifts reductions under restrictions.
Lastly, $\cdelir$ \resp{$\cdelcr$} models the standard out-of-order execution of independent communications that can be reduced by rule $\ccomr$ \resp{both rules $\cchoir$ and $\clabr$} -- this is the choreographic equivalent of parallel composition in process calculi~\cite{montesi:book}.
\begin{example}\label{ex:chC}
	The next choreography expresses the communication behaviour of the processes given in~\Cref{eq:intro}.
	\begin{equation}\label{eq:choreo}
		\comc paqax
		;
		\choic p{{\set{\lab}}}q{{\set{\lab,\lab'}}}y{
			\lab\colon \comc pbqby ; \chnil
			\\
			\lab'\colon \psend zc \pnil
		}
	\end{equation}
	It can be executed by applying rule $\ccomr$ and then rule $\clabr$.
	Note that we do not need to use rule $\cchoir$ before applying $\clabr$, because the set of labels $L$ in the choice constructor is a singleton.
\end{example}

\begin{remark}\label{rem:garbage}
	From the programmer's viewpoint, choice instructions may contain some unnecessary information since no label $\lab'\in L'\setminus L$ will never be selected during the execution of a choreography -- and thus no continuation will execute the process $S_{\lab'}$.
	This `garbage' code is typical of works on choreographies and logic~\cite{CMS18,carbone:multiparty_choreo}, and we share the same motivation: we want to be able to capture the entire \flat fragment of the \picalc, where such garbage code cannot be prohibited.
	For example, without garbage code, the choreography in \Cref{eq:choreo} would not be a complete representation of the \eproc in \Cref{eq:exampleIntroEndpoint} (see also \Cref{eq:intro}).
\end{remark}

\subsection{Endpoint Projection}

\begin{figure}[t]
	\adjustbox{max width=\textwidth}{$\begin{array}{c}
			\begin{array}{c}
				\mbox{Structural Equivalence}
				\\
				\proloc{\pp}\pnil \ppar P
				\;\steq\;
				P
				\qquand
				\pnu{x_1}\cdots\pnu{x_k}\prod_{i=1}^{n}\proloc{\pp_i}{S_i}
				\;\steq\;
				\pnu{x_{\tau(1)}}\cdot\pnu{x_{\tau(k)}}\prod_{i=1}^{n}\proloc{\pp_{\sigma(i)}}{S_{\sigma(i)}}
				\\
				\mbox{for any $\sigma$ permutation over $\intset1n$ and $\tau$ over $\intset1k$}
				\\\hline\\[-10pt]
				\begin{array}{c}
					\mbox{\OpSem}
					\\
					\begin{array}{l@{\;:\;}rcll}
						\Ecomr	&
						\cN\ctx[
						\proloc\pp {\psend k x S} \ppar  \proloc \pq {\precv k y S'}
						]
						&\osto{\tcom \pp\pq k}&
						\cN\ctx[
						\proloc\pp{S} \ppar  \proloc\pq{S'\fsubst xy}
						]
						\\
						\Elsendr	&
						\cN\Ctx[
						\proloc \pp{\plsend{k}{\lab: S_\lab}{\lab\in L}}
						]
						&\osto{\;\;\tbra \pp k\;\;}&
						\cN\Ctx[
						\proloc \pp{\plsend{k}{\lab_i:S_{\lab_i}}{}}
						]
						& \mbox{ for each } \lab_i\in L
						\\
						\Elrecvr	&
						\cN\Ctx[
						\proloc \pp{\plsend {k}{\lab_i:S_{\lab_i}}{}}
						\ppar
						\proloc \pq{\plrecv {k}{\lab : S'_\lab}{\lab\in L}}
						]
						&\osto{\tcom \pp\pq k}&
						\cN\Ctx[
						\proloc \pp{S_{\lab_i}}
						\ppar
						\proloc \pq{S'_{\lab_i}}
						]
						& \mbox{ if } \lab_i\in L
					\end{array}
				\end{array}
			\end{array}
		\end{array}$}
	\caption{
		Simplified presentation of the structural equivalence and \opsem for the endpoint calculus,
		where $\cN\Ctx[P]\steq\pnus x \left( P \ppar \prod_{i=1}^{n}\proloc{\pp_i}{T_i} \right)$.
	}
	\label{fig:EPC}
\end{figure}

Our choreographies can be mechanically translated into processes via the standard technique of endpoint projection (EPP)~\cite{montesi:book}.
To simplify the presentation of projection, we adopt the standard convention of enriching the language of \procs with process names labelling each \seqcomp of a \fproc~\cite{H07,montesi:book}.
That is, a \fproc of the form $\pnus x\left(S_1 \ppar \cdots \ppar S_n\right)$ is represented by an \defn{\eproc}
\begin{equation}
	\pnus x\left(\proloc{\pp_1}{S_1} \ppar \cdots \ppar \proloc{\pp_n}{S_n}\right)
	\quor
	\pnus x\left(\prod_{i=1}^{n}\proloc{\pp_i}{S_i}\right)
\end{equation}
where all names $\pp_1,\ldots,\pp_n$ are distinct.\footnote{In the literature, a process $P$ with name $\pp$ is usually written $\pp[P]$~\cite{H07,montesi:book}. We adopt the alternative writing $\proloc{\pp}{P}$ to avoid confusion with the notation used for contexts.}
The process calculus over these processes is dubbed \defn{endpoint calculus}.

The definition of \defn{endpoint projection} is provided by the partial function $\mathsf{EPP}$ in \Cref{fig:EPP} (it is defined only for \cflat \chors, as in~\cite{CM13}).
It is a straightforward adaptation to our syntax of the textbook presentation of projection~\cite{montesi:book}. In particular, it uses a \defn{merge} operator $\sqcup$ (originally from ~\cite{CHY12}) to support propagation of knowledge about choices. That is, if a process $\pr$ needs to behave differently in two branches of a choice communicated from $\pp$ to $\pq$, it can do so by receiving different labels in these two branches. Merge then produces a term for $\pr$ that behaves as prescribed by the first (respectively the second) branch when it receives the first (respectively the second) label.
If $\EPP \chC$ is defined for $\chC$ we say that $\chC$ is \defn{\prochor}.
\begin{example}
The EPP of the \chor in \Cref{eq:choreo} is
\begin{equation}\label{eq:exampleIntroEndpoint}
	\adjustbox{max width=.9\textwidth}{$
		\pnu x\pnu y
		\left(
		\proloc{\pp}{\psend xa
			\plsend y{\lab: \psend yb \pnil}{}
		}
		\ppar
		\proloc{\pq}{
			\precv xa
			\plrecv y {\lab: \precv yb \pnil , \lab': \psend zc \pnil}{}
		}
		\right)
	$}
\end{equation}
which is precisely the one in~\Cref{eq:intro} annotated with process names.
\end{example}

Structural equivalence and \opsem for the endpoint calculus are obtained by the one of the \picalc assuming that each structural equivalence ($\steq$) and reduction step ($\redsem$) preserves the process names.
Note that, for the purpose of studying \lfreedom, structural equivalence and reduction rules can be simplified as shown in \Cref{fig:EPC}.
Each reduction step is labelled by the same labels used in the \opsem of choreographies (see \Cref{eq:labels}), allowing us to retain the information about which \seqcomps and channel are involved in each reduction step.

\begin{figure}[t]
	\adjustbox{max width=\textwidth}{$\begin{array}{c}
		\EPP\chC
		=
		\begin{cases}
			\pnus x \EPP{\chCrf}
			&
			\mbox{if $\chC=\pnus x \chCrf$ with $\chCrf$ restriction-free}
			\\
			\proloc{\pp_0}{\pnil}
			&
			\mbox{if } \chC=\chnil \mbox{ (for a given $\pp_0$)}
			\\
			\EPP[\pp_1]{\chC}\ppar \cdots \ppar \EPP[\pp_n]{\chC}
			&
			\mbox{otherwise, with $\pp_1,\ldots ,\pp_n$ all process names in $\chC$}
		\end{cases}
	\\[10pt]\hline\\[-8pt]
		\begin{array}{r@{=}l}
			\EPP[\pp_i]{\chnil}=\pnil
		\qquad\qquad
			\EPP[\pp_i]{\chS}
			&
			\begin{cases}
				\EPP[\pp_i]{\chS} & \mbox{if $\chS$ is a \chor}
				\\
				\EPP[\pp_i]{\proloc{\pp_i}{S_i}} & \mbox{if
				$\chS=\prod_{i=1}^{n}\proloc{\pp_i}{S_i}$ is a \proc}
			\end{cases}
		\\
			\EPP[\pp_i]{\gencom \chC}
			&
			\begin{cases}
				\psend k x \EPP[\pp_i]\chC	&\mbox{if } \pp_i=\pp
				\\
				\precv k y \EPP[\pp_i]\chC	&\mbox{if } \pp_i=\pq
				\\
				\EPP[\pp_i]\chC				&\mbox{if }\pp_i\notin\set{\pp,\pq}
			\end{cases}
		\\
			\EPP[\pp_i]{\genchois}
			&
			\begin{cases}
				\plsend {k}{\lab : \EPP[\pp_i]{\chC_\lab}}{\lab\in L}
				&
				\mbox{if } \pp_i=\pp
				\\
				\plrecv {k}{
					\begin{array}{l|@{\;}l}
						\lab : \EPP[\pp_i]{\chC_\lab} 	& \lab\in L
						\\
						\lab : S_\lab					& \lab\in L'\setminus L
					\end{array}
				}{}
				&
				\mbox{if } \pp_i=\pq
				\\
				\Merge_{\lab \in L} \EPP[\pp_i]{\chC_\lab}
				&
				\mbox{if }\pp_i\notin\set{\pp,\pq}
			\end{cases}
		\end{array}
	\\\\[-5pt]\hline\hline\\[-5pt]
		\begin{array}{c}
			\big(\pnus x \left(\prod_{i=1}^{n}\proloc{\pp_i}{S_i} \right) \big)
			\Merge
			\big(\pnus x \left(\prod_{i=1}^{n}\proloc{\pp_i}{T_i} \right)\big)
			=
			\pnus x \left( \prod_{i=1}^{n}\proloc{\pp_i}{S_i \merge T_i} \right)
		\\[10pt]
			\pnil \merge \pnil = \pnil
			\qquad
			(\psend xy T) \Merge (\psend xy S) = \psend xy T \merge S
			\qquad
			(\precv xy T) \Merge (\precv xy S) = \precv xy T \merge S
		\\[8pt]
			\left(\plrecv x{\lab : P_\lab}{\lab\in L} \right)
			\Merge
			\left(\plrecv x{\lab : Q_{\lab}}{\lab \in L'} \right)
			=
			x \lcoseq
			\left(
				\set{\lab : P_{\lab} \merge Q_{\lab}}_{\lab \in L \cap L'} \cup
				\set{\lab : P_\lab}_{\lab \in L \setminus L'}
				\cup
				\set{\lab : Q_{\lab}}_{\lab \in L' \setminus L}
			\right)
		\\[8pt]
			\mbox{ only if $L=L'$ : }
			\left(\plsend x{\lab : P_\lab}{\lab\in L} \right)
			\Merge
			\left(\plsend x{\lab : Q_{\lab}}{\lab \in L'} \right)
			=
			x \lseq
			\set{\lab : P_{\lab} \merge Q_{\lab}}_{\lab \in L}
		\end{array}
	\end{array}$}
	\caption{Endpoint Projection for \cflat \chors, and the merge operator ($\merge$).}
	\label{fig:EPP}
\end{figure}

\begin{nota}\label{def:chOrd}
	If $P$ and $Q$ are \eprocs, we write $P\chOrd Q$ iff $P \merge Q = P$.
\end{nota}
\begin{restatable}{theorem}{EPPthm}\label{thm:choreo}
	Let $\chC$ be a \prochor \cflat \chor.
	\begin{itemize}
		\item\label{thm:choreo.item1} \textbf{Completeness}: if $\chC\osto{\;\tm\;}\chC'$, then $\EPP\chC\osto{\;\tm\;}P \chOrd \EPP{\chC'}$;
		\item\label{thm:choreo.item2} \textbf{Soundness}:
		if $P\chOrd\EPP{\chC}$ and $P\osto{\;\tm\;} P'$, then there is a \chor $\chC'$ such that $\chC\osto{\;\tm\;}\chC'$ and $P'\chOrd\EPP{\chC'}$.
	\end{itemize}
\end{restatable}
\begin{proof}
	The proof is obtained by adapting the proof provided in, e.g., \cite{M13:phd,montesi:book,mon:yos:compositional} to the
	language we consider in this paper. Details are given in \Cref{app:proofs}.
\end{proof}

\subsection{A Sequent Calculus for the Endpoint Calculus}\label{subsec:endpoint}

To establish the correspondence between proofs of (formulas encoding) \eprocs and \chors, we enrich the syntax of formulas by adding labels (on sub-formulas) carrying the same information of the process names used in the syntax of the endpoint calculus.
More precisely, we consider a translation $\fofp{\cdot}$ from \eprocs to \defn{annotated formulas} of the form $\named A\pp$.
\begin{definition}
	For any \eproc $P=\pnus x \left(\prod_{i=1}^{n}\proloc{\pp_i}{S_i}\right)$ we define the formula $\fofp{P}=\lNu {x_1}{\dots \lNu {x_n}{\named{\fof{S_1}}{\pp_1}\lpar\cdots\lpar \named{\fof{S_n}}{\pp_n}}}$, and the sequent $\ffof P=\named{\fof{S_1}}{\pp_1},\ldots, \named{\fof{S_n}}{\pp_n}$.
\end{definition}
Note that because of the subformula property%
\footnote{
	Assuming an initial $\alpha$-renaming on $P$ such that $P$ is unambiguous, and such each variable bound by a receive action is the same as the unique (because of \rfreedom) variable sent by the matching send action.
	{For example, we would write $\Apsend xa \ppar \Aprecv xa$ instead of $\Apsend xa \ppar \Aprecv xb$.}
}
of rules in $\PIL$, such labelling can be propagated in a derivation by labelling each active formula of a rule (which is a sub-formula of one of the principal formulas of the rule) with the same process name of the corresponding active formula.

\begin{example}
	Consider the following derivation in $\PIL$ with conclusion the formula $\fofp{\proloc \pp{\psend xy \pnil} \ppar \proloc \pq {\precv xy \pnil}}$:
	\vspace{-10pt}
	\begin{equation}\label{eq:exampleNamed}
		\adjustbox{max width=.45\textwidth}{$\vlderivation{
			\vlin{\lpar}{}{
				\vdash
				\named[pzbrickred]{\lsend xy \lprec \lunit}{\pp}
				\lpar
				\named[pzgreen]{\lExp x{\lrecv xy \lprec \lunit}}{\pq}
			}{
				\vlin{\exists}{}{
					\vdash
					\named[pzbrickred]{\lsend xy \lprec \lunit}\pp,
					\named[pzgreen]{\lExp x{\lrecv xy \lprec \lunit}}\pq
				}{
					\vliin{\lprec}{}{
						\vdash
						\named[pzbrickred]{\lsend xy \lprec \lunit}\pp,
						\named[pzgreen]{\lrecv xy \lprec \lunit}\pq
					}{
						\vlin{\axrule}{}{
							\vdash
							\named[pzbrickred]{\lsend xy}\pp,
							\named[pzgreen]{\lrecv xy}\pq
						}{\vlhy{}}
					}{
						\vliin{\mixr}{}{\vdash \named[pzbrickred]\lunit\pp, \named[pzgreen]\lunit\pq}{
							\vlin{\lunit}{}{\vdash \named[pzbrickred]\lunit\pp}{\vlhy{}}
						}{
							\vlin{\lunit}{}{\vdash \named[pzgreen]\lunit\pq}{\vlhy{}}
						}
					}
				}
			}
		}$}
	\end{equation}
\end{example}

\begin{figure}[t]
	\centering
	\adjustbox{max width=\textwidth}{$\begin{array}{c}
		\vlinf{\Crestr}{m\in \N}{
			\vdash
			\lNu {x_1}{\dots\lNu{x_m}{\left(
					\named[pzbrickred]{\fof {T_1}}{\pp_1}
					\lpar \cdots\lpar
					\named[pzgreen]{\fof {T_n}}{\pp_n}
					\right)}}
		}{
			\vdash \named[pzbrickred]{\fof{T_1}}{\pp_1},\ldots,\named[pzgreen]{\fof {T_n}}{\pp_n}
		}
		\qquad\qquad
		\vlinf{\Ccomr}{
		}{
			\vdash \Gamma,
			\named[pzbrickred]{\fof{\psend xy T}}{\pp},
			\named[pzgreen]{\fof{\precv xz T'}}{\pq}
		}{
			\vdash \Gamma,
			\named[pzbrickred]{\fof{T}}{\pp},
			\named[pzgreen]{\fof{T'\fsubst yz}}{\pq}
		}
		\\[5pt]
		\vlinf{\Caxr}{}{\vdash \named{\lunit}{\pp_1},\ldots, \named{\lunit}{\pp_n}}{}
		\qquad\qquad
		\vlinf{\Cchor}{L \subseteq L'}{
			\vdash
			\Gamma,
			\named[pzbrickred]{\fof{\plsend{k}{\lab:T_\lab}{\lab\in L}}}{\pp} , \named[pzgreen]{\fof{\plrecv{k}{\lab:T'_\lab}{\lab\in L'}}}{\pq}
		}{
			\left\{\vdash
			\Gamma,
			\named[pzbrickred]{\fof{T_\lab}}{\pp} , \named[pzgreen]{\fof{T'_\lab}}{\pq}
			\right\}_{\lab\in L}
		}
	\end{array}$}
	\caption{Sequent calculus rules for the system $\ChorL$.}
	\label{fig:ChorL}
\end{figure}

\begin{figure}[t]
	\centering
	\adjustbox{max width=.95\textwidth}{$\begin{array}{c|c}
		\Caxr&
		\dD_{n}=\vlderivation{
			\vliiiq{\mixr}{}{\vdash
				\named{\unit}{\pp_1} , \dots , \named{\unit}{\pp_n}
			}{
				\vlin{}{}{\vdash
				\named{\unit}{\pp_1}
				}{\vlhy{}}
			}{\vlhy{\cdots}}{
				\vlin{}{}{\vdash
				\named{\unit}{\pp_n}
				}{\vlhy{}}
			}
		}
	\\\\[-5pt]\hline\\[-10pt]
		\Ccomr&
		\dD_{\tuple{\pp,x,\pq,y,k}}=
		\vlderivation{
			\vlin{\exists}{}{
				\vdash
				\Gamma
				,
				\named[pzbrickred]{\lsend{
					k
					}{x}\lprec\fof{S}}{\pp}

				,
				\named[pzgreen]{\lExp y{\lrecv{
					k}{y}\lprec\fof{S'}
					}}{\pq}
			}{
				\vliin{\lprec}{}{
					\vdash
					\Gamma
					,
					\named[pzbrickred]{\lsend{
						k}{x}\lprec\fof{S}}{\pp}

					,
					\named[pzgreen]{\lrecv{
						k}{x}\lprec \left(\fof{S'}
						\fsubst{x}{y}\right)}{\pq}
				}{
					\vlin{\axrule}{}{
						\vdash \named[pzbrickred]{\lsend kx}{\pp},\named[pzgreen]{\lrecv kx}{\pq}
					}{\vlhy{}}
				}{
					\vlhy{
						\qquad
						\vdash
						\Gamma,
						\named[pzbrickred]{\fof{S}}{\pp}
						,
						\named[pzgreen]{\fof{S'}
						\fsubst{x}{y}}{\pq}
					}
				}
			}
		}
	\\\\[-10pt]\hline\\[-5pt]
		\Cchor&
		\dD_{\tuple{\pp,L,\pq,L',k}}=
		\vlderivation{
		  	\vlin{\lwith}{}{
		  		\vdash
		  		\Gamma,
		  		\named[pzbrickred]{\biglbra[\ell\in L]{\left(\lsend{k}\ell\lprec\fof{S_{\ell}} \right)}}{\pp}, \named[pzgreen]{\biglsel[\ell \in L']{\left(\lrecv{k}\ell\lprec\fof{S'_{\ell}} \right)}}{\pq}
			}{\multiprem{
					\vliin{\lprec}{}{\vdash
					\Gamma,
					\named[pzbrickred]{\lsend{k}\ell\lprec\fof{S_{\ell}}}{\pp}, \named[pzgreen]{\lrecv{k}{\ell}\lprec\fof{S'_{\ell}}}{\pq}
					}{
						\vlin{\axrule}{}{
							\vdash \named[pzbrickred]{\lsend k\lab}{\pp} , \named[pzgreen]{\lrecv k\lab}{\pq}
						}{\vlhy{}}
					}{
						\vlhy{\vdash \named[pzbrickred]{\fof{S_\ell}}{\pp}, \named[pzgreen]{\fof{S'_{\ell}}}{\pq}}
					}
				}{\lab\in L}
	  		}
		}
	\\\\[-10pt]\hline\\[-5pt]
		\Crestr&
		\dD_{\emptyset}= \vlderivation{
			\vliq{\lpar}{}{\vdash
				\fof{S_1}\lpar \ldots\lpar\fof {S_n}
			}{
				\vlhy{\vdash
					\fof{S_1}, \ldots, \fof {S_n}
				}
			}
		}
		\quor
		\dD_{\tuple{x_1, \dots, x_m}}= \vlderivation{
			\vliq{\nurule}{}{\vdash
				\lNu{x_1}{\dots\lNu{x_m}{\left(
					\fof{S_1}
					\lpar \ldots \lpar
					\fof {S_n}
				\right)}}
			}{
				\vliq{\lpar}{}{\vdash
					\fof{S_1}
					\lpar \ldots\lpar
					\fof {S_n}
				}{
					\vlhy{\vdash
						\fof{S_1}
						, \ldots,
						\fof {S_n}
					}
				}
			}
		}
	\end{array}$}
	\caption{Derivability (in $\PIL$) of the rules in $\ChorL$.}
	\label{fig:derChorL}
\end{figure}

We now introduce a sequent calculus for the endpoint calculus, given in~\Cref{fig:ChorL}, which consists purely of rules that are derivable in $\PIL$. That is, for every rule $\rrule$ in $\ChorL$ there is an open derivation in $\PIL$ with the same open premises and conclusion as $\rrule$.

\begin{lemma}\label{lem:ChorLderiv}
	Each rule in $\ChorL$ is derivable in $\PIL$.

	Moreover, if the premise \resp{conclusion} of a rule in $\ChorL$ is of the form $\fofp{P}$, then its conclusion is so \resp{all its premise are so}.
\end{lemma}
\begin{proof}
	Derivations with the same premise(s) and conclusion of a rule in $\ChorL$ are shown in \Cref{fig:derChorL}.
	The second part of the statement follows by rule inspection.
\end{proof}

We can refine \Cref{thm:deadlock} for \rfree \fprocs thanks to the fact that in \eprocs parallel and restrictions can only occur at the top level.
This allows us to consider derivations in $\PIL$ made of blocks of rule applications as those in \Cref{fig:derChorL}, each corresponding to a single instance of a rule in $\ChorL$.

\begin{theorem}\label{thm:deadlockfreeChor}
	A \rfree \eproc $P$ is \lfree iff $\proves[\ChorL] \fofp P$.
\end{theorem}
\begin{proof}
	If $P=\pnus x \left(\prod_{i=1}^{n}\proloc{\pp_i}{S_i}\right)$ is \lfree,
	then by \Cref{thm:deadlock} there is a derivation $\widetilde\dD$ in $\PIL$ with conclusion
	$\fofp P$.
	Using rule permutations,
	we can transform $\widetilde{\dD}$ into a derivation made (bottom-up) of possibly some $\nurule$-rules followed by $\lpar$-rules (i.e., an open derivation of the same shape of $\dD_{\tuple{x_1,\ldots,x_k}}$ or $\dD_{\emptyset}$ from \Cref{fig:derChorL}), followed by a derivation $\dD$ of the sequent $\fofp {P_1},\ldots, \fofp{P_n}$ which is organized in blocks of rules as open derivations in \Cref{fig:derChorL}.
	Note that derivations of the form $\dD_{\tuple{\pp,x,\pq,y,k}}$ \resp{$\dD_{\tuple{\pp,L,\pq,L',k}}$} correspond to the open derivations
	in \Cref{eq:blocks}.
	Thus we can replace $\dD_{\tuple{x_1,\ldots,x_k}}$ or $\dD_{\emptyset}$  by a $\Crestr$,
	each
	$\dD_{\tuple{\pp,x,\pq,y,k}}$ \resp{$\dD_{\tuple{\pp,L,\pq,L',k}}$}
	by a
	$\Ccomr$ \resp{by a $\Cchor$},
	and
	each $\dD_n$ with a $\Caxr$,
	obtaining the desired derivation in $\ChorL$.

	The converse is a consequence of \Cref{lem:ChorLderiv} and \Cref{thm:deadlock}.
\end{proof}

\begin{figure}[t]\centering
	\adjustbox{max width=\textwidth}{$\begin{array}{c}
		\chCof[\vlinf{\Caxr}{}{\vdash
			\named{\unit}{\pp_1}, \dots, \named{\unit}{\pp_n}
		}{}
		]=\chnil
	\qquad
		\chCof[\vlderivation{
			\vlin{\Crestr}{}{
				\vdash \fofp{P}
			}{
				\vlpr{\dD'}{}{
					\vdash \ffof{P}
				}
			}
		}]=
		\begin{cases}
			\cnus x \chCof[\vlderivation{
				\vlpr{\dD'}{}{
					\vdash \ffof P
				}
			}]
			&\mbox{if }k>0
		\\
			\chCof[\vlderivation{
				\vlpr{\dD'}{}{
					\vdash \ffof P
				}
			}]
			&\mbox{if }k=0
		\end{cases}
	\\\\
		\chCof[\vlderivation{\vlin{\Ccomr}{}{
			\vdash \Gamma, \named[pzbrickred]{\lsend{k}{x}\lprec\fofp{T}}{\pp}
					, \named[pzgreen]{\lExp y{\lrecv{k}{y}\lprec\fofp{S}}}{\pq}
			}{
				\vlpr{\dD'}{}{
					\vdash \Gamma,
					\named[pzbrickred]{\fofp{T}}{\pp},
					\named[pzgreen]{\fofp{S}\fsubst{x}{y}}{\pq}
				}
			}
		}]=
		\gencom\chCof[\vlderivation{\vlpr{\dD'}{}{\vdash \Gamma \named[pzbrickred]{\fofp{T}}{\pp}, \named[pzgreen]{\fofp{S}\fsubst{x}{y}}{\pq} }}]
	\\\\
		\chCof[
			\vlderivation{
				\vlin{\Cchor}{}{
					\vdash
					\Gamma,
					\named[pzbrickred]{\biglbra[\ell\in L]{\left(\lsend{k}\ell\lprec\fofp{T_{\ell}} \right)}}{\pp},
					\named[pzgreen]{\biglsel[\ell \in L']{\left(\lrecv{k}\ell\lprec\fofp{T'_{\ell}} \right)}}{\pq}
				}{
					\multiprem{
						\vlpr{\dD_{\lab}}{}{\vdash
							\Gamma,
							\named[pzbrickred]{\fofp{T_\ell}}{\pp},
							\named[pzgreen]{\fofp{T'_\ell}}{\pq}
						}
					}{\lab\in L}
				}
			}
		]
		=
		\specchoics{p}{q}{k}{L}{\lab}{
			{\chCof[\vlderivation{
					\vlpr{\dD_{\lab}}{}{\vdash
						\Gamma,
						\named[pzbrickred]{\fofp{T_\ell}}{\pp},
						\named[pzgreen]{\fofp{T'_\ell}}{\pq}
						}
				}]}
		}{
			T'_\lab
		}
	\end{array}$}
	\caption{Interpretation of a derivation in $\ChorL$ as a \chor.}
	\label{fig:DerToChor}
\end{figure}

\subsection{Proofs as \Chors}\label{subsec:proofsASchors}

We can now prove a completeness result of \chors with respect to the set of \lfree \fprocs{}: each \lfree \fproc is the EPP of a (\cflat) \chor.
To prove this result, we rely on \Cref{thm:deadlockfreeChor} to establish a direct correspondence between derivations of a sequent encoding a \rfree \eproc in the sequent system $\ChorL$, and \comptrees of the same \eproc.
In \Cref{app:example} we provide an example the process of \chor extraction from a derivation in $\ChorL$ of a \lfree \eproc.

\begin{restatable}{theorem}{thmChor}\label{thm:proofsASchoreo}
	Let $P$ be a \rfree \eproc.
	Then
	$$
	\mbox{$P$ is  \lfree}
	\iff
	\mbox{there is a \chor $\chC$ such that $\EPP \chC = P$.}
	$$
\end{restatable}
\begin{proof}
	If $P$ is \lfree, then by \Cref{thm:deadlockfreeChor} there is a derivation in $\ChorL$ with conclusion $\fofp P$.
	We define the \chor $\chCof[\dD]$ by case analysis on the bottom-most rule $\rrule$ in $\dD$ as shown in \Cref{fig:DerToChor}.
	We conclude by showing that $\EPP{\chCof[\dD]}=P$ by induction
	on the structure of $\dD$ reasoning on the bottom-most rule $\rrule$ in $\dD$.
	The rigth-to-left implication follows by \Cref{thm:choreo}.
\end{proof}
\begin{remark}
	Note that the statement could be made stronger by requiring the \chor $\chC$ to be \flat. In this case, the proof of the right-to-left implication can be proven directly using the inverse of the translation in \Cref{fig:DerToChor}.
\end{remark}

Since every \fproc in the \picalc can be decorated with process names, we can easily extend the completeness result to the \picalc.
We need to pay attention to the difference that $P\strews\EPP{\chC}$ (instead of $P=\EPP{\chC}$) because of the definition of $\EPP\chnil$.
For example, the choreography that captures $\pnil\ppar\pnil$ is $\chnil$, and $\pnil\ppar\pnil\strews\EPP\chnil=\pnil$.
\begin{corollary}\label{cor:flatChor}
	Every \rfree and \lfree \fproc $P$ admits a \chor $\chC$ such that $P\strews\EPP{\chC}$.
\end{corollary}

An important consequence of our results is that the processes that can be captured by choreographies can have cyclic dependencies, as we exemplified with \Cref{eq:intro,eq:choreo}.
This significantly extends the proven expressivity of choreographic languages with name mobility, which so far have been shown to capture only the acyclic processes typable with linear logic~\cite{CMS18,carbone:multiparty_choreo}.

\section{Related Work}\label{sec:rel}

\def\ok{\text{\ding{51}}}
\def\ko{\text{\ding{55}}}
\def\okbutko{\ok \mbox{ but }\ko}
\def\kobutok{\ko  \mbox{ but }\ok}
\def\NA{\mbox{N/A}}
\def\mezzo{\sim}
\begin{figure}[t]
	\centering
	\adjustbox{max width=\textwidth}{$\begin{array}{l|c|c|c|c|c}&
			\begin{tabular}{c}
				Independent\\
				$\nu$ and $\ppar$
			\end{tabular}
			&
			\begin{tabular}{c}
				Cyclic\\
				dependencies
			\end{tabular}
			&
			\begin{tabular}{c}
				$\nu$-free\\
				interaction
			\end{tabular}
			&
			\begin{tabular}{c}
				Choreography\\
				expressivity
			\end{tabular}
			&
			\begin{tabular}{c}
				Proof\\
				System
			\end{tabular}
			\\\hline
			\mbox{Caires \& Pfenning \cite{cai:pfe:ST}}&
			\ko	&	\ko	&	\ko	& \NA & \iLL
			\\\hline
			\mbox{Wadler \cite{wadler:PaS}}&
			\ko	&	\ko	&	\ko	& \NA & \LL
			\\\hline
			\mbox{Dardha \& Gay \cite{dardha:gay}}&
			\ok	&	\ok	&	\ko	& \NA & \LL+\mbox{mix}+\mbox{ordering}
			\\\hline
			\mbox{Kokke \& Montesi \& Peressotti \cite{kok:mon:per:better}}&
			\ok	&	\ko	&	\ok	& \NA & \LL+\mbox{hyperenvironments}
			\\\hline
			\mbox{Carbone \& Montesi \& Sch\"urmann \cite{CMS18}}&
			\ko	&	\ko	&	\ko	& \ok & \LL
			\\\hline
			\mbox{This paper}&
			\ok&\ok&\ok&\ok& \PIL
		\end{array}$}
	\caption{
		Summary of key results in the literature.
		We describe each column in order: the term constructors for restriction and parallel are separate (independent) in the syntax of processes; cyclic dependencies are allowed; processes can interact (communicate) on a free name (the name does not need to be restricted in the context); choreographies are proven complete for a class of processes (\NA{} means that this was not considered).
	}
	\label{fig:SPQR}
\end{figure}

We now report on relevant related work. A summarising table of the differences between our work and others based on logic is given in~\Cref{fig:SPQR}.

\myparagraph{Proofs as Processes: Linear Logic and Session Types}
The proofs-as-processes agenda investigates how linear logic~\cite{gir:ll} can be used to reason about the behaviour of concurrent processes~\cite{abramsky1994proofs,bellin:scot:pi}.
It has inspired a number of works that aim at preventing safety issues, like processes performing incompatible actions in erroneous attempts to interact (e.g., sending a message with the wrong type).
Notable examples include \emph{session types}~\cite{honda:dyadic,hon:vas:kub} and linear types for the $\pi$-calculus~\cite{DBLP:journals/toplas/KobayashiPT99}. The former can actually be encoded into the latter -- a formal reminder of their joint source of inspiration~\cite{DGS17}.

A more recent line of research formally interprets propositions and proofs in linear logic as, respectively, session types and processes~\cite{cai:pfe:ST,wadler:PaS}.
This proofs-as-processes correspondence based on linear logic works for race-free processes, as we consider here.
However, it also presents some limitations compared to our framework.
Parallel composition and restriction are not offered as independent operators, because of a misalignment with the structures given by the standard rules of linear logic. For example, the $\cutr$ rule in linear logic handles both parallel composition and hiding, yielding a `fused' restriction-parallel operator $\pnu x ( - \ppar - )$.
Also, the $\otimes$ rule for typing output has two premises, yielding another fused output-parallel operator $x!(y).( - \ppar - )$ -- note that only bound names can be sent, as in the internal $\pi$-calculus~\cite{S95}.
In particular, interaction between processes does not arise simply from parallel composition as in the standard \picalc, but rather requires both parallel composition and restricting all names on which communication can take place (so communication is always an internal action).
This syntactic and semantic gap prevents linear logic from typing safe cyclic dependencies among processes, as in this simplification of \Cref{eq:intro}:
\begin{equation}\label{eq:S}
	\pnu x\pnu y\left(\psend xa \precv yb \pnil \ppar \precv xa \psend yb \pnil\right)
\end{equation}
The same gap prevents having communication on unrestricted channels (as in $\psend xa \pnil \ppar \precv xa \pnil$) and having a private channel used by more than two \procs.
Using $\lpar$ and $\ltens$ to type input and output is also in tension with the associativity and commutativity isomorphisms that come with these connectives. These isomorphisms yield unexpected equivalences at the process level, like $\precv xa \precv yb \pnil \equiv \precv yb \precv xa \pnil$.

These shortcomings are not present in our approach, thanks to the use of:
\begin{enumerate}
	\item The connective $\lprec$ for prefixing. The latter then has the expected `rigid' non-commutative and non-associative semantics.
	\item The connective $\lpar$ for parallelism. The latter then has the expected equivalences supported by the isomorphisms for $\lpar$.
	\item Nominal quantifiers, which allow for restricting names without imposing artificial constraints on the structure of processes.
\end{enumerate}

While it is not the first time that these limitations are pointed out, our method is the first logical approach that overcomes them without ad-hoc machinery.
Previous works have introduced additional structures to linear logic, like hyperenvironments or indexed families of connectives, in order to address some of these issues \cite{dardha:gay,torres:vasconcelos,padovani:deadlock,kok:mon:per:better,montesi:wadler:al:coherence,DBLP:journals/acta/CarboneMSY17}.
These additional structures are not necessary to our approach.

\myparagraph{Choreographic Programming}
%
Choreographic programming was introduced in~\cite{M13:phd} as a paradigm for simplifying concurrent and distributed programming.
Crucial to the success of this paradigm is building choreographic programming languages that are expressive enough to capture as many safe concurrent behaviours as possible~\cite{montesi:book}.
However, most of the work conducted so far on the study of such expressivity is driven by applications, and a systematic understanding of the classes of processes that can be captured in choreographies is still relatively green.

In \cite{montesi:cruz:al:choreo_extraction,montesi:cruz:extraction_spwaning} the authors present methods for \emph{choreography extraction} -- inferring from a network of processes an equivalent choreography -- for an asynchronous process calculus, respectively without and with \proc spawning.
Another purely algorithmic extraction procedure is provided in \cite{yoshida:al:extraction_mpst}, for simple choreographies without data -- global types, which are roughly the choreographic equivalent of session types.
In linear logic, extracting global types from session types can be achieved via derivable rules \cite{montesi:wadler:al:coherence}.

The only previous completeness result for the expressivity of choreographic programming is given in~\cite{CMS18,carbone:multiparty_choreo}, where it is shown that choreographies can capture the behaviours of all well-typed processes in linear logic.
Our work extends the completeness of choreographies to processes that, notably, can (i) have cyclic dependencies (like ~\Cref{eq:S}), (ii) perform communication over free channels, (iii) respect the sequentiality of prefixing.
Moreover, and similarly to our previous discussion for proofs-as-processes, extraction in~\cite{CMS18,carbone:multiparty_choreo} requires additional structures (hypersequents and modalities to represent connections) that are not necessary in our method.

\myparagraph{Non-Commutative Logic and Nominal Quantifiers}
%
Guglielmi proposed in \cite{guglielmi:concurrecy,guglielmi:95:sequentiality} an extension of multiplicative linear logic with a non-commutative operator modelling the interaction of parallel and sequential operators.
This led to the design of the \emph{calculus of structures} \cite{gug:SIS}, a formalism for proofs where inference rules can be applied at any depth inside a formula rather than at the top-level connective, and the logic $\BV$ including the (associative) non-commutative self-dual connective $\lseq$ to model sequentiality.
In \cite{bru:02} Bruscoli has established a computation-as-deduction correspondence between specific derivations in $\BV$ and executions in a simple fragment of $\CCS$.
This correspondence has been extended in the works of Horne, Tiu et al. to include the choice operator ($+$) of Milner's $\CCS$ (modelled via the additive connective $\lplus$), as well as the restriction to model private channels in the \picalc \cite{hor:tiu:tow,hor:nom}.
In these works, restriction has been modelled via \emph{nominal quantifiers} in the spirit of the ones introduced by Pitts and Gabbay for \emph{nominal logic} \cite{pitts:nominal,gabbay:pitts:nominal}, by  considering a pair of dual quantifiers%
\footnote{
	{In \cite{hor:tiu:19} the authors report the use of a non-self-dual quantifier to model restriction was suggested them by Alessio Guglielmi in a private communication.}
	As explained in detail in \cite{acc:man:NL}, the pair of dual nominal quantifiers in \cite{hor:tiu:19,hor:nom,hor:tiu:tow} is not the same pair we consider in this paper.
	This can be observed by looking at the implication $\left(\lNup x A\ltens \lNup x B\right) \limp \lNup x{A\lpar B}$, which is valid in these works, but it is not valid in $\PIL$.
	For this reason we adopted a different symbol for the dual quantifier of $\lnewsymb$ -- i.e., our $\lyasymb$ instead of their $\lwensymb$.
},
instead of a single self-dual quantifier as in \cite{menni:nominal,rov:bind,mil:tiu:nabla}.

The logic $\PIL$ we use as logical framework to establish our correspondences in this paper takes inspiration from Bruscoli's work and its extension, but it uses a non-associative non-commutative self-dual connective $\lprec$ instead of the $\lseq$ from $\BV$.
This seemingly irrelevant difference (the non-associativity of $\lprec$) guarantees the existence of a $\cutr$-free sequent calculus to be used as a framework for our correspondence, while for the logic $\BV$ and its extension $\cutr$-free sequent calculi cannot exist, as proven in \cite{tiu:SIS-II}.
Note that requiring non-associativity for the connective modelling sequentiality is not a syntactical stretch, because the same restriction naturally occurs in process calculi such as $\CCS$ and the \picalc, where sequentiality is defined by an asymmetric prefix operation only allowing to sequentially compose (on the left) atomic instructions, such as send and receive.
The other main difference is that in these works derivations represent a single execution, while our derivations represent \comptrees.
This allows us to state our \Cref{thm:deadlock} without quantifying on the set of derivations of $\fof P$.

\section{Conclusion and Future Works}\label{sec:conc}

We presented a new approach to the study of processes based on logic, which leverages an interpretation of processes in the \picalc as formulas in the proof system $\PIL$.
By seeing derivations as computation trees, we obtained an elegant method to reason about \lfreedom that goes beyond the syntactic and semantic limitations of previous work based on logic.
This led us to establishing the first completeness result for the expressivity of choreographic programming with respect to mobile processes with cyclic dependencies.

We discuss next some interesting future directions.

\myparagraph{Recursion}\label{sec:rec}%
Recursion could be modelled by extending $\PIL$ with fixpoint operators and rules like the ones in \cite{BaeldeM07,bae:dou:sau:16,acc:cur:gue:CSL24,acc:cur:gue:CSL24ext,acc:mon:per:OPDL}.
We foresee no major challenges in extending the proof of cut-elimination for $\muMALL$ to $\PIL$,
since the behaviour of the connective $\lprec$ is purely multiplicative (in the sense of \cite{dan:reg:MLL,gir:mean,acc:mai:CSL20}) and the rules for nominal quantifiers do not require the employment of new techniques.
In $\PIL$ with recursion, properties such as \emph{justness} or \emph{fairness} could be characterised by specific constraints on derivations, corresponding to constraints on threads and paths of the (possibly cyclic) \comptrees.

\myparagraph{Asynchronous \picalc}%
We foresee the possibility of modelling asynchronous communication by including shared buffers, inspired by previous work on concurrent constraint programming~\cite{vij:mar:CCP,ola:rue:val:models} and its strong ties to logic programming~\cite{montanari:74,shapiro:concLP,jaf:mah:CLP,linear:CCP,ola:pim:concurrentFocussing,subexp,pimentel:procAsFor}.
However, buffers with capacity greater than $2$ have non-sequential-parallel structures and therefore cannot be described efficiently using binary connectives~\cite{Valdes1979,cographs}.
We may thus need to consider \emph{graphical connectives}~\cite{acc:IJCAR24,acc:FSCD22}.

\myparagraph{Proof Nets}
%
In \cite{acc:man:NL} we define proof nets for $\PIL$, capturing local rule permutations, and providing canonical representative for \comptrees up-to interleaving concurrency.
This syntax could be used to 
refine the correspondence between proofs and \chors (\Cref{thm:proofsASchoreo}).
We plan to study the extension of the computation-as-deduction paradigm in the case of proof net expansion, following the ideas in \cite{andreoli:focPN,and:maz:PN,acc:mai:DCM}, as well as to use a notion of orthogonality for modules of proof nets (in the sense of \cite{dan:reg:MLL,and:maz:PN,acc:mai:DCM}) to study testing preorders \cite{den:hen:testing,den:hen:CCStau,hennessy1988algebraic,ber:hen:testing}.

\myparagraph{Completeness of Choreographies}
%
The literature of choreographic programming languages includes features of practical interest that extend the expressivity of choreographies -- like process spawning~\cite{CM17} and nondeterminism~\cite{montesi:book}.
Exploring extensions of $\PIL$ to capture these features is interesting future work.
For process spawning, a simple solution could be achieved by defining a way to dynamically assign process names to properly define the map $\chCof$.

\subsubsection*{Acknowledgements.}
Partially supported by Villum Fonden (grants no. 29518 and 50079). 
Co-funded by the European Union’s Horizon 2020 research and innovation program under the Marie Sklodowska-Curie grant agreement No 945332.
Co-funded by the European Union (ERC, CHORDS, 101124225). Views and opinions expressed are however those of the authors only and do not necessarily reflect those of the European Union or the European Research Council. Neither the European Union nor the granting authority can be held responsible for them.

%
%
\bibliographystyle{splncs04}
\bibliography{biblio}
%

\clearpage

\appendix

\section{Additional Proofs}\label{app:proofs}

\begin{figure}[t]
	\centering
	\adjustbox{max width=.94\textwidth}{$\begin{array}{r@{\;=\;}l}
		z\fsubst xy
		&
		\begin{cases}z &\mbox{ if }z\neq y\\x &\mbox{ if }z=y\end{cases}
	\\
		(P \ppar Q)\fsubst xy
		&
		P \fsubst xy \ppar Q \fsubst xy
	\\
		(\pnu z P)\fsubst xy
		&
		\pnu z P\fsubst xy
	\\
		(\psend zw P)\fsubst xy
		&
		\begin{cases}
			\psend zw P\fsubst xy &\mbox{ if } z\notin\set{x,y}	\\
			\psend xw P\fsubst xy &\mbox{ if } y = z		\\
			\psend zx P\fsubst xy &\mbox{ if } y = w
		\end{cases}
	\\
		\left(\plsend z{\lab : P_\lab}{\lab\in L}\right)\fsubst xy
		&
		\begin{cases}
			\plsend z{\lab : P_\lab\fsubst xy}{\lab\in L} & \mbox{if } z \neq y\\
			\plsend x{\lab : P_\lab\fsubst xy}{\lab\in L} & \mbox{otherwise}
		\end{cases}
	\\
		(\precv zw P)\fsubst xy
		&
		\begin{cases}
			\precv zw P\fsubst xy &\mbox{ if } z\notin\set{x,y}	\\
			\precv xw P\fsubst xy &\mbox{ if } y= z				\\
			\precv zw P &\mbox{ if } y 			= w
		\end{cases}
	\\
		\left(\plrecv z{\lab : P_\lab}{\lab\in L}\right)\fsubst xy
		&
		\begin{cases}
			\plrecv z{\lab : P_\lab\fsubst xy}{\lab\in L} & \mbox{if } z \neq y\\
			\plrecv x{\lab : P_\lab\fsubst xy}{\lab\in L} & \mbox{otherwise}
		\end{cases}
	\end{array}$}
	\caption{Notation for substitution}
	\label{fig:subst}
\end{figure}

\maxcomptree*
\begin{proof}
	If $P\steq \pnil$ then $\comptreeof{P}$ is maximal.
	Otherwise, we assume there is a leaf $P'\not\steq$ of $\comptreeof{P}$.
	Since $P$ is \lfree, then there is a \proc $Q$ such that $P' \redsem Q$.
	If the core reduction of  $P' \redsem Q$ is a $\rscomr$ or a $\rslabr$, we extend $\comptreeof{P}$ grafting $Q$ in $P'$.
	If the core reduction of  $P' \redsem Q$ is a $\rschoir$, then there is a set $\set{Q_\lab \mid \lab \in L} \ni Q$ such that
	$P' \redsem Q_\lab$ for each $\lab \in L$. In this case $\comptreeof{P}$ can be extended grafting the processes in $\set{Q_\lab \mid \lab \in L}$ in $P'$.
	Since this reasoning can be iterated on any leaf $P\not\steq \pnil$, we conclude because no branch of the \comptree of a  \proc (without recursion) can be infinite.
\end{proof}

\contextDers*
\begin{proof}
	By induction on the structure of the contexts.
	\begin{enumerate}
		\item
		If $\proves[\PIsys] A\limp B$ and $\cC\ctx = \ctx$ the claim is trivial.
		Otherwise, we have a derivation starting with a $\lpar$-rule and continuing as follows:
		\begin{itemize}
			\item
			If $\cC\ctx\in \set{ \cC_1\ctx \lpar D,\cC_1\ctx \ltens D }$,
			we apply a $\lpar$-rule followed by a $\ltens$ rule splitting the context so that one of premise is $\vdash \cneg{\cC_1\ctx[A]}, \cC_1\ctx[B]$ and the other is $\vdash D, \cneg D$.
			The former is provable by inductive hypothesis, the last by \Cref{thm:cutelim}.\ref{cutelim:1}.

			\item
			If $\cC\ctx = \cC_1\ctx \lprec D $ we conclude similarly, applying a $\lprec$-rule instead of the $\ltens$-rule.
			\item
			If $\cC\ctx\in \set{\cC_1\ctx \lwith D, \cC_1\ctx \lplus D}$
			we apply a $\lpar$-rule followed by a $\lwith$-rule. Then on each branch we apply the $\oplus$-rule so that
			we obtain $\vdash\cneg{\cC_1\ctx[A]}, \cC_1\ctx[B] $ and  $\vdash D, \cneg D$.
			The former is provable by inductive hypothesis, the last by  \Cref{thm:cutelim}.\ref{cutelim:1}.
			\item
			If $\cC\ctx\in\set{\lNu{x}{\cC_1 \ctx},\lYa{x}{\cC_1 \ctx}}$, then we apply a $\nqsrule$-rule and we conclude by inductive hypothesis.
			\item
			If $\cC\ctx\in\set{\lFa{x}{\cC_1 \ctx}, \lEx{x}{\cC_1 \ctx}}$, then we first apply a $\forall$-rule, and then an $\exists$-rule where the variable $c$ is the same fresh variable used by the $\forall$-rule.
			We conclude by inductive hypothesis.
		\end{itemize}
		\item
		If $\cK\ctx=\chole$, then the result is trivial.
		Otherwise, $\cK\ctx ={(\chole \lpar C)}$ or $\cK\ctx=\lNus x{\left(\chole \lpar C\right)}$. We conclude since there are the following derivations
		$$
		\vlderivation{
			\vliq{\lpar}{}{
				\vdash
				\left(\bigoplus_{i=1}^n (\cneg A_i \ltens \cneg C)\right)
				\lpar
				(B \lpar C)
			}{
				\vlin{\oplus}{}{
					\vdash
					\bigoplus_{i=1}^n (\cneg{A_i} \ltens \cneg C)
					,
					B, C
				}{
					\vliin{\ltens}{}{
						\vdash
						\cneg A_i \ltens \cneg C
						,
						B, C
					}{
						\vlpr{}{\IH}{\cneg A_i,B}
					}{
						\vlpr{}{\text{\Cref{thm:cutelim}.\ref{cutelim:1}}}{\cneg C,C}
					}
				}
			}
		}
		\hskip2em
		\vlderivation{
			\vlin{\nqsrule}{}{
				\vdash
				\lYas{x}{\left(\bigoplus_{i=1}^n (\cneg A_i \ltens \cneg C)\right)}
				\lpar
				\lNus{x}{(B \lpar C)}
			}{
				\vliq{\lpar}{}{
					\vdash
					\left(\bigoplus_{i=1}^n (\cneg A_i \ltens \cneg C)\right)
					\lpar
					(B \lpar C)
				}{
					\vlin{\oplus}{}{
						\vdash
						\bigoplus_{i=1}^n (\cneg{A_i} \ltens \cneg C)
						,
						B, C
					}{
						\vliin{\ltens}{}{
							\vdash
							\cneg A_i \ltens \cneg C
							,
							B, C
						}{
							\vlpr{}{\IH}{\cneg A_i,B}
						}{
							\vlpr{}{\text{\Cref{thm:cutelim}.\ref{cutelim:1}}}{\cneg C,C}
						}
					}
				}
			}
		}
		$$
	\end{enumerate}
\end{proof}

\lemSimulation*
\begin{proof}
	\Cref{simulation:1} is proven using \Cref{prop:feq}
	and the transitivity of $\limp$ (see \Cref{thm:cutelim}.\ref{cutelim:3}).
	To prove \Cref{simulation:4} we reason by induction on the \entropy:
	\begin{itemize}
		\item if $\entof{P}{P'}=0$ then
		$P\redsem P'$ via $\rscomr$, $\rslabr$ or $\rschoir$.
		We conclude using the derivations in \Cref{fig:derivationspisys}.
		Note that, if $P\redsem P'$ via $\rschoir$, then $P=\plsend x{\lab: P_\lab}{\lab\in L}$. Therefore, there is a set of \procs $\set{P_\lab\mid\lab\in L}\ni P'$ such that $P\redsem P_\lab$ via $\rschoir$ for each $\lab\in L$;

		\item
		if $P\redsem P'$ via $\rsparr$ \resp{$\rsresr$},
		then there is a context $\cP\ctx=\left(\chole \ppar R\right)$ \resp{$\cP\ctx=\pnu x\chole$} such that
		$P=\cP\ctx[S]$  and  $P'=\cP\ctx[S']$.
		Then $\crof S{S'}=\crof P{P'}$ and:
		\begin{itemize}
			\item
			if the core-reduction of $S\redsem S'$ is a $\rscomr$ \resp{$\rslabr$},
			then we have
			$\proves[\PIL]\fof {S'}\limp \fof{ S}$
			by inductive hypothesis because $\entof{S}{S'}<\entof{P}{P'}$.
			We conclude by \Cref{thm:der}.\ref{der:2};

			\item
			otherwise,
			the core-reduction of $S\redsem S'$ is a $\rschoir$
			and,
			by inductive hypothesis,
			there is a set of \procs $\set{S_\lab\mid\lab\in L}\ni S'$
			such that
			$\entof S{S_\lab}<\entof P{P_\lab}$ for each $\lab\in L$
			and
			$\proves[\PIL] \left(\bigwith_{\lab\in L}\fof{S_\lab}\right) \limp \fof S$.
			\begin{itemize}
				\item
				If $P \redsem P'$ via $\rsparr$, then we conclude by
				\Cref{thm:der}.\ref{der:2} and \Cref{thm:cutelim}.\ref{cutelim:3}
				using the fact that
				$\proves[\PIL] \bigwith_{\lab\in L}(\fof{S_\lab} \lpar \fof{R}) \limp \bigwith_{\lab\in L}\fof{S_\lab} \lpar \fof{R}$
				(see \Cref{eq:logEqs}).
				\item
				If $P \redsem P'$ via $\rsresr$,
				then conclude by \Cref{thm:der}.\ref{der:2}
				and \Cref{thm:cutelim}.\ref{cutelim:3} since $\bigwith_{\lab\in L}\left(\lNu{x}{\fof{S_\lab}}\right) = \lNup{x}{\bigwith_{\lab\in L}\fof{S_\lab}}$ by \Cref{prop:feq}.
			\end{itemize}
		\end{itemize}

		\item
		if $P\redsem P'$ via $\rsstreqr$, then there is $S$ such that $P\strews S$ and $S\redsem P'$ (via a rule different from $\rsstreqr$).
		By \Cref{simulation:1} we have
		$
		\proves[\PIL]\fof{S} \limp \fof P$ .
		Let $(Q,Q')=\crof S{P'}=\crof P{P'}$.
		\begin{itemize}
			\item If $Q\redsem Q'$ via $\rscomr$ or $\rslabr$,
			then we conclude by \Cref{thm:cutelim}.\ref{cutelim:3}
			since
			$\proves[\PIL]\fof S\limp \fof P$ by \Cref{simulation:1}
			and
			$\proves[\PIL]\fof {P'}\limp \fof S$ by inductive hypothesis;

			\item
			otherwise,
			$Q\redsem Q'$ via $\rschoir$.
			In this case, using \Cref{eq:prestreq}, we can assume w.l.o.g. that there is a \netcont $\cN\ctx$ such that $P\strews \cN\ctx[Q]=S$
			and a set of \procs $\set{Q_\lab \mid \lab \in L}\ni Q'$ such that $Q\redsem Q_\lab$ via $\rschoir$.
			Thus,
			by inductive hypothesis, that $ \proves[\PIL] \left(  \bigwith_{\lab\in L}\fof{Q_\lab}\right)\limp \fof{ Q} $.
			Therefore, we have a set of \procs $\set{ P_\lab=\cN\ctx[Q_\lab] \mid \lab \in L} \ni P'$ such that $S \redsem P_\lab$ for each $\lab \in L$.
			Moreover, by \Cref{thm:der}\ref{der:3} we have
			$ \proves[\PIL] \left(  \bigwith_{\lab\in L}\fof{P_\lab}\right)\limp \fof{ S} $
			because
			$P_\lab=\cN\ctx[Q_\lab]$ and $S=\cN\ctx[Q]$.
			We conclude by \Cref{thm:cutelim}.\ref{cutelim:3} since
			$\proves[\PIL]\fof S\limp \fof P$ by \Cref{simulation:1}.

		\end{itemize}
	\end{itemize}
\end{proof}

\thmTrees*
\begin{proof}
	To prove the theorem we show a correspondence between maximal \comptrees and derivations in $\PIL$.

	\textbf{($\Rightarrow$)}
	If $P$ is \lfree, then, by \Cref{lem:completeTree}, any maximal \comptree $\comptreeof{P}$ with root $P$ has leaves which are processes structurally equivalent to $\pnil$.
	We reason on the structure of any such tree $\tree$.
	\begin{itemize}
		\item
		If the root $P$ is a leaf, then $P\strews \pnil$ because $P$ is \lfree.
		Then, by \Cref{lem:simulation}.\ref{simulation:1} there is a derivation $\dD$ of $\vdash \unit\limp \fof P$.
		We conclude by applying cut-elimination (\Cref{thm:cutelim}.\ref{cutelim:2}) to the derivation made by composing (via $\cutr$-rule) the derivation $\dD$ with the single-rule derivation made of a $\urule$.

		\item
		If $P$ has a unique child $P'$ and the core-reduction of $P\redsem P'$ is a $\rscomr$ or a $\rslabr$, then there is a derivation $\dD$ of $\vdash \fof{P'}\limp\fof P$ by \Cref{lem:simulation}.
		We conclude by applying cut-elimination (\Cref{thm:cutelim}.\ref{cutelim:2}) to the derivation made by composing (via $\cutr$-rule) the derivation $\dD$ with the derivation with conclusion $\fof{P'}$, which exists by inductive hypothesis.

		\item
		Otherwise, $P$ has set of children $\set{P_\lab \mid \lab\in L}$ and the core-reduction of $P\redsem P_\lab$ is a $\rschoir$ for each $\lab\in L$.
		Then there is a derivation $\dD$ of $\vdash \bigwith_{\lab\in L}\fof{P_\lab} \limp \fof P$ by \Cref{lem:simulation}.
		We conclude by applying cut-elimination (\Cref{thm:cutelim}.\ref{cutelim:2}) to the derivation made by composing (via $\cutr$-rule) the derivation $\dD$ with the derivation starting with a $\lwith$-rule with premises the conclusion of a derivation $\dD_\lab$ with conclusion $\fof{P_\lab}$, which exists by inductive hypothesis.
	\end{itemize}

	\textbf{($\Leftarrow$)}
	To prove the converse, we show that each derivation $\dD_P$ of $\fof P$ can be transformed using the \emph{rule permutations}\footnote{{
		These rule permuations are some of the rule permutations used for the proof of cut-elimination (see \cite{acc:man:NL})
	}} in \Cref{fig:permutations} into a derivation $\widetilde{\dD_P}$ of a specific shape, and then applying inductive reasoning on the size of $P$.
	Note that we here assume, w.l.o.g., that processes are written in unambiguous forms, therefore their corresponding formulas are {clean}.
	\begin{enumerate}
		\item
		If $P=\cN\Ctx[\psend xyQ]$ \resp{$P=\cN\Ctx[\precv xzQ]$} for a \netcont,
		then in $\dD_P$ there must be a $\exists$-rule renaming (bottom-up) each occurrence of $z$ in $\lrecv xz\lprec\fof Q$ into occurrences of $y$, and an $\axrule$-rule with conclusion $\vdash \lsend xy,\lrecv xy$ immediately above
		a $\lprec$-rule.
		That is, there is a derivation $\dD_P$ of $\fof P$ of the following shape below on the left.
		\begin{equation}\label{eq:deadlock:com}
			\hskip-2em\adjustbox{max width=.92\textwidth}{$\begin{array}{c}
					\begin{array}{ccccc}
						\vlderivation{
							\vlde{\dD_1}{}{\vdash \fof P}{
								\vlin{}{}{
									\vdash \lsend xy\lprec A,\lEx z\left(\lrecv xz \lprec B \right), \Gamma'
								}{
									\vlde{\dD_2}{}{
										\vdash \lsend xy\lprec A ,\lrecv xy \lprec B\fsubst yz, \Gamma'
									}{
										\vliin{\lprec}{}{
											\vdash \lsend xy\lprec A ,\lrecv xy \lprec B\fsubst yz, \Gamma'
										}{
											\vlin{\axrule}{}{\vdash \lsend xy,\lrecv xy}{\vlhy{}}
										}{
											\vlpr{\dD_3}{}{\vdash A , B\fsubst yz, \Gamma'}
										}
									}
								}
							}
						}
						&
						\rightsquigarrow
						&
						\vlderivation{
							\vlde{\dD^-_{(\alpha, \beta,\Gamma)}}{}{\vdash \fof P}{
								\vlin{\exists}{}{
									\vpz1{}\vdash \lsend xy\lprec A,\lEx z\left(\lrecv xz \lprec B \right), \Gamma
								}{
									\vliin{\lprec}{}{
										\vdash \lsend xy\lprec A ,\lrecv xy \lprec B\fsubst yz, \Gamma
									}{
										\vlin{\axrule}{}{\vdash \lsend xy,\lrecv xy}{\vlhy{}}
									}{
										\vlpr{\dD'}{}{
											\vdash A , B\fsubst yz, \Gamma
										}
									}
								}
							}
							\tikz[overlay,remember picture]{
								\draw[draw=yellow,line width = 4pt,opacity=.5] %
								(pz1) ++(-12pt,-8pt) -- %
								++(144pt,0pt) 	--
								++(0pt,50pt)	--
								++(-144pt,0pt)	--
								++(0pt,-50pt);
							}
						}
						&\mapsto&
						\vlderivation{
							\vlde{\dD^-_{(\alpha', \beta',\Gamma)}}{}{\vdash \fof {P'}}{
								\vlpr{\dD'}{}{
									\vdash A,B\fsubst yz, \Gamma
								}
							}
						}
					\end{array}
					\\
					\mbox{
						with $\alpha=\lsend xy\lprec A$ and $\beta=\lEx z\left(\lrecv xz \lprec B \right)$,
						and
						with $\alpha'=A$ and $\beta'=B\fsubst yz$
					}
				\end{array}$}
		\end{equation}
		Therefore, we can transform $\dD_P$ into a derivation of the shape above in the center, where the derivations $\dD^-$ are made of $\nurule$- and $\lpar$-rules only.
		We can conclude by inductive hypothesis on the derivation in the right of \Cref{eq:deadlock:com}, since it provides a derivation in $\PIL$ for the ${P'}=\cN'\Ctx[Q \ppar R\fsubst yz]$, where $P'$ is the process such that $P\redsem P'$ via a reduction step with core-reduction a $\rscomr$.

		\item
		Otherwise, we can deduce that $P=\cN\Ctx[\plrecv{x}{\ell:Q_\ell}{\ell\in L},\plsend{x}{\ell:R_\ell}{\ell\in L}]$,
		then we can apply a similar reasoning to deduce that $\dD_P$ is of the following shape below on the left.
		\begin{equation}\label{eq:deadlock:choice}
			\hskip-2.3em\adjustbox{max width=.96\textwidth}{$\begin{array}{c}
					\begin{array}{c@{\!}cccc}
						\vlderivation{
							\vlde{\dD_1}{}{\vdash \fof P}{
								\vlin{\lwith}{}{
									\vdash
									\biglsel[\ell\in L_1]{\left(\lrecv{x}{\ell} \lprec \fof{Q_{\ell}}\right)},
									\biglbra[\ell\in L_2]{\left(\lsend{x}\ell\lprec\fof{R_{\ell}} \right)},
									\Gamma_1
								}{\multiprem{
										\vlde{\dD_2^\ell}{}{
											\vdash
											\biglsel[\ell\in L_1]{\left(\lrecv{x}{\ell} \lprec \fof{Q_{\ell}}\right)},
											\lsend{x}\ell\lprec\fof{R_{\ell}},
											\Gamma_2
										}{\vlin{\lplus}{}{
												\vdash
												\biglsel[\ell\in L_1]{\left(\lrecv{x}{\ell} \lprec \fof{Q_{\ell}}\right)},
												\lsend{x}\ell\lprec\fof{R_{\ell}},
												\Gamma_3
											}{
												\vlde{\dD_3^\ell}{}{
													\vdash
													\lrecv{x}{\ell} \lprec \fof{Q_{\ell}},
													\lsend{x}\ell\lprec\fof{R_{\ell}},
													\Gamma_4
												}{
													\vliin{\lprec}{}{
														\vdash
														\lrecv{x}{\ell} \lprec \fof{Q_{\ell}},
														\lsend{x}\ell\lprec\fof{R_{\ell}},
														\Gamma_5
													}{
														\vlin{\axrule}{}{\vdash \lsend x\lab,\lrecv x\lab}{\vlhy{}}
													}{
														\vlpr{\dD_4^\ell}{}{\vdash \fof{Q_{\ell}}, \fof{R_{\ell}}, \Gamma_5}
													}
												}
											}
										}
									}{\ell\in L_1}
								}
							}
						}
						&
						\rightsquigarrow
						&
						\vlderivation{
							\vlde{\dD^-_{(A,B,\Gamma)}}{}{\vdash \fof P}{
								\vlin{\lwith}{}{
									\vpz1{}\vdash
									\biglsel[\ell\in L_1]{\left(\lrecv{x}{\ell} \lprec \fof{Q_{\ell}}\right)},
									\biglbra[\ell\in L_2]{\left(\lsend{x}\ell\lprec\fof{R_{\ell}} \right)},
									\Gamma
								}{\multiprem{
										\vlin{\oplus}{}{
											\vdash
											\biglsel[\ell\in L_1]{\left(\lrecv{x}{\ell} \lprec \fof{Q_{\ell}}\right)},
											\lsend{x}\ell\lprec\fof{R_{\ell}},
											\Gamma\;
										}{
											\vliin{\lprec}{}{
												\vdash
												\lrecv{x}{\ell} \lprec \fof{Q_{\ell}},
												\lsend{x}\ell\lprec\fof{R_{\ell}},
												\Gamma
											}{
												\vlin{\axrule}{}{\vdash \lsend x\lab,\lrecv x\lab}{\vlhy{}}
											}{
												\vlpr{\dD_\ell}{}{
													\vdash
													\fof{Q_{\ell}},
													\fof{R_{\ell}},
													\Gamma
												}
											}
										}
									}{\lab\in L_1}
								}
							}
							\tikz[overlay,remember picture]{
								\draw[draw=yellow,line width = 4pt,opacity=.5] %
								(pz1) ++(-15pt,-14pt) -- %
								++(215pt,0pt) 	--
								++(0pt,80pt)	--
								++(-215pt,0pt)	--
								++(0pt,-80pt);
							}
						}
						&\mapsto&
						\left\{
						\vlderivation{
							\vlde{\dD^-_{(\alpha', \beta',\Gamma)}}{}{\vdash \fof {P_\ell}}{
								\vlpr{\dD_\lab}{}{
									\vdash \fof{Q_{\ell}},\fof{R_\ell}, \Gamma
								}
							}
						}
						\right\}_{\lab\in L}
						\!\!
					\end{array}
					\\
					\mbox{
						with
						$\alpha=\biglsel[\ell\in L_1]{\left(\lrecv{x}{\ell} \lprec \fof{Q_{\ell}}\right)}$
						and
						$\beta=\biglbra[\ell\in L_2]{\left(\lsend{x}\ell\lprec\fof{R_{\ell}} \right)}$
						,
						and
						with $\alpha'=\fof{Q_{\ell}}$
						and  $\beta'=\fof{R_\ell}$
					}
				\end{array}$}
			\hskip-1em
		\end{equation}
		Then we can apply rule permutation to obtain a derivation $\dD'_P$ of the shape above in the center,
		where the derivations $\dD^-$ are made of $\nurule$- and $\lpar$-rules only.
		We again conclude by inductive hypothesis since
		there is a derivation of the formula
		$\fof{P_\ell}=\fof{\cN\Ctx[Q_\ell\ppar R_\ell]}$
		as shown above in the right of \Cref{eq:deadlock:choice}
		and
		$P\redsem P_\ell$.
	\end{enumerate}

	Note that since $P$ is \rfree, then it suffices to reason on a single \comptree and not to take into account all possible \comptrees of $P$.
\end{proof}

\begin{lemma}\label{lem:mer}
	Let $P$, $Q$, $R$ and $S$ be \eprocs.
	\begin{enumerate}
		\item\label{lem:mer1}
		For all $P'$ and  $Q'$,
		$$(P \ppar Q) \chOrd (P' \ppar Q') \iff P \chOrd P' \mbox{ and } Q \chOrd Q'\mydot$$.

		\item\label{lem:mer2}
		If $P \chOrd Q$,
		then $\precv xy P \chOrd \precv xy Q$ and  $\psend xy P \chOrd \psend xy Q$ .

		\item\label{lem:mer3}
		If $P \chOrd Q$ and $P \chOrd R$,
		then $P \chOrd( Q \merge R)$ .

		\item\label{lem:mer4}
		If $P \chOrd Q$ and $R \chOrd S$,
		then $(P \merge R) \chOrd (Q \merge S)$.
	\end{enumerate}
\end{lemma}
\begin{proof}
These standard properties of merge are proven in Chapter $8$ of \cite{montesi:book}.
\end{proof}

\begin{lemma}\label{lemma:choreo}
	Let $\chC$ be a \cflat \prochor \chor.
	If $\chC\osto{\;\tm\;}\chC'$, then $\chC'$ is \prochor and $\EPP[\pr]\chC \chOrd \EPP[\pr]{\chC'}$
	for all $\pr \notin \pnamesin{\tm}$.
\end{lemma}
\begin{proof}
	By induction on the structure of $\chC$.
	\begin{itemize}
		\item
		If $\chC'$ is obtained by $\ccomr$, then
		$\chC = \gencom \chC^1 \osto{\tcom \pp\pq k} \chC'=\chC^1 \chorfsubst xy\pq$.
		Since substitution does not affect $\pr\notin\set{\pp,\pq}$, we conclude by definition of EPP observing that
		$\EPP[\pp]{\chC^1} = \EPP[\pp]{\chC^1 \chorfsubst xy\pq}$.

		\item
		If $\chC'$ is obtained by $\cchoir$, then there is a $\lab_i\in L$ such that
		$$\adjustbox{max width= .9\textwidth}{$\chC =\genchois \osto{\;\;\tbra\pp k\;\;} \genchoi = \chC'$}\mydot$$
		We conclude since
		$\EPP[\pr]{\chC } = \merge_{\lab \in L} \EPP[\pr]{\chC_\lab} = \EPP[\pr]{\chC'}$.

		\item If $\chC'$ is obtained by $\clabr$, then $\lab_i \in L'$
		and
		$$\chC = \genchoi \osto{\tcom \pp\pq k} \chC_{\lab_i} = \chC'\mydot$$
		We conclude since $\EPP[\pr]{\chC } = \merge_{\lab \in L} \EPP[\pr]{\chC_{\lab}} \chOrd \EPP[\pr]{\chC'}$.

		\item If  $\chC'$ is obtained by $\cdelir$, then
		$$\chC = \gencom \chC^1 \osto{\;\;\tm\;\;}
		\gencom\chC^2 = \chC'$$
		for some \chors $\chC^1$ and $\chC^2$  such that
		$\chC^1 \osto{\;\;\tm\;\;}\chC^2$ with
		$\pnamesin{\tm}\cap\set{\pp,\pq,k}=\emptyset$.
		We have the following cases:
		\begin{itemize}
			\item
			If $\pq \neq \pr \neq \pp$, then
			$\EPP[\pr]{\chC} = \EPP[\pr]{\chC^1}$ and $\EPP[\pr]{\chC'} =  \EPP[\pr]{\chC^2}$
			. We conclude by inductive hypothesis.

			\item
			If $\pr = \pp$, then
			$\EPP[\pr]{\chC} = \psend k x \EPP[\pp]{\chC^1}$
			and $\EPP[\pr]{\chC'} = \psend k x \EPP[\pp]{\chC^2}$.
			We conclude by inductive hypothesis and \Cref{lem:mer2}.

			\item
			Otherwise, $\pr = \pq$ and
			$\EPP[\pr]{\chC} = \precv k x \EPP[\pp]{\chC^1}$
			and $\EPP[\pr]{\chC'} = \precv k x \EPP[\pp]{\chC^2}$.
			We conclude by inductive hypothesis and \Cref{lem:mer2}.

		\end{itemize}

		\item
		If  $\chC'$ is obtained by $\cdelcr$, then
		$$\adjustbox{max width=.9\textwidth}{
			$\chC = \genchois\osto{\;\;\tm\;\;}\Choics[\chC^1]{p}{q}{k}{\lab}{L} = \chC'
		$}$$
		for some \chors $\chC^1_\lab$ such that
		$\chC_\lab\osto{\;\tm\;}\chC_\lab^1$ with $\pnamesin{\tm}\cap\set{\pp,\pq,k}=\emptyset$ for all $\lab\in L$.
		We have the following cases:
			\begin{itemize}
				\item
				If $\pq \neq \pr \neq \pp$, then $\EPP[\pr]{\chC} = \merge_{\lab \in L}\EPP[\pr]{\chC_\lab}$ and
				$\EPP[\pr]{\chC'} = \merge_{\lab \in L} \EPP[\pr]{\chC^1_\lab}$.
				We conclude by inductive hypothesis and \Cref{lem:mer4}.

				\item If $\pr = \pp$, then
				$\EPP[\pr]{\chC} = \pbras {k}{\lab : \EPP[\pr]{\chC_\lab}}{\lab\in L}$
				and
				$\EPP[\pr]{\chC'} = \pbras {k}{\lab : \EPP[\pr]{\chC^1_\lab}}{\lab\in L}$.
				We conclude by inductive hypothesis.

				\item
				If $\pr = \pq$ then
				$\EPP[\pr]{\chC} = \psels {k}{
					\begin{array}{l|@{\;}l}
						\lab : \EPP[\pp_i]{\chC_\lab} 	& \lab\in L
						\\
						\lab : S_\lab					& \lab\in L'\setminus L
					\end{array}
				}{}$
				and
				$\EPP[\pr]{\chC'} = \psels {k}{
					\begin{array}{l|@{\;}l}
						\lab : \EPP[\pp_i]{\chC^1_\lab} 	& \lab\in L
						\\
						\lab : S_\lab					& \lab\in L'\setminus L
					\end{array}
				}{} \mydot$
				We conclude by inductive hypothesis.
			\end{itemize}
		\item
		If $\chC'$ is obtained by $\crestr$ the claim follows by inductive hypothesis.
	\end{itemize}
\end{proof}

\EPPthm*
\begin{proof}~

	\textbf{Completeness}.
	By case analysis on the reduction steps.
	\begin{itemize}
		\item
		If $\chC'$ is obtained by a $\ccomr$ step, then
		$\chC = \gencom \chC^1$ and $\chC' = \chC^1 \chorfsubst xy\pq$.
		By definition
		$$
		\EPP \chC
		=
		\prod_{i=1}^{n}\proloc{\pp_i}{\EPP[\pp_i]{\chC}}
		\ppar
		\proloc\pp{\psend k x \EPP[\pp]{\chC^1}}
		\ppar
		\proloc\pq{\precv k y \EPP[\pq]{\chC^1}}
		$$
		and by \Cref{lemma:choreo} we also have that
		$$
		\EPP{\chC'}
		=
		\prod_{i=1}^{n}\proloc{\pp_i}{\EPP[\pp_i]{\chC}}
		\ppar
		\proloc\pp{\EPP[\pp]{\chC^1 \chorfsubst xy\pq}}
		\ppar
		\proloc\pq{\EPP[\pq]{\chC^1\chorfsubst xy\pq} }
		\mydot
		$$

		Since $\EPP[\pp]{\chC^1 \chorfsubst xy\pq} =\EPP[\pp]{\chC^1}$
		and
		{$\EPP[\pq]{\chC^1} \chorfsubst xy\pq= \EPP[\pq]{\chC^1 \chorfsubst xy\pq}$}
		we conclude observing that $\EPP\chC\osto{\;\tm\;}\EPP{\chC'}$ by a $\Ecomr$ step.

		\item
		If $\chC'$ is obtained by a $\cchoir$ step, then
		$\chC =\genchois$.
		By definition and by \Cref{lemma:choreo}, we have that
		$$\adjustbox{max width=.9\textwidth}{$\begin{array}{l}
			\EPP \chC=
			Q
			\ppar
			\proloc\pp{\pbras {k}{\lab : \EPP[\pp]{\chC_\lab}}{\lab\in L}}
			\ppar
			\proloc\pq{\psels {k}{
				\begin{array}{l|@{\;}l}
					\lab : \EPP[\pq]{\chC_\lab} 	& \lab\in L
					\\
					\lab : S_\lab					& \lab\in L'\setminus L
				\end{array}
			}{}}
		\\\mbox{and}\\
			\EPP{\chC'}=
			Q
			\ppar
			\proloc\pp{\pbras {k}{\lab_i : \EPP[\pp]{\chC_{\lab_i}}}{}}
			\ppar
			\proloc\pq{\psels {k}{
					\begin{array}{l|@{\;}l}
						\lab : \EPP[\pq]{\chC_\lab} 	& \lab\in L
						\\
						\lab : S_\lab					& \lab\in L'\setminus L
					\end{array}
				}{}}
		\end{array}$}$$
		with $Q=\Merge_{\lab \in L}\prod_{i=1}^{n}\proloc{\pp_i}{\EPP[\pp_i]{\chC_\lab}}$ .
		Hence, $\EPP\chC\osto{\;\tm\;}\EPP{\chC'}$ by a $\Cchor$ step.

		\item
		If $\chC'$ is obtained by a $\clabr$ step, then $\chC = \genchoi$.
		By definition
		$$\adjustbox{max width=.9\textwidth}{$\begin{array}{l}
			\EPP{\chC}
			=
			Q
			\ppar
			\proloc\pp{\pbras {k}{\lab_i : \EPP[\pp]{\chC_{\lab_i}}}{}}
			\ppar
			\proloc\pq{\psels {k}{
					\begin{array}{l|@{\;}l}
						\lab : \EPP[\pq]{\chC_\lab} 	& \lab\in L
						\\
						\lab : S_\lab					& \lab\in L'\setminus L
					\end{array}
			}{}}
		\\\mbox{and}\\
			\EPP{\chC'}
			=
			Q
			\ppar
			\proloc\pp{\EPP[\pp]{\chC_{\lab_i}}}
			\ppar
			\proloc\pq{\EPP[\pq]{\chC_{\lab_i}}}
		\end{array}			$}
		$$
		with $Q=\Merge_{\lab \in L}\prod_{i=1}^{n}\proloc{\pp_i}{\EPP[\pp_i]{\chC_\lab}}$.
		Since $\merge_{\lab \in L} \EPP[\pp_i]{\chC_\lab} \chOrd \EPP[\pp_i]{\chC_{\lab_i}}$ by
		\Cref{lem:mer}.\ref{lem:mer1}, then we conclude
		$$\adjustbox{max width=.9\textwidth}{$
			\EPP{\chC}
			\redsem
			\left(\Merge_{\lab \in L} \prod_{i=1}^{n}\proloc{\pp_i}{\EPP[\pp_i]{\chC_\lab}}
			\ppar
			\proloc\pp{\EPP[\pp]{\chC_{\lab_i}} }
			\ppar
			\proloc\pq{\EPP[\pq]{\chC_{\lab_i}}}\right)
			\chOrd
			\EPP{\chC'}
		\mydot
		$}$$

		\item
		If $\chC'$ is obtained by a $\crestr$ step,
		then $\chC =\cnu x  \chC^1$ and
		$\chC' = \cnu x  \chC^2$.

		By inductive hypothesis $\EPP{\chC^1} \osto{\;\tm\;} P \chOrd \EPP{\chC^2}$,
		and by definition
		$$\begin{array}{l}
			\EPP{\chC} =
			\pnu x \EPP{\chC^1} \osto{\;\tm\;} \pnu x P
			\chOrd \chC' \mydot
		\end{array}$$

		\item
		If $\chC'$ is obtained by a $\cdelir$ step,
		then $\chC = \gencom \chC^1$ and $\chC' = \gencom \chC^2$.

		By induction hypothesis
		$\EPP{\chC^1} \osto{\;\tm\;} P \chOrd \EPP{\chC^2}$.
		In particular since $\tm \notin \set{\pp, \pq}$ there is an \eproc $Q$
		such that
		\begin{equation}\label{eq:choreo_comp1}
		\adjustbox{max width=.9\textwidth}{$\begin{array}{rcccl}
			\EPP{\chC^1}
			& = &
			\left(Q \ppar \proloc\pp{\EPP[\pp]{\chC_1}} \ppar \proloc\pq{\EPP[\pq]{\chC_1}}\right)
			&\osto{\;\tm\;}\\&\osto{\;\tm\;}&
			\left(Q' \ppar \proloc\pp{\EPP[\pp]{\chC_1}} \ppar \proloc\pq{\EPP[\pq]{\chC_1}}\right)
			&\chOrd& \EPP{\chC^2}
		\end{array}$}
		\end{equation}
		with
		$Q \osto{\;\tm\;} Q'$ and $P = Q' \ppar \proloc\pp{\EPP[\pp]{\chC_1}} \ppar \proloc\pq{\EPP[\pq]{\chC_1}}$.

		We deduce:
		\begin{equation}\label{eq:choreo_comp}
			\EPP{\chC}
			\osto{\;\tm\;}
			Q' \ppar \proloc\pp{\psend k x \EPP[\pp]{\chC^1}}\ppar
			\proloc\pq{\precv k y \EPP[\pq]{\chC^1}}
			\mydot
		\end{equation}

		By \Cref{lemma:choreo} applied to \Cref{eq:choreo_comp1}
		we deduce
		$ Q' \chOrd \prod_{i=1}^{n}\proloc{\pp_i}{\EPP[\pp_i]{\chC^2} }$.

		By \Cref{lemma:choreo} the following equations hold:
		\begin{equation}\label{eq:choreo_comp3}
			\psend k x \EPP[\pp]{\chC^1} \chOrd \psend k x \EPP[\pp]{\chC^2}
		\qquad \psend k x \EPP[\pq]{\chC^1} \chOrd \psend k x \EPP[\pq]{\chC^2}
		\end{equation}

		From \Cref{eq:choreo_comp} and \Cref{eq:choreo_comp3}, using \Cref{lem:mer}.\ref{lem:mer1}
		we deduce
		$$
		Q'
		\ppar
		\proloc\pp{\psend k x \EPP[\pp]{\chC^1}}
		\ppar
		\proloc\pq{\precv k y \EPP[\pq]{\chC^1}}
		\chOrd
		\EPP{\chC'}
		\mydot
		$$

		\item
		if $\chC'$ is obtained by a $\cdelcr$ step, then $\chC' = \Choics[\chC^1]{\pp}{\pq}k\lab{L}$ with $\chC_\lab \osto{\; \tm \;} \chC^1_\lab$
		for $\lab \in L$.
		By inductive hypothesis $\EPP[]{\chC_\lab} \osto{\; \tm \;} P_\lab \chOrd \EPP[]{\chC^1_\lab}$ for $\lab \in L$.

		Since $\pnamesin{\tm'}\cap\set{\pp,\pq,k}=\emptyset$ there must be a $Q_\lab$ such that
		\begin{equation}\label{eq:delchoi_compl}
			\prod_{i=1}^{n}\proloc{\pp_i}{\EPP[\pp_i]{\chC_\lab}} \osto{\; \tm \;} Q_\lab \chOrd \prod_{i=1}^{n}\proloc{\pp_i}{\EPP[\pp_i]{\chC^1_\lab}}
		\end{equation}
		 and
		$P_\lab = Q_\lab \ppar \proloc\pp{\EPP[\pp]{\chC_\lab}} \ppar \proloc\pq{\EPP[\pq]{\chC_\lab}}$.

		Let us fix the following notation
		$$\adjustbox{max width=.9\textwidth}{$\begin{array}{l}
			\EPP{\chC}
			=
			Q
			\ppar
			\proloc\pp{\pbras {k}{\lab : \EPP[\pp]{\chC_{\lab}}}{\lab \in L}}
			\ppar
			\proloc\pq{\psels {k}{
					\begin{array}{l|@{\;}l}
						\lab : \EPP[\pq]{\chC_\lab} 	& \lab\in L
						\\
						\lab : S_\lab					& \lab\in L'\setminus L
					\end{array}
			}{}}
		\\\mbox{and}\\
			\EPP{\chC'}
			=
			Q'
			\ppar
			\proloc\pp{\pbras {k}{\lab : \EPP[\pp]{\chC^1_{\lab}}}{\lab \in L}}
			\ppar
			\proloc\pq{\psels {k}{
					\begin{array}{l|@{\;}l}
						\lab : \EPP[\pq]{\chC^1_\lab} 	& \lab\in L
						\\
						\lab : S_\lab					& \lab\in L'\setminus L
					\end{array}
			}{}}
		\end{array}			$}
		$$
		with $Q=\Merge_{\lab \in L}\prod_{i=1}^{n}\proloc{\pp_i}{\EPP[\pp_i]{\chC_\lab}}$ and
		$Q'=\Merge_{\lab \in L}\prod_{i=1}^{n}\proloc{\pp_i}{\EPP[\pp_i]{\chC^1_\lab}}$.

		From \Cref{eq:delchoi_compl} by \Cref{lem:mer} we deduce
		\begin{equation}\label{punto1}
			Q \osto{\;\tm\;} \Merge_{\lab \in L} Q_\lab \chOrd Q' \mydot
		\end{equation}

		Since
		by \Cref{lemma:choreo} $\EPP[\pp]{\chC} \chOrd \EPP[\pp]{\chC^1}$ and $\EPP[\pq]{\chC} \chOrd \EPP[\pq]{\chC^1}$
		we deduce
		\begin{equation}\label{punto2}
			\proloc\pp{\pbras {k}{\lab : \EPP[\pp]{\chC_{\lab}}}{\lab \in L}} \chOrd
			\proloc\pp{\pbras {k}{\lab : \EPP[\pp]{\chC^1_{\lab}}}{\lab \in L}}
		\end{equation}
		and
		\begin{equation}\label{punto3}
			\adjustbox{max width=.9\textwidth}{$\proloc\pq{\psels {k}{
		\begin{array}{l|@{\;}l}
			\lab : \EPP[\pq]{\chC_\lab} 	& \lab\in L
			\\
			\lab : S_\lab					& \lab\in L'\setminus L
		\end{array}
		}{}} \chOrd \proloc\pq{\psels {k}{
			\begin{array}{l|@{\;}l}
				\lab : \EPP[\pq]{\chC^1_\lab} 	& \lab\in L
				\\
				\lab : S_\lab					& \lab\in L'\setminus L
			\end{array}
		}{}} \mydot$}
		\end{equation}

		We conclude by \Cref{lem:mer}.\ref{lem:mer1} and \Cref{punto1}, \Cref{punto2}, \Cref{punto3}.
	\end{itemize}

	\textbf{Soundness}. By induction on the structure of $\chC$:

	\begin{itemize}
		\item if $\chC = \chnil$ then $P = \pnil$ and there is no transition to consider;

		\item if $\chC = \gencom \chC^1 $, then by \Cref{lem:mer}.\ref{lem:mer2}, by definition of merge and by definition of endpoint projection
		$P = R \ppar \proloc\pp{\psend kx P_\pp} \ppar \proloc\pq{\precv ky P_\pq}$
		with
		$R \chOrd \left(\prod_{i=1}^{n}\proloc{\pp_i}{\EPP[\pp_i]{\chC^1}}\right)$,
		and
		$P_{\pp} \chOrd \EPP[\pp]{\chC^1} $
		and
		$ P_{\pq} \chOrd \EPP[\pq]{\chC^1}$.

		There are two possible cases:

		\begin{itemize}
			\item $\tm = \tcom \pp\pq k$, then $P' = R \ppar P_{\pp} \ppar P_{\pq} \fsubst xy $ . We set $\chC' = \chC^1 \chorfsubst xy\pq $.
			We conclude observing that $\EPP[\pp_i]{\chC^1 \chorfsubst xy\pq} = \EPP[\pp_i]{\chC^1}$
			for all $\pp_i \neq \pq$
			and $\EPP[\pq]{\chC^1 \chorfsubst xy\pq} = \EPP[\pq]{\chC^1} \fsubst xy$.

			\item
			$\tm \neq \tcom \pp\pq k$,
			then $P' = R' \ppar \psend kx P_{\pp} \ppar \precv ky P_{\pq}$ with $R \osto{\;\tm\;} R'$. We set $\chC' = \gencom \chC^2$
			with $\chC^1 \osto{\;\tm\;} \chC^2$.
			Note that $R \ppar P_{\pp} \ppar P_{\pq} \osto{\;\tm\;} R' \ppar P_{\pp} \ppar P_{\pq}$.
			By inductive hypothesis $\chC^1 \osto{\;\tm\;} \chC^2$ and $R' \ppar P_{\pp} \ppar P_{\pq} \chOrd \EPP{\chC^2}$.
			We conclude by \Cref{lem:mer}.
		\end{itemize}

		\item
		If $\chC = \genchois$, then by \Cref{lem:mer}.\ref{lem:mer2}, by definition of merge and by definition of endpoint projection

		$$
		\adjustbox{max width=.9\textwidth}{$P = R \ppar \pbras {k}{\lab : P_{\lab}}{\lab \in L} \ppar \psels {k}{
			\begin{array}{l|@{\;}l}
				\lab : {Q_\lab} 	& \lab\in L
				\\
				\lab : R_\lab					& \lab\in L'\setminus L
				\\
				\lab : T_\lab 		& \lab \in L'' \setminus L'
			\end{array}
		}{}$}
		$$

		with $L'' \subseteq L'$
		set of labels,
		$R \chOrd \Merge_{\lab \in L} \prod_{i=1}^n\EPP[\pp_i]{\chC_{\lab}}$,
		and
		$P_{\lab} \chOrd \EPP[\pp]{\chC_\lab} $ ,
		and
		$ P_{\lab} \chOrd \EPP[\pq]{\chC_\lab}$ for all $\lab \in L$,
		and
		$R_\lab \chOrd S_\lab$ for $\lab \in L'$.

		\begin{itemize}
			\item
			If $\pnamesin{\tm }\cap \set{\pp, \pq } \neq \varnothing$,
			then
			either
			$P \osto{\;\tm\;} P'$ via a step $\cchoir$, in this case we set $\chC' = \genchoi$,
			or $L = \set{\lab_i}$ and
			$P \osto{\;\tm\;} P'$ via a step $\clabr$, we set $\chC' = \chC_{\lab_i}$.
			The claim follows from \Cref{lem:mer} and definition of endpoint projection.

			\item
			If $\pnamesin{\tm }\cap \set{\pp, \pq } = \varnothing$,
			then
			$$\adjustbox{max width=.9\textwidth}{$
				P =
				R' \ppar
				\proloc\pp{\pbras {k}{\lab : P_{\lab}}{\lab \in L}}
				\ppar
				\proloc\pq{\psels {k}{
				\begin{array}{l|@{\;}l}
					\lab : {Q_\lab} 	& \lab\in L
					\\
					\lab : R_\lab					& \lab\in L'\setminus L
					\\
					\lab : T_\lab 		& \lab \in L'' \setminus L'
				\end{array}
				}{}}
			$}$$

			with $R \osto{\; \tm \;} R'$.
			In this case we set $\chC' = \Choics[\chC^1]{\pp}{\pq}k\lab{L}$ with $\chC_\lab \osto{\; \tm \;} \chC^1_\lab$ for $\lab \in L$.
			By inductive hypothesis
			$ R' \chOrd \prod_{i=1}^n \EPP[\pp_i]{\chC^1_\lab}$
			for all $\lab \in L$.
			We conclude by \Cref{lem:mer}.\ref{lem:mer3}.

		\end{itemize}

	\end{itemize}
\end{proof}


\thmChor*
\begin{proof}
	If $P =\pnus x \left(\prod_{i=1}^{n}\proloc{\pp_i}{S_i}\right)$ is \lfree, then, by \Cref{thm:deadlockfreeChor}, there is a derivation in $\ChorL$ with conclusion $\fofp P$.
	In particular, $\dD$ contains at most a unique an occurrence of the rule $\Crestr$ (as the bottom-most rule in $\dD$) only if $k>0$.
	That is,
	$$
	\vlderivation{
		\dD=
		\vlin{\Crestr}{}{
			\vdash \fofp{\pnus x \left( \proloc{\pp_1}{S_1}\ppars \proloc{\pp_n}{S_n}\right)}
		}{
			\vlpr{\dD'}{}{
				\vdash \named{\fof{S_1}}{\pp_1},\ldots , \named{\fof{S_n}}{\pp_n}
			}
		}
	}
	\mbox{ if $k>0$, otherwise } \dD = \dD' \mydot
	$$
	Therefore it suffices to prove that
	$\EPP{\chCof[\dD']}=\prod_{i=1}^{n}\proloc{\pp_i}{S_i}$
	to conclude, since  $\EPP{\chCof[\dD]} = \pnus[k]{x}(\EPP{\chCof[\dD']}) = P$.

	We proceed by induction on the structure of $\dD'$ reasoning on the bottom-most rule $\rrule$ in $\dD'$:

	\begin{itemize}
		\item if
		$\rrule=\Caxr$,
		then $P = \pnil$ and
 		$\chCof[\dD']=\chnil$
		thus $P = {\EPP\chnil}$;

		\item
		if
		$\rrule=\Ccomr$,
		then
		$\chCof[\dD']=\gencom \chCof[\dD'']$ for a $\dD''$ such that
		$$
		\dD'=
		\vlderivation{
			\vlin{\Ccomr}{}{
				\vdash \Gamma,
				\named[pzbrickred]{\fof{\psend xy S}}{\pp},
				\named[pzgreen]{\fof{\precv xz S'}}{\pq}
			}{
				\vlpr{\dD''}{}{
					\vdash \Gamma,
					\named[pzbrickred]{\fof{S}}{\pp},
					\named[pzgreen]{\fof{S'\fsubst yz}}{\pq}
				}
			}
		}
		$$

		By inductive hypothesis
		$\ffof{\EPP{\chCof[\dD'']}}= \Gamma , \named{\fof S}{\pp}  , \named{\fof{S'\fsubst yz}}\pq$.

		Thus, we conclude since
		$\dD'$ has conclusion the following sequent.
		$$\begin{array}{lcl}
			\Gamma, \named{\fof{\psend xy S}}\pp , \named{\fof{\precv xz S'}}\pq
			&=&
			\ffof{\EPP{\gencom\chCof[\dD'']}}=
			\\
			&=&
			\ffof{\EPP{\chCof[\dD']}}
		\end{array}
		\quad;
		$$

		\item
		if
		$\rrule=\Cchor$,
		then
		$$
		\dD'=
			\vlderivation{
				\vlin{\Cchor}{L \subseteq L'}{
					\vdash
					\Gamma,
					\named[pzbrickred]{\fof{\plsend{k}{\lab:S_\lab}{\lab\in L}}}{\pp} , 	\named[pzgreen]{\fof{\plrecv{k}{\lab:S'_\lab}{\lab\in L'}}}{\pq}
				}{\multiprem{
					\vlpr{\dD_\lab}{}{
								\vdash
								\Gamma,
								\named[pzbrickred]{\fof{S_\lab}}{\pp} , \named[pzgreen]{\fof{S'_\lab}}{\pq}
					}
				}{\lab\in L}}
			}
		$$
	\end{itemize}
	By inductive hypothesis $\Gamma,
	\named[pzbrickred]{\fof{S_\lab}}{\pp} , \named[pzgreen]{\fof{S'_\lab}}{\pq}
	=
	\EPP{\chCof[\dD_\lab]}$ for each $\lab \in L$.
	We conclude observing that
	$$
		\begin{array}{l}
			\Gamma, \named[pzbrickred]{\fof{\plsend{k}{\lab:S_\lab}{\lab\in L}}}{\pp} ,
			\named[pzgreen]{\fof{\plrecv{k}{\lab:S'_\lab}{\lab\in L'}}}{\pq}
			=
			\\\\
			= \ffof{\EPP{
				\specchoics{p}{q}{k}{L}{\lab}{
					\chCof[\vlderivation{
					\vlpr{\dD_{\lab}}{}{\vdash \Gamma, \named[pzbrickred]{\fof{S_\ell}}{\pp}, \named[pzgreen]{\fof{S'_\ell}}{\pq}}
					}]
				}{S'_\lab}
			}}=
			\\\\
			= \ffof{\EPP{\chCof[\dD']}}
		\end{array}
	$$
	We conclude by remarking that
	if $k = 0$ then $\chCof[\dD] = \chCof[\dD']$
	and if $k  > 0$ then $\EPP{\chCof[\dD]} = \pnus[k]{x}(\EPP{\chCof[\dD']}) = P$.

	The left-to-right implication follows by \Cref{thm:choreo}.
\end{proof}

\clearpage
\section{Running example of \chor extraction}\label{app:example}

Consider the following \chors.
\begin{equation}
	\adjustbox{max width = .96\textwidth}{$\begin{array}{r@{=}l}
		\chC&
		\choic{p}L{q}L{k_{\pp\pq}}{
			\lab_1\colon
			\comc{q}{x_1}{r}{a}{k_{\pq\pr}}		;
			\choic{p}{\set{\lab_1'}}{q}{\set{\lab_1'}}{k_{\pq\pr}}{
				\lab_1'\colon
				\comc{q}{y_1}{p}{z}{k_{\pp\pr}}
			}
			\\
			\lab_2\colon
			\comc{q}{x_2}{r}{a}{k_{\pq\pr}}		;
			\choic{p}{\set{\lab_2'}}{q}{\set{\lab_2'}}{k_{\pq\pr}}{
				\lab_2'\colon
				\comc{q}{y_1}{p}{z}{k_{\pp\pr}}
			}
		}
		\\
		\chC'
		&
		\choic{p}L{q}L{k_{\pp\pq}}{
			\lab_1 \colon
			\comc{q}{x_1}{r}{a}{k_{\pq\pr}}		;
			\comc{r}{y_1}{p}z{k_{\pp\pr}}	;
		   \chnil
		   \\
			\lab_2 \colon
		   \comc{q}{x_2}{r}{a}{k_{\pq\pr}}		;
			\comc{r}{y_2}pz{k_{\pp\pr}}		;
		   \chnil
	   }
	\end{array}$}
\end{equation}
We observe that $\chC'$ is not projectable, and that $\chC$ can be obtained from $\chC'$ by amendment~\cite{BB16,DLMZ13,CM23}.
The endpoint projection of $\chC$ is defined as follows.
$$
\EPP{\chC}=\proloc\pp \EPP[\pp]{\chC} \ppar \proloc\pq \EPP[\pq]{\chC} \ppar \proloc\pr \EPP[\pr]{\chC}
$$
with
$$
\begin{array}{l@{=}l}
	\EPP[\pp]{\chC}&
	\plsend{k_{\pp\pq}}{
		\begin{array}{l@{\colon}l}
			\lab_1&\precv{k_{\pp\pr}}{z} \pnil
			,\\
			\lab_2& \precv{k_{\pp\pr}}{z} \pnil
		\end{array}
	}{}
	\\
	\EPP[\pq]{\chC}&
	\plrecv{k_{\pp\pq}}{
		\begin{array}{l@{\colon}l}
			\lab_1& \psend{k_{\pq\pr}}{x_1} \plsend{k_{\pq\pr}}{\lab_1'\colon \pnil }{}
			,\\
			\lab_2& \psend{k_{\pq\pr}}{x_2}\plsend{k_{\pq\pr}}{\lab_2'\colon\pnil }{}
		\end{array}
	}{}
	\\
	\EPP[\pr]{\chC}&
		\precv{k_{\pq\pr}}{a}
		\plrecv{k_{\pp\pq}}{
			\begin{array}{l@{\colon}l}
				\lab'_1& \psend{k_{\pp\pr}}{y_1}\pnil
				,\\
				\lab'_2& \psend{k_{\pp\pr}}{y_1}\pnil
			\end{array}
		}{}
\end{array}
$$

The \proc $\EPP{\chC}$ is derivable in $\ChorL$ via the following derivation $\dD_{\EPP{\chC}}$
\begin{equation}\adjustbox{max width=\textwidth}{$
	\vlderivation{
		\vlin{\Crestr}{}{
			\named[pzbrickred]{\fof P}{\pp}
			\lpar
			\named[pzgreen]{\fof Q}{\pq}
			\lpar
			\named[orange]{\fof R}{\pr}
		}{
			\vlin{\Cchor}{}{
				\named[pzbrickred]{\fof P}{\pp}
				,
				\named[pzgreen]{\fof Q}{\pq}
				,
				\named[orange]{\fof R}{\pr}
			}{
				\multiprem{
					\vlpr{\dD_i}{}{
						\named[pzbrickred]{\lrecv{k_{\pp\pr}}{z}\lprec \lunit}{\pp}
						,
						\named[pzgreen]{
						\lsend{k_{\pq\pr}}{x_i}
						\lprec
						((\lsend{k_{\pq\pr}}{\lab'_1} \lprec \lunit)
						\lwith
						(\lsend{k_{\pq\pr}}{\lab'_2}\lprec \lunit))
					}{\pq}
					,
					\named[orange]{Q}{\pr}}
				}{i\in\set{1,2}}
			}
		}
	}
$}\end{equation}
where each derivation $\dD_i$ is defined as follows.
$$\adjustbox{max width=\textwidth}{$
\dD_i=
\vlderivation{
	\vlin{\Ccomr}{}{
		\named[pzbrickred]{\lrecv{k_{\pp\pr}}{z}\lprec \lunit}{\pp}
		,
		\named[pzgreen]{
			\lsend{k_{\pq\pr}}{x_i}
			\lprec
			((\lsend{k_{\pq\pr}}{\lab'_1} \lprec \lunit)
			\lwith
			(\lsend{k_{\pq\pr}}{\lab'_2}\lprec \lunit))
		}{\pq}
		,
		\named[orange]{Q}{\pr}
	}{
		\vlin{\Cchor}{}{
			\named[pzbrickred]{\lrecv{k_{\pp\pr}}{z}\lprec \lunit}{\pp}
			,
			\named[pzgreen]{\lwith(\lsend{k_{\pq\pr}}{\lab'_i} \lprec \lunit)}{\pq}
			,
			\named[orange]{
				(\lrecv{k_{\pp\pr}}{\lab_1'} \lprec \lsend{k_{\pp\pr}}{y_1}\lprec \lunit)
				\lplus
				(\lrecv{k_{\pp\pr}}{\lab_2'} \lprec \lsend{k_{\pp\pr}}{y_2}\lprec \lunit)
			}{\pr}
		}{
			\vlin{\Ccomr}{}{
				\named[pzbrickred]{\lrecv{k_{\pp\pr}}{z}\lprec \lunit}{\pp}
				,
				\named[pzgreen]{\lunit}{\pq}
				,
				\named[orange]{\lsend{k_{\pp\pr}}{y_i}\lprec \lunit}{\pr}
			}{
				\vlin{\Caxr}{}{
					\named[pzbrickred]{\lunit}{\pp}
				,
					\named[pzgreen]{\lunit}{\pq}
					,
					\named[orange]{\lunit}{\pr}
				}{\vlhy{}}
			}
		}
	}
}
$}$$

The \chor associated to $\dD_{\EPP{\chC}}$ is the following
$$
\begin{array}{l}
	\chCof[\dD_{\EPP{\chC}}]
	=
	\choic{p}L{q}L{k_{\pp\pq}}{
 		\lab_1 \colon
 		\chCof[\dD_1]
		\\
 		\lab_2 \colon
		\chCof[\dD_2]
	}
	=\\=
	\choic{p}L{q}L{k_{\pp\pq}}{
		\lab_1\colon
		\comc{q}{x_1}{r}{a}{k_{\pq\pr}}		;
		\choic{p}{\set{\lab_1'}}{q}{\set{\lab_1'}}{k_{\pq\pr}}{
			\lab_1'\colon
			\comc{q}{y_1}{p}{z}{k_{\pp\pr}}
		}
		\\
		\lab_2\colon
		\comc{q}{x_2}{r}{a}{k_{\pq\pr}}		;
		\choic{p}{\set{\lab_2'}}{q}{\set{\lab_2'}}{k_{\pq\pr}}{
			\lab_2'\colon
			\comc{q}{y_1}{p}{z}{k_{\pp\pr}}
		}
	}
\end{array}
$$
since \chor associated to each $\dD_i$ is the following.
$$\adjustbox{max width=\textwidth}{$\begin{array}{l}
	\chCof[\dD_i]
	=\\=
	\comc{q}{x_i}{r}{a}{k_{\pq\pr}}		;
	\chCof[
		\vlderivation{
			\vlin{\Cchor}{}{
				\named{\lrecv{k_{\pp\pr}}{z}\lprec \lunit}{\pp}
				,
				\named{\lwith(\lsend{k_{\pq\pr}}{\lab'_i} \lprec \lunit)}{\pq}
				,
				\named{
					(\lrecv{k_{\pp\pr}}{\lab_1'} \lprec \lsend{k_{\pp\pr}}{y_1}\lprec \lunit)
					\lplus
					(\lrecv{k_{\pp\pr}}{\lab_2'} \lprec \lsend{k_{\pp\pr}}{y_2}\lprec \lunit)
				}{\pr}
			}{
				\vlin{\Ccomr}{}{
					\named{\lrecv{k_{\pp\pr}}{z}\lprec \lunit}{\pp}
					,
					\named{\lunit}{\pq}
					,
					\named{\lsend{k_{\pp\pr}}{y_i}\lprec \lunit}{\pr}
				}{
					\vlin{\Caxr}{}{
						\named{\lunit}{\pp}
					,
						\named{\lunit}{\pq}
						,
						\named{\lunit}{\pr}
					}{\vlhy{}}
				}
			}
		}
	]
	=\\=
	\comc{q}{x_i}{r}{a}{k_{\pq\pr}}		;
	\choic{p}{\set{\lab_i'}}{q}{\set{\lab_i'}}{k_{\pq\pr}}{
		\lab_i'\colon
		\chCof[\vlderivation{
			\vlin{\Ccomr}{}{
				\named{\lrecv{k_{\pp\pr}}{z}\lprec \lunit}{\pp}
				,
				\named{\lunit}{\pq}
				,
				\named{\lsend{k_{\pp\pr}}{y_i}\lprec \lunit}{\pr}
			}{
				\vlin{\Caxr}{}{
					\named{\lunit}{\pp}
				,
					\named{\lunit}{\pq}
					,
					\named{\lunit}{\pr}
				}{\vlhy{}}
			}
		}]
	}
	=\\=
	\comc{q}{x_i}{r}{a}{k_{\pq\pr}}		;
	\choic{p}{\set{\lab_i'}}{q}{\set{\lab_i'}}{k_{\pq\pr}}{
		\lab_i'\colon
		\comc{q}{y_i}{p}{z}{k_{\pp\pr}}
	}
\end{array}$}$$

\end{document}

%% file: macros.tex
\def\entropy{entropy\xspace}
\def\entof#1#2{\mathsf{Ent}_{(#1,#2)}}
\def\crof#1#2{\mathsf{Core}_{(#1,#2)}}
\def\crxof#1#2{\mathsf{rdx}_{(#1,#2)}}
\def\crtof#1#2{\mathsf{rdt}_{(#1,#2)}}
\def\tree{\mathcal T}
\def\comptreeof#1{\mathsf{Ctree}(#1)}

\def\comptree{execution tree\xspace}
\def\comptrees{execution trees\xspace}

\def\N{\mathbb N}
\def\set#1{\{#1\}}
\def\Set#1{\left\{#1\right\}}

\def\tuple#1{\langle#1\rangle}
\usepackage{scalerel}

\def\intset#1#2{\set{#1,\ldots,#2}}

\def\xs{x_1,\ldots,x_k}

\renewcommand{\emptyset}{\varnothing}

\def\CCS{\textsf{CCS}\xspace}
\def\flat{flat\xspace}
\def\prochor{projectable\xspace}

\def\sequential{sequential\xspace}
\def\sproc{\sequential \proc\xspace}
\def\sprocs{\sequential \procs\xspace}
\def\seqcomp{sequential component\xspace}
\def\seqcomps{sequential components\xspace}

\def\cflat{flat\xspace}
\def\picalc{$\pi$-calculus\xspace}

\def\opsem{reduction semantics\xspace}
\def\Opsem{Reduction semantics\xspace}
\def\OpSem{Reduction Semantics\xspace}

\def\defn#1{{\itshape\bfseries\boldmath #1}}

\def\rclr#1{{\color{red}#1}}
\def\bclr#1{{\color{blue}#1}}
\def\sclr#1{{\color{magenta}#1}}

\def\resp#1{(resp.~{#1})\xspace}

\newcommand{\mydot}{\;.}

\newcommand{\qquand}{\qquad\mbox{and}\qquad}

\newcommand{\quor}{\quad\mbox{or}\quad}

\newcommand{\qomma}{\;,\;}

\def\myparagraph#1{\textbf{#1.}\xspace}

\def\IH{\mathsf{IH}}

\def\proc{process\xspace}

\def\procs{processes\xspace}

\def\fproc{\flat process\xspace}
\def\fprocs{\flat processes\xspace}
\def\eproc{endpoint process\xspace}
\def\eprocs{endpoint processes\xspace}
\def\subproc{sub-process\xspace}

\def\netcont{network context\xspace}
\def\nuparcont{$\lnewsymb \lpar$-context \xspace}

\def\rfree{race-free\xspace}
\def\rfreedom{race-freedom\xspace}

\def\controlName#1{#1}
\def\stuck{\controlName{stuck}\xspace}

\def\lfree{\controlName{deadlock-free}\xspace}

\def\lfreedom{\controlName{deadlock-freedom}\xspace}

\def\LFreedom{\controlName{Deadlock-Freedom}\xspace}

\def\dfree{\controlName{deadlock-free}\xspace}

\def\dfreedom{\controlName{deadlock-freedom}\xspace}

\def\progress{\controlName{progress}\xspace}

\makeatletter

\DeclareRobustCommand{\btleft}{\mathrel{\mathpalette\btlr@\blacktriangleleft}}
\DeclareRobustCommand{\btright}{\mathrel{\mathpalette\btlr@\blacktriangleright}}

\newcommand{\btlr@}[2]{%
  \begingroup
  \sbox\z@{$\m@th#1\triangleright$}%
  \sbox\tw@{\resizebox{1.1\wd\z@}{1.1\ht\z@}{\raisebox{\depth}{$\m@th#1\mkern-1mu#2$}}}%
  \ht\tw@=\ht\z@ \dp\tw@=\dp\z@ \wd\tw@=\wd\z@
  \copy\tw@
  \endgroup
}

\newcommand{\varSet}{\mathcal X}

\def\cneg#1{{#1}^\lbot}

\def\lprec{\btleft}

\def\unit{\circ}

\def\free{\mathsf{free}}
\def\freeof#1{\free(#1)}

\newcommand{\lsplus}[1][]{
	\mathop{
		\vphantom{\bigoplus}
		\mathchoice
		{\vcenter{\hbox{\resizebox{\widthof{$\displaystyle\bigoplus$}}{!}{$\boxplus$}}}}
		{\vcenter{\hbox{\resizebox{\widthof{$\bigoplus$}}{!}{$\boxplus$}}}}
		{\vcenter{\hbox{\resizebox{\widthof{$\scriptstyle\oplus$}}{!}{$\boxplus$}}}}
		{\vcenter{\hbox{\resizebox{\widthof{$\scriptscriptstyle\oplus$}}{!}{$\boxplus$}}}}
	}\displaylimits
}

\def\multioplus#1#2{\bigoplus_{#1}^{#2}}

\def\lsend#1#2{\langle#1\oc#2\rangle}
\def\lrecv#1#2{(#1\wn#2)}

\DeclareRobustCommand{\lnewsymb}{%
	\mathord{\text{\usefont{X2}{cmr}{m}{n}\symbol{200}}}%
}
\DeclareRobustCommand{\lwensymb}{%
	\mathord{\text{\usefont{X2}{cmr}{m}{n}\symbol{221}}}%
}
\DeclareRobustCommand{\lyasymb}{%
	\mathord{\text{\usefont{X2}{cmr}{m}{n}\symbol{223}}}%
}

\DeclareRobustCommand{\lqusymb}{%
	\mathord{\text{\raisebox{.1em}{\rotatebox[origin=c]{180}{$\mathsf{Q}$}}}}%
}

\def\lNu#1#2{\lnewsymb{#1}.#2}

\def\lYa#1#2{\lyasymb{#1}.#2}
\def\lFa#1#2{\forall{#1}.#2}
\def\lEx#1#2{\exists{#1}.#2}
\def\lQu#1#2{\lqusymb{#1}.#2}

\def\lNup#1#2{\lNu{#1}{\!\left( #2 \right)}}

\def\lExp#1#2{\lEx{#1}{\!\left( #2 \right)}}

\def\lNus#1#2{\widetilde{\lnewsymb{#1}}.{#2}}
\def\lYas#1#2{\widetilde{\lyasymb{#1}}.{#2}}

\newcommand{\biglbra}[1][]{\bigwith\limits_{#1}}
\newcommand{\biglsel}[1][]{\bigoplus\limits_{#1}}

\def\feq{\lleq}

\def\chole{\bullet}
\newcommand{\ctx}[1][\chole]{[#1]}
\newcommand{\Ctx}[1][]{\left[#1\right]}

\newcommand{\fof}[1]{\left[\!\left[ #1 \right]\!\right]}
\newcommand{\fofp}[1]{\left[\!\left[ #1 \right]\!\right]^\namesset}
\newcommand{\ffof}[1]{\left\lfloor\!\left\lfloor#1\right\rfloor\!\right\rfloor}
\newcommand{\fofi}[2][\xs]{\partial_{#1} \fof{#2}}

\newcommand{\chCof}[1][\cdot]{\chorpol{\mathsf{Chor}{\left( {\color{black} #1}\right)}}}

\def\alphaeq{\equiv_{\alpha}}

\def\namesset{\mathcal N}
\def\procset{\mathcal P}
\def\varset{\mathcal V}
\def\labelsset{\mathcal L}

\def\lab{\ell}

\def\pnu#1{(\nu #1)}
\newcommand{\pnus}[2][k]{\pnu {#2_1\dots #2_#1}}

\def\psend#1#2{#1!\langle#2\rangle.}
\def\precv#1#2{#1?(#2).}
\def\Apsend#1#2{#1!\langle#2\rangle}
\def\Aprecv#1#2{#1?(#2)}

\def\ppar{\;|\;}
\def\ppars{\ppar\cdots\ppar}
\def\pnil{\mathsf{Nil}}

\def\plsend#1#2#3{#1 \lseq \Set{#2}_{#3}}
\def\plrecv#1#2#3{#1 \lcoseq \Set{#2}_{#3}}

\def\Plsend#1#2#3#4{\plsend{#1}{#2:#3_#2}{#2\in #4}}
\def\Plrecv#1#2#3#4{\plrecv{#1}{#2:#3_#2}{#2\in #4}}

\newcommand{\named}[3][black]{{\color{#1}{\left(#2\right)}_{#3}}}

\def\fsubst#1#2{\left[#1/#2 \right]}

\def\boundnamesof#1{\mathcal{B}_{#1}}
\def\namesof#1{\mathcal N_{#1}}
\def\freenamesof#1{\mathcal{F}_{#1}}
\def\pnamesin#1{\mathsf{pn}(#1)}

\def\steq{\equiv}
\def\strew{\Rrightarrow}
\def\strews{\strew\!\!\!\!\!\strew}
\def\steqr{{\Lleftarrow\!\!\!\!\!\Rrightarrow}}

\def\cC{\mathcal C}
\def\cN{\mathcal N}
\def\cP{\mathcal P}
\def\cK{\mathcal K}
\def\xX{\mathcal X}

\def\merge{\sqcup}

\def\redsem{\to}
\def\redsems{\twoheadrightarrow}

\def\rscomr{\mathsf{Com}}
\def\rschoir{\mathsf{Choice}}
\def\rslabr{\mathsf{Label}}

\def\rsresr{\mathsf{Res}}
\def\rsparr{\mathsf{Par}}
\def\rsstrr{\mathsf{Struc}}
\def\rsstreqr{\mathsf{Struc}^\Rrightarrow}

\def\tcol#1{\rclr{#1}}
\def\tt{\tcol{\tau}}
\def\tm{\tcol{\mu}}

\def\tcom#1#2#3{\tcol{#1\tto#2:#3}}
\def\tbra#1#2{\tcol{#1:#2}}

\def\chorpol#1{\bclr{#1}}

\def\chor{choreography\xspace}
\def\chors{choreographies\xspace}
\def\Chors{Choreographies\xspace}
\def\proloc#1#2{#1\colon\!\!\!\colon #2}

\def\chC{\chorpol C}
\def\chS{\chorpol S}
\def\chCrf{\chorpol {C^{\mathsf{rf}}}}

\def\cnu#1{\chorpol{\pnu#1}}
\newcommand{\cnus}[2][k]{\chorpol{\pnu{#2_{1} \dots #2_{#1}}}}

\newcommand{\EPP}[2][]{\mathsf{EPP}_{#1}\left(#2\right)}

\def\merge{\sqcup}
\def\Merge{\bigsqcup}
\def\chOrd{\sqsupseteq}

\def\chnil{\chorpol{\boldsymbol 0}}

\def\chorfsubst#1#2#3{\left[#1/#2_{#3} \right]}

\newcommand{\pid}[1]{\mathsf{#1}}

\newcommand{\tto}{\mathbin{\boldsymbol{\rightarrow}}}
\def\pp{{\pid p}}

\def\pq{{\pid q}}
\def\pr{\pid r}

\newcommand{\comc}[5]{\chorpol{\pid{#1}.{#2} \tto\pid{#3}.{#4}:#5}}

\newcommand{\choic}[6]{
	\comc{#1}{#2}{#3}{#4}{#5}{\Set{
			\begin{array}{l}
				#6
			\end{array}
	}}
}

\newcommand{\specchoics}[7]{\chorpol{\pid{#1}.{#4}\tto\pid{#2}.{#4'} :#3
		\Set{
			\begin{array}{l|@{\;}l}
				#5: {#6} & #5\in #4
				\\
				#5: {#7} &#5\in #4'\setminus #4
			\end{array}
		}
}}

\newcommand{\Choic}[6][C]{ \chorpol{\pid{#2}.{\set{#5}}\tto\pid{#3}.{#6'} :#4
		\Set{
			\begin{array}{l|@{\;}l}
				\lab: {#1}_{\lab} 	&\lab\in #6
				\\
				\lab: {S}_{\lab} 	&\lab\in #6'\setminus #6
			\end{array}
		}
}}
\newcommand{\Choics}[6][C]{\chorpol{\pid{#2}.{#6}\tto\pid{#3}.{#6'} :#4
	\Set{
		\begin{array}{l|@{\;}l}
			#5: #1_{#5} & #5\in #6
			\\
			#5: {S}_{#5} &#5\in #6'\setminus #6
		\end{array}
	}
}}

\newcommand{\gencom}{\comc {p}x{q}yk;}

\newcommand{\genchoi}{\Choic{p}{q}{k}{\lab_i}{L}}
\newcommand{\genchois}{\Choics{p}{q}{k}{\lab}{L}}

\def\seqEval#1#2{e \downarrow_\ldia v}

\def\ccomr{\chorpol{\mathsf{Com}}}
\def\cchoir{\chorpol{\mathsf{Choice}}}
\def\clabr{\chorpol{\mathsf{Label}}}
\def\cdelir{\chorpol{\mathsf{D\mhyphen Com}}}
\def\cdelcr{\chorpol{\mathsf{D\mhyphen Choice}}}
\def\crestr{\chorpol{\mathsf{Rest}}}

\def\osto#1{\xrightarrow{#1}}

\newcommand{\proves}[1][]{\mathord{\vdash_{#1}\,}}
\def\dD{\mathcal D}

\def\BV{\mathsf{BV}}

\def\MLL{\mathsf{MLL}}
\def\MALL{\mathsf{MALL}}

\def\muMALL{\mu\MALL}

\mathchardef\mhyphen="2D

\def\wcut{\cup\set{\cutr}}
\def\wF{^{1}}

\def\XS{\mathsf{X}}

\newcommand{\sS}[1][\mathcal S]{\sclr{ #1 }}
\newcommand{\sdash}[1][\sS]{\sS[{#1}] \vdash}

\def\LL{\mathsf{LL}}
\def\iLL{\mathsf{iLL}}

\def\ChorL{\mathsf{ChorL}}

\def\PIsys{\mathsf{PiL}}

\def\PIL{\mathsf{PiL}}

\def\NL{\PIsys}

\def\init{\mathsf{Init}}

\def\comr{\mathsf{Com}}

\def\brar{\mathsf{Bra}}
\def\selr{\mathsf{Sel}}

\def\Ecomr{\mathsf{E}\mhyphen\comr}
\def\Elsendr{\mathsf{E}\mhyphen\mathsf{Choice}}
\def\Elrecvr{\mathsf{E}\mhyphen\mathsf{Label}}

\def\Caxr{\chorpol{C}\mhyphen\init}
\def\Crestr{\chorpol{C}\mhyphen\mathsf{flat}}
\def\Ccomr{\chorpol{C}\mhyphen\comr}
\def\Cchor{\chorpol{C}\mhyphen\selr}

\def\axrule{\mathsf {ax}}

\def\cutr{\mathsf {cut}}
\def\mixr{\mathsf{mix}}
\def\precur{\lprec^\unit}
\newcommand{\rrule}[1][]{\mathsf{r}^{#1}}
\def\urule{\unit}

\def\nurule{\lnewsymb}

\def\nuurule{\lnewsymb^\lunit}

\def\nqsrule{\lnewsymb\mhyphen\lyasymb}

\def\wSr{_{\mathcal{S}}}
\def\saxrule#1{\axrule\wSr}

\def\loadr#1{#1_\mathsf{load}}
\def\popr#1{#1_\mathsf{pop}}
\def\unitr#1{#1_{\lunit}}

\def\nuur{\unitr\lnewsymb}
\def\nuloadr{\loadr\lnewsymb}
\def\nupopr{\popr\lnewsymb}
\def\yaur{\unitr\lyasymb}
\def\yaloadr{\loadr\lyasymb}
\def\yapopr{\popr\lyasymb}

\def\isnusymb{\lnewsymb}
\def\isyasymb{\lyasymb}
\def\isnasymb{\nabla}

\def\isnu#1{\sclr{#1^{\isnusymb}}}
\def\isya#1{\sclr{#1^{\isyasymb}}}

\newcommand{\isna}[2][]{\sclr{#2^{\isnasymb_{#1}}}}

\def\pbras#1#2#3{\plsend{#1}{#2}{#3}}

\def\psels#1#2#3{\plrecv{#1}{#2}{#3}}

\usepackage{marginnote}

\newcommand{\multiprem}[2]{\vlhy{\left\{\vlderivation{#1}\right\}_{#2}}}